%% file: main-arxiv2021-iovpts.tex
\documentclass[10pt]{article}
\usepackage[letterpaper]{geometry}

\input{preamb}

\input{mycmds}

\usepackage{fullpage}

\usepackage{varioref}
\usepackage{color}

\newcommand{\avm}[1] {\textcolor{red}{#1}}

\begin{document}

\input{title-arxiv2021-iovpts}
\input{introduction}

\input{preliminaries}

\input{vpts}
\input{our-psuites}
\input{testing-iovpts}

\input{conclusions}

\bibliographystyle{plain}
\bibliography{main-arxiv2021-iovpts}

\input{apendice}

\end{document}

%% file: preamb.tex
\usepackage{amsmath}
\usepackage{amsthm}
\usepackage{mathtools,amssymb,latexsym} 
\usepackage{amsxtra} 
\usepackage[mathscr]{eucal}
\usepackage{turnstile}
\usepackage{upgreek}
\usepackage{pifont}
\usepackage[10pt]{moresize}

\usepackage{tikz}
\usetikzlibrary{shapes,shapes.multipart, calc,matrix,arrows,arrows,positioning,automata}
\tikzset{
    >=stealth',
    punkt/.style={
           circle,
           rounded corners,
           draw=black, thick, 
           text width=1.5em,
           minimum height=2em,
           text centered},
               punkts/.style={
                      circle,
                      rounded corners,
                      draw=black, thick, 
                      text width=1em,
                      minimum height=1em,
                      text centered},
    invisible/.style={
           draw=none,
           text width=1.5em,
           minimum height=0em,
           text centered},
    inv/.style={
           draw=none,
           text width=2.5em,
           minimum height=3em,
           text centered},
    pil/.style={
           ->,
           thick,
           shorten <=2pt,
           shorten >=2pt,}
}

\usepackage[ruled,titlenumbered,linesnumbered,boxed,vlined,noend,algo2e]{algorithm2e}
\usepackage{algorithm}
\usepackage{changepage}
\usepackage{ragged2e}

%% file: mycmds.tex
\newcommand{\ysum}{$+$}
\newcommand{\yminus}{$-$}

\newcommand{\ysi}{\sigma}

\newcommand{\ySis}{\Sigma^\star}
\newcommand{\yal}{\alpha}
\newcommand{\ybe}{\beta}
\newcommand{\yga}{\gamma}
\newcommand{\yGa}{\Gamma}
\newcommand{\yGas}{\Gamma^\star}

\newcommand{\yde}{\delta}
\newcommand{\yDe}{\Delta}

\newcommand{\yte}{\theta}
\newcommand{\ytau}{\varsigma}

\newcommand{\yemp}{\emptyset}
\newcommand{\ysse}{\subseteq}
\newcommand{\ypow}[1]{\mathcal{P}(#1)}
\newcommand{\yst}{\,|\,} 

\newcommand{\yvpts}[5]{\langle #1, #2, #3, #4, #5 \rangle}  
\newcommand{\yvptsS}{\yvpts{S}{S_{in}}{L}{\yGa}{T}}  
\newcommand{\yvptsQ}{\yvpts{Q}{Q_{in}}{L}{\yDe}{R}}  

\newcommand{\yltsconf}[1]{\yltsn{C}_{#1}}  

\newcommand{\yiovpts}[6]{\langle #1, #2, #3, #4,#5, #6\rangle}  
\newcommand{\yiovptsS}{\yiovpts{S}{S_{in}}{L_I}{L_U}{\yGa}{T}}  
\newcommand{\yiovptsQ}{\yiovpts{Q}{Q_{in}}{L_I}{L_U}{\yDe}{R}}  
\newcommand{\yiovptsU}{\yiovpts{Q}{Q_{in}}{L_U}{L_I}{\yDe}{R}}  

\newcommand{\yiovp}[2]{\yltsn{IOVP}(#1,#2)}  

\newcommand{\yait}{\sharp}  
\newcommand{\yvpa}[6]{\langle #1, #2, #3, #4, #5, #6\rangle}  
\newcommand{\yvpaA}{\yvpa{S}{S_{in}}{A}{\yGa}{\rho}{F}}  
\newcommand{\yvpaB}{\yvpa{Q}{Q_{in}}{A}{\yDe}{\mu}{G}}  

\newcommand{\yltsn}[1]{\mathscr{#1}}  
\newcommand{\yS}{\yltsn{S}}  
\newcommand{\yA}{\yltsn{A}}  
\newcommand{\yB}{\yltsn{B}}  
\newcommand{\yI}{\yltsn{I}}  
\newcommand{\yQ}{\yltsn{Q}}  
\newcommand{\yP}{\yltsn{P}}  
\newcommand{\yV}{\yltsn{V}}  

\newcommand{\ytr}[3]{#1\overset{#2}{\rightarrow}#3} 
\newcommand{\ytrt}[3]{#1\overset{#2}{\Rightarrow}#3} 

\newcommand{\ytrtf}[4]{#1\overset{#2}{\underset{#3}{\rightarrow}}#4} 

\newcommand{\ypdaconf}[1]{\yltsn{C}_{#1}}  

\newcommand{\ypdatrt}[3]{#1\underset{#2}{\mapsto}#3} 
\newcommand{\ypdatrtf}[4]{#1\overset{#2}{\underset{#3}{\mapsto}}#4} 

\newcommand{\ycfgtrtf}[4]{#1\overset{#2}{\underset{#3}{\leadsto}}#4} 

\newcommand{\ycfgtrtfl}[4]{#1\overset{#2}{\underset{#3}{\hookrightarrow}}#4} 
\newcommand{\yout}{\text{\,\bf out}} 

\newcommand{\yioltsn}[1]{\mathscr{#1}} 
\newcommand{\yT}{\yioltsn{T}}  
\newcommand{\yD}{\yioltsn{D}}  
\newcommand{\yF}{\yioltsn{F}}  
\newcommand{\yTS}{\yioltsn{TS}}  

\newcommand{\ytru}[4]{#1\overset{#2}{\underset{#3}{\rightarrow}}#4} 
\newcommand{\ytrut}[4]{#1\overset{#2}{\underset{#3}{\Rightarrow}}#4} 

\newcommand{\ypass}{\text{\bfseries\sffamily pass}} 
\newcommand{\yfail}{\text{\bfseries\sffamily fail}} 

\newcommand{\yconf}[2]{\,\text{\bf vconf}_{#1,#2}\,\,} 
\newcommand{\yiocolike}{\text{\,\,\bf ioco-like}\,\,} 
\newcommand{\yafter}{\text{\,\,\bf after}\,\,} 

\newtheorem{theo}{Theorem}[section] 
\newtheorem{lemm}[theo]{Lemma} 
\newtheorem{coro}[theo]{Corollary} 

\newtheorem{prop}[theo]{Proposition}
\newtheorem{defi}[theo]{Definition} 

\newtheorem{rema}[theo]{Remark}

\theoremstyle{definition}
\newtheorem{exam}[theo]{Example}

\newcommand{\yfun}[3]{#1:#2\rightarrow #3} 
\newcommand{\yoh}[1]{\mathcal{O}(#1)} 
\newcommand{\yeps}{\varepsilon} 
\newcommand{\yfim}{\hfill$\Box$} 
 
\newcommand{\ycomp}[1]{\overline{#1}}

\newcommand{\ypush}[1]{#1_+}
\newcommand{\ypop}[1]{#1_-}

%% file: title-arxiv2021-iovpts.tex

\title{Testing Pushdown Systems \protect\thanks{In collaboration with Computing Institute - University of Campinas.}}


\author{Adilson Luiz Bonifacio\thanks{Computing Department, University of Londrina, Londrina, Brazil.} \and
	Arnaldo Vieira Moura\thanks{Computing Institute, University of Campinas, Campinas, Brazil}}

\date{} 

\maketitle

\input{abstract}

%% file: abstract.tex
\begin{abstract}
Testing on reactive systems is a well-known laborious activity on software development due to their asynchronous interaction with the environment. 
In this setting model based testing has been employed when checking conformance and generating test suites of such systems using labeled transition system as a formalism as well as the classical {\bf ioco} conformance relation. 

In this work we turn to a more complex scenario where the target systems have an auxiliary memory, a stack. 
We then studied a more powerful model, the Visibly Pushdown Labeled Transition System (VPTS), its variant Input/Output VPTS (IOVPTS), its associated model Visibly Pushdown Automaton (VPA), and aspects of conformance testing and test suite generation. 
This scenario is much more challenge since the base model has a pushdown stack to capture more complex behaviors which commonly found on reactive systems.
We then defined a more general  conformance relation for pushdown reactive systems such that it prevents any observable implementation behavior that was not already present in the given specification. 
Further we gave an efficient algorithm to check conformance in this testing scenario and also showed that it runs in worst case asymptotic polynomial time in the size of both the given specification and the implementation that are put under test. 
\end{abstract}

%% file: introduction.tex
\section{Introduction}

Reactive systems have become increasingly common among computer systems,  whether  in a usual technological solution or in a critical application of the industry. 
We see everywhere real-world systems being ruled by reactive behaviors where the systems interact with an external environment by receiving input stimuli and producing outputs in response. 
Usually the development of such systems requires a precise and an automatic support during the process, specially in the testing activity because high costs in terms of resources and maintenance time can be generated when the test step is inappropriately performed.

Model-based testing is an important approach that has been employed to test critical and reactive systems because it offers guarantees to its correctness and also because properties, such as completeness of test suites, can be formally proven~\cite{Gargantini05CTES,sidhu89}. 
These aspects have been studied using appropriate formalisms that capture the behavior of reactive systems, where the exchange of input and output stimuli can occur asynchronously. 
Prominent among such formalism are Input/Output Labeled Transition Systems (IOLTSs)~\cite{tret-formal-1992,tre99a}.

Usually, the testing process aims to verify whether an implementation conforms to a given specification, both described using a common chosen formalism~\cite{tret-test-1996}. 
This conformance checking process depends on the specification formalism, on the kind of fault model used, and on a conformance relation that is to be verified~\cite{ananbcc-orchestrated-2013,tret-formal-1992,tre99a}.
For IOLTS models, the classical Input/Output Conformance ({\bf ioco}) is a well-studied relation~\cite{tret-model-2008,simap-generating-2014, 2019arXiv190210278B}.
A more recent work~\cite{bonifacio2020cleiej} proposed a more general approach, based on regular languages, to check {\bf ioco} conformance for IOLTS models. 

In this work we study aspects of conformance testing and test suite generation for a class of more powerful models, the so called Input/Output Visibly Pushdown Labeled Transition Systems (IOVPTSs), and their associated formalism, Visibly Pushdown Automata (VPAs). 
Using an auxiliary pushdown stack, this formalism can capture the behavior of more complex reactive systems.
Accordingly, we use a more general  conformance relation, but one still in the same spirit of the classical {\bf ioco} relation, in the sense that it prevents any observable implementation behavior that was not already present in the given specification. 
We also develop and prove the correctness of an efficient algorithm to verify conformance in this testing scenario. 
Further, we show that our algorithm runs in worst case asymptotic polynomial time in the size of both the given specification and the implementation that are put under test. 

We organize this paper as follows.
Section~\ref{sec:models} establishes some notations, definitions and preliminary results over VPAs.
Section~\ref{sec:conformance} defines (IO)VPTSs and relates these models to the associated VPAs.
Section~\ref{sec:cfl-suites} defines conformance relation based on IOVPTS models, and shows how to obtain complete test suites, of polynomial complexity, for this class of models. 
Section~\ref{sec:tretma-suites} gives an ioco-like conformance relation for IOVPTS models and the notion of test suite completeness, and then provides a  polynomial time algorithm to check ioco-like conformance for IOVPTS implementations. 
Section~\ref{sec:conclusion} offers some concluding remarks.

%% file: preliminaries.tex
\section{Notation and preliminary results}\label{sec:models}

In this section we establish some notation and define Visibly Pushdown Automata (VPA). 
Some preliminary results that will be useful later are also indicated.

\subsection{Basic Notation}\label{subsec:notation}
Let $X$ and $Y$ be sets. We indicate by $\ypow{X}=\{Z\yst Z\ysse X\}$ the power set of $X$, and $X-Y=\{z\yst \text{$z \in X$ and $z \not\in Y$}\}$ indicates set difference.
We will let $X_Y=X\cup Y$. When $Y=\{y\}$ is a singleton we may also write $X_y$ for $X_{\{y\}}$.
If $X$ is a finite set, the size of $X$ will be indicated by $|X|$.

An alphabet is any non-empty set of symbols.
Let $A$ be an alphabet.
A word over $A$ is any finite sequence $\ysi=x_1\ldots x_n$ of symbols in $A$, that is, $n\geq 0$ and $x_i\in A$, for all $i=1,2,\ldots, n$.
When $n=0$, $\ysi$ is the empty sequence, also indicated by $\yeps$.
The set of all finite sequences, or words, over $A$ is denoted by $A^{\star}$, and the set of all nonempty finite words over $A$ is indicated by $A^+$.
When we write $x_1x_2\ldots x_n\in A^{\star}$, it is implicitly assumed that $n\geq 0$ and that $x_i\in A$, $1\leq i\leq n$, unless explicitly noted otherwise.
The length of a word $\alpha$ over $A$ is indicated by $|\alpha|$.
Hence, $|\yeps|=0$.
Let $\ysi=\ysi_1\ldots \ysi_n$ and $\rho=\rho_1\ldots \rho_m $ be words over $A$. The concatenation of $\ysi$ and $\rho$, indicated by $\ysi\rho$, is the word $\ysi_1\ldots\ysi_n\rho_1\ldots\rho_m$.
Clearly, $|\ysi\rho|=|\ysi|+|\rho|$.
A language $G$ over $A$ is any set $G\ysse A^\star$ of words over $A$.
Let $G_1$, $G_2\ysse A^\star$ be languages over $A$. Their product, indicated by  $G_1G_2$, is the language 
$\{\ysi\rho\yst \ysi\in G_1, \rho\in G_2\}$. 
If $G\ysse A^\star$ is a language over $A$, then its complement is the language $\ycomp{G}=A^\star-G$.

We will also need the notion of a morphism between alphabets.
\begin{defi}\label{def:morph}
Let $A$, $B$ be alphabets.
A \emph{homomorphism}, or just a \emph{morphism}, from $A$ to $B$ is any function $h:A\rightarrow B^\star$.
\end{defi}

A morphism $h:A\rightarrow B^\star$ can be extended in a natural way to a function 
$\widehat{h}:A^\star\rightarrow B^\star$, thus
\begin{equation*}
\widehat{h}(\ysi) = 
\begin{cases}
\yeps \quad &\text{if $\ysi=\yeps$}\\
h(a)\widehat{h}(\rho) & \text{if $\ysi=a\rho$ with $a\in A$}.
\end{cases}
\end{equation*}
We can further lift $\widehat{h}$ to a function $\widetilde{h}: \ypow{A^\star}\rightarrow \ypow{B^\star}$ in a natural way, by letting $\widetilde{h}(G)=\bigcup\limits_{\ysi\in G}\widehat{h}(\ysi)$,
for all $G\ysse A^\star$. 
In order to avoid cluttering the notation, we may write $h$ instead of $\widehat{h}$, or of $\widetilde{h}$, when no confusion can arise.
When $a\in A$, we define the simple morphism $h_a:A\rightarrow A-\{a\}$ by letting $h_a(a)=\yeps$, and $h_a(x)=x$ when $x\neq a$.  
Hence, $h_a(\ysi)$ erases all occurrences of  $a$ in the word $\ysi$.
\input{vpa}

%% file: vpa.tex
\subsection{Visibly Pushdown Automata}\label{subsec:vpa}

A Visibly Pushdown Automaton (VPA)~\cite{alurm-visibly-2004} is, basically, a Pushdown Automaton (PDA)~\cite{hopcu-introduction-1979}, with a transition relation over an alphabet and a pushdown stack (or just a stack, for short) associated to it. 
Thus a VPA can make use of a potentially infinite memory as with PDA models. 
Any alphabet $L$ is always partitioned into three disjoint subsets $L=L_c\cup L_r\cup L_i$.
Elements in the set $L_c$ are ``call symbols'', or ``push symbols'', and specify push actions on the stack.
Elements in $L_r$ are ``return symbols'', or ``pop symbols'', and indicate pop actions, and in $L_i$ we find ``simple symbols'',  that do not change the stack%
\footnote{In~\cite{alurm-visibly-2004} symbols in $L_i$ are called ``internal action symbols''.
	In this text we will reserve that denomination for another special symbol as will be apparent later.}.%

The next definition is a slight extension of the similar notion appearing in~\cite{alurm-visibly-2004}. 
Here, we also allow $\yeps$-moves, that is, the VPA can change states without reading any symbol from the input.
\begin{defi}\label{def:pda}
A \emph{Visibly Pushdown Automaton} (VPA)~\cite{alurm-visibly-2004} over a finite input alphabet $A$ is a tuple $\yA = \yvpaA$, where:
\begin{itemize}
\item[---] $A=A_c\cup A_r\cup A_i$ and $A_c$, $A_r$, $A_i$ are pairwise disjoint; 
\item[---] $S$ is a finite set of \emph{states};
\item[---] $S_{in}\ysse S$ is set of \emph{initial states};
\item[---] $\Gamma$ is a finite \emph{stack alphabet}, with $\bot\not \in \yGa$ the \emph{initial stack symbol};
\item[---] The \emph{transition relation} is $\rho=\rho_c\cup\rho_r\cup\rho_i$, where $\rho_c\ysse S\times A_c \times \yGa  \times S$, $\rho_r\ysse S\times A_r \times \yGa_\bot  \times S$, and 
$\rho_i\ysse S\times (A_i\cup\{\yeps\}) \times \{\yait\} \times S$,  where $\yait\not\in\yGa_\bot$ is a place-holder symbol;
\item[---] $F\ysse S$ is the set of \emph{final states}. 
\end{itemize}
A transition $(p,\yeps,\yait,q)\in \rho_i$ is called an \emph{$\yeps$-move} of $\yA$. 
A \emph{configuration} of $\yA$ is any triple $(p,\ysi,\yal)\in S\times A^\star\times(\yGas\{\bot\})$, and the set of all configurations of $\yA$ it is indicated by $\ypdaconf{\yA}$.
\end{defi}

A tuple $(p,a,Z,q)\in \rho_c$ specifies a \emph{push-transition}.
We have a \emph{pop-transition} if $(p,a,Z,q)\in \rho_r$, and a \emph{simple-transition} if $(p,a,Z,q)\in \rho_i$.
The intended meaning for a push-transition $(p,a,Z,q)\in \rho_c$ is that $\yS$, in  state $p$ and reading input $a$,  changes to state $q$ and pushes $Z$ in the stack. 
A pop-transition $(p,a,Z,q)\in \rho_r$ makes $\yS$, in state $p$ and reading $a$,  pop $Z$ from the stack and change to state $q$. 
A simple-transition $(p,a,\yait,q)\in \rho_i$  has the intended meaning of reading $a$ and changing  from state $p$ to state $q$, no matter what the topmost stack symbol is. 
An $\yeps$-transition  $(p,\yeps,\yait,q)\in \rho_i$ indicates that $\yS$ must move from state $p$ to state  $q$, while reading no symbol from the input. 

We can now define the relation $\ypdatrt{}{}{}\ysse\ypdaconf{\yA}\times\ypdaconf{\yA}$ which captures simple moves of a VPA $\yA$.
Note that the model can make pop moves when the stack contains only the initial stack symbol $\bot$. 
The semantics of a VPA is the language comprised by all input strings it accepts.

\begin{defi}\label{def:fsa-move}
Let $\yA=\yvpaA$ be a VPA. 
For all $\ysi\in A^\star$:
\begin{enumerate}
	\item If $(p,a,Z,q)\in\rho_c$  then $\ypdatrt{(p,a\ysi,\yal\bot)}{}{(q,\ysi,Z\yal\bot)}$;
	\item If $(p,a,Z,q)\in\rho_r$ then $\ypdatrt{(p,a\ysi, Z\yal\bot)}{}{(q,\ysi,\yal\bot)}$ when $Z\neq \bot$, or  $\ypdatrt{(p,a\ysi, \bot)}{}{(q,\ysi,\bot)}$ when $Z=\bot$; 
  \item If $(p,a,\yait,q)\in\rho_i$  then $\ypdatrt{(p,a\ysi, \yal\bot)}{}{(q,\ysi, \yal\bot)}$. 
\end{enumerate}
The set $L(\yA)=\big\{ \ysi\in A^\star\yst \ypdatrtf{(s_0,\ysi,\bot)}{\star}{\yA}{(p,\yeps,\ybe)},\,\,s_0\in S_{in},\,\,p\in F\big\}$ 
is the \emph{language accepted} by $\yA$.
Two VPAs $\yA$ and $\yB$ are said to be \emph{equivalent} when $L(\yA)=L(\yB)$. 
\end{defi}
We may also write $\ypdatrt{(p,a\ysi,\yal)}{\yA}{(q,\ysi,\yga)}$ if it is important to explicitly mention  the VPA $\yA$.
It is clear that when $\ypdatrt{(p,a\ysi,\yal)}{}{(q,\mu,\ybe)}$ then we get  $(q,\mu,\ybe)\in \ypdaconf{\yA}$, so that $\mapsto$ is a relation on the set $\ypdaconf{\yA}$. 
The $n$-th power of the relation $\ypdatrt{}{}{}$ will be indicated by $\ypdatrtf{}{n}{}{}$ for all $n\geq 0$, and its reflexive and transitive closure will be indicated by $\ypdatrtf{}{\star}{}{}$.

A language is said to be a Visibly Pushdown Language (VPL) when it is accepted by a VPA.
\begin{defi}\label{def:llc}
Let $A$ be an alphabet and let $G\ysse A^\star$ be a language over $A$.
Then $G$ is a \emph{visibly pushdown language} if there is a VPA $\yA$  such that $L(\yA)=G$.
\end{defi}

Given any VPA, we can always construct and equivalent VPA with no $\yeps$-moves and the same number of states as the original VPA. 
\begin{prop}\label{prop:no-eps}
For any  VPA  we can effectively construct an equivalent VPA with no $\yeps$-moves and with the same number of states.
\end{prop}
\begin{proof}
We first install in the resulting VPA $\yB$ all transitions of $\yA$ that are not $\yeps$ transitions.
Then, the main idea is to find sequences of $\yeps$-moves between pairs of states in $\yA$. If there is such a sequence of $\yeps$-moves from $s$ to $r$, and another such sequence from $p$ to $q$, then for every transition of $\yA$ having $r$ as the start state and $p$ as the target state,  we define  a similar transition from $s$ to $q$ in $\yB$.  
A complete proof can be seen in Appendix~\vref{app:prop:no-eps}. 
\end{proof}

Next we define deterministic VPAs.
Determinism captures the idea that there is at most one computation for a given input string.
Since our notion of a VPA extends the original notion in~\cite{alurm-visibly-2004}, we also have to deal with $\yeps$-moves.
\begin{defi}\label{def:vpa-determinism}
Let $\yA=\yvpaA$ be a VPA. 
We say that $\yA$ is \emph{deterministic} if $\vert S_{in}\vert\le 1$,  and the following conditions hold: 
\begin{enumerate}
\item $(p,x,Z_i,q_i)\in\rho_c$ with $i=1,2$, implies $Z_1=Z_2$ and $q_1=q_2$; 
\item $(p,x,Z,q_i)\in\rho_r\cup \rho_i$ with $i=1,2$,  implies $q_1=q_2$; 
\item  $(p,x,Z,q_1)\in\rho$ with $x\neq \yeps$, implies $(p,\yeps,\yait,q_2)\not\in\rho$ for all $q_2\in S$.
\end{enumerate}
A language $L$ is a \emph{deterministic VPL} if $L=L(\yA)$ for some deterministic VPA $\yA$.
\end{defi}

It is worth noticing that Definition~\ref{def:vpa-determinism} does not prohibit $\yeps$-moves in deterministic VPAs.
With deterministic VPAs, however,  computations are always  unique, as expected.
\begin{prop}\label{prop:vpa-determ}
Let $\yA=\yvpaA$ be a deterministic VPA, let $(p,\ysi,\yal)\in\ypdaconf{\yA}$, and take $n\geq 0$.
Then $\ypdatrtf{(p,\ysi,\yal)}{n}{}{(q_i,\mu_i,\yga_i)}$, $i=1,2$, implies 
$(q_1,\mu_1,\yga_1)=(q_2,\mu_2,\yga_2)$.
\end{prop}
\begin{proof}
By induction on $n\geq 0$. When $n=0$ the result is immediate.
Take $n>0$. 
Using the induction hypothesis we can write $\ypdatrtf{(p,\ysi,\yal)}{n-1}{}{(q,\mu,\yga)}\ypdatrtf{}{1}{}{(q_i,\mu_i,\yga_i)}$, where the last single steps 
used the transitions $(q,x_i,Z_i,q_i)\in\rho$, $i=1,2$.

First assume that $x_1=\yeps$. Then Definition~\ref{def:pda} gives $(q,x_1,Z_1,q_1)=(q,\yeps,\yait,q_1)\in\rho_i$.
From $\ypdatrtf{(q,\mu,\yga)}{}{}{(q_1,\mu_1,\yga_1)}$ and Definition~\ref{def:fsa-move}, we conclude that $\mu_1=\mu$ and $\yga=\yga_1$.
Since $\yA$ is deterministic, Definition~\ref{def:vpa-determinism}(3) forces
$x_2=\yeps=x_1$, and then Definition~\ref{def:pda} gives $(q,x_2,Z_2,q_2)=(q,\yeps,\yait,q_2)\in \rho_i$.
So, from $\ypdatrtf{(q,\mu,\yga)}{}{}{(q_2,\mu_2,\yga_2)}$ we now get $\mu=\mu_2$ and $\yga=\yga_2$.
Since $\yA$ is deterministic, using Definition~\ref{def:vpa-determinism}(2) we get $q_1=q_2$, and then we have $(q_1,\mu_1,\yga_1)=(q_2,\mu_2,\yga_2)$.

Likewise when $x_2=\yeps$.
We can now assume  $x_1$, $x_2\in A$.
From $\ypdatrtf{(q,\mu,\yga)}{}{}{(q_i,\mu_i,\yga_i)}$ and Definition~\ref{def:fsa-move}  we get  $x_1\mu_1=\mu=x_2\mu_2$.
Hence, $x_1=x_2=x$, and $\mu_1=\mu_2=\yde$.
It suffices to verify that $q_1=q_2$ and $\yga_1=\yga_2$. 
We have $\ypdatrtf{(q,x\yde,\yga)}{}{}{(q_i,\yde,\yga_i)}$ using the transitions $(q,x,Z_i,q_i)\in\rho$. 
There are three cases:
\begin{description}
\item[{\sc Case 1:}] $x\in A_c$.
Since $\yA$ is deterministic and $(q,x,Z_i,q_i)\in\rho_c$,  Definition~\ref{def:vpa-determinism}(1) guarantees that $Z_1=Z_2$ and $q_1=q_2$.
From $\ypdatrtf{(q,x\yde,\yga)}{}{}{(q_i,\yde,\yga_i)}$  and Definition~\ref{def:fsa-move}(1) we get 
$\yga_1=Z_1\yga$ and $\yga_2=Z_2\yga$, so that $\yga_1=\yga_2$, and we are done.

\item[{\sc Case 2:}] $x\in A_i$.
From Definition~\ref{def:pda} we get $(q,x_i,Z_i,q_i)=(q,x,\yait,q_i)\in\rho$, $i=1,2$.
Then, Definition~\ref{def:vpa-determinism}(2) forces $q_1=q_2$.
Since  $\ypdatrtf{(q,x\yde,\yga)}{}{}{(q_i,\yde,\yga_i)}$ Definition~\ref{def:fsa-move}(3) implies
$\yga_1=\yga=\yga_2$, concluding this case.

\item[{\sc Case 3:}] $x\in A_r$. 
First, assume $\yga=\bot$.
Then $\ypdatrtf{(q,x\yde,\bot)}{}{}{(q_i,\yde,\yga_i)}$ using the transitions $(q,x,Z_i,q_i)\in\rho_r$.
Definition~\ref{def:fsa-move}(2) gives $\yga_1=\bot=\yga_2$ and $Z_1=\bot=Z_2$.
So, $(q,x,\bot,q_1)$, $(q,x,\bot,q_2)\in\rho_r$, and Definition~\ref{def:vpa-determinism}(2) implies $q_1=q_2$.
Now, assume $\yga\neq \bot$.
Definition~\ref{def:fsa-move}(2) now gives $\yga=Z_1\yga_1$ and $\yga=Z_2\yga_2$, so that $Z_1=Z_2$ and $\yga_1=\yga_2$.
Once again  Definition~\ref{def:vpa-determinism}(2) implies $q_1=q_2$, and we conclude this case. 
\end{description}
The proof is complete. 
\end{proof}

The next result shows that we can remove $\yeps$-moves, while still preserving determinism.
\begin{prop}\label{prop:no-eps-determ}
Given a deterministic VPA $\yA$, we can obtain an equivalent deterministic VPA $\yB$ with no $\yeps$-moves and with the same number of states. 
\end{prop}
\begin{proof}
The argument resembles the proof of Proposition~\ref{prop:no-eps}, 
but now we have to maintain the determinism at every step.
The main strategy is, first, to remove $\yeps$-cycles in $\yA$, and then proceed to remove any acyclic sequence of $\yeps$ moves that can still be found in $\yA$.   
A complete proof is given in Appendix~\vref{app:prop:no-eps-determ}. 
\end{proof}

%

\subsection{The Synchronous Product of VPAs}\label{subsec:vpa-product}

The product of two VPAs captures their synchronous behavior. 
It will be useful when testing conformance between pushdown memory models. 
\begin{defi}\label{def:productVPA}
Let $\yS=\yvpaA$  and $\yQ=\yvpa{Q}{Q_{in}}{A}{\yDe}{\mu}{G}$ be  VPAs with  a common alphabet. 
Their product is the VPA
$\yS\times \yQ=\yvpa{S\times Q}{S_{in}\times Q_{in}}{A}{\yGa\times\yDe}{\nu}{F\times G}$, where $((s_1,q_1),a,Z,(s_2,q_2))\in\nu$ if and only if either:
\begin{enumerate}
	\item  $a\neq\yeps$ and we have $(s_1,a,Z_1,s_2)\in \rho$ and  $(q_1,a,Z_2,q_2)\in \mu$ with $Z=(Z_1,Z_2)$ or $Z=Z_1=Z_2\in\{\bot,\yait\}$; or 
	\item $a=\yeps$, $Z=\yait$ and we have either $s_1=s_2$ and $(q_1,\yeps,\yait,q_2)\in\mu$; or $q_1=q_2$ and  $(s_1,\yeps,\yait,s_2)\in\rho$. 
\end{enumerate}
\end{defi}

We first note that the product construction preservers both determinism and the absence of $\yeps$-moves when the original VPAs satisfy such conditions.
\begin{prop}\label{prop:epsilon-deterministic}
Let $\yS=\yvpaA$  and $\yQ=\yvpaB$ be VPAs and let $\yP=\yS\times\yQ$ be their product as in Definition~\ref{def:productVPA}.
Then the following conditions hold:
\begin{enumerate}
\item If $\yS$ and $\yQ$ have no $\yeps$-moves, then $\yP$ has no $\yeps$-moves;
\item If $\yS$ and $\yQ$ are deterministic with no $\yeps$-moves, then $\yP$ is also deterministic with no $\yeps$-moves.
\end{enumerate}
\end{prop}
\begin{proof}
Write $\yP=\yS\times \yQ=\yvpa{S\times Q}{S_{in}\times Q_{in}}{A}{\Gamma \times \yDe}{\nu}{F\times G}$.

When $\yS$ and $\yQ$ have no $\yeps$-moves, only item (1) of Definition~\ref{def:productVPA} applies, so that $\yP$ has no $\yeps$-moves too.
Hence, assertion (1) holds.

It remains to show that $\yP$ is deterministic when $\yS$ and $\yQ$ are also deterministic and have no $\yeps$-moves.
We just argued that $\yP$ has no $\yeps$-moves, so that $\yP$ can not violate condition (3) of Definition~\ref{def:vpa-determinism}.
For the sake of contradiction, assume $\yP$ has transitions 
$((s,q),a,Z_i,(p_i,r_i))\in\nu$, where $Z_i\in\yGa\times \yDe\cup\{\bot,\yait\}$, $i=1,2$, and $a\neq \yeps$.
Using Definition~\ref{def:productVPA} (1), from $((s,q),a,Z_1,(p_1,r_1))$ we get  
\begin{align}
&(s,a,X_1,p_1)\in\rho,\,\,\, (q,a,Y_1,r_1)\in\mu\label{prop2.26a}
\end{align}
where $X_1\in \yGa\cup{\{\bot,\yait\}}$, $Y_1\in \yDe\cup{\{\bot,\yait\}}$.
Likewise, from $((s,q),a,Z_2,(p_2,r_2))$ we get 
\begin{align}
&(s,a,X_2,p_2)\in\rho,\,\,\, (q,a,Y_2,r_2)\in\mu\label{prop2.26b}
\end{align}
where $X_2\in \yGa\cup{\{\bot,\yait\}}$, $Y_2\in \yDe\cup{\{\bot,\yait\}}$.

If $a\in A_c$, we must have $Z_1,Z_2\in \yGa\times \yDe$, and Definition~\ref{def:productVPA} (1) forces $Z_i=(X_i,Y_i)$, $i=1,2$.
Since $a\in A_c$, the determinism of $\yS$, together with Definition~\ref{def:vpa-determinism} (1) applied to Eqs. (\ref{prop2.26a}, \ref{prop2.26b}) gives $X_1=X_2$, $p_1=p_2$.
Likewise, the determinism of $\yQ$ gives $Y_1=Y_2$, $r_1=r_2$.
But then $Z_1=(X_1,Y_1)=(X_2,Y_2)=Z_2$ and $(p_1,r_1)=(p_2,r_2)$.
We conclude that when $a\in A_c$, $\yP$ cannot violate condition (1) of Definition~\ref{def:vpa-determinism}.

Now let $a\in A_r\cup A_i$ and $Z_1=Z_2$.
If $Z_1=(X,Y)\in \yGa\times \yDe$, Eq. (\ref{prop2.26a}) and Definition~\ref{def:productVPA} (1) imply $X_1=X$ and $Y_1=Y$.
Since $Z_2=Z_1$, Eq. (\ref{prop2.26b}) gives $X_2=X$ and $Y_2=Y$.
Hence, $X_1=X_2$, $Y_1=Y_2$.
Now, the determinism of $\yS$ applied to Eqs. (\ref{prop2.26a}, \ref{prop2.26b}), together with Definition~\ref{def:vpa-determinism} (2) implies $p_1=p_2$.
Likewise, $r_1=r_2$.
Thus, $(p_1,r_1)=(p_2,r_2)$ and we conclude that $\yP$ does not violate Definition~\ref{def:vpa-determinism}(2).
Finally let $Z_1\in \{\bot,\yait\}$.
Because $a\neq \yeps$, Definition~\ref{def:productVPA} (1) and Eq. (\ref{prop2.26a}) say that
$Z_1=X_1=Y_1$.
Likewise, now with Eq. (\ref{prop2.26b}) and knowing that $Z_1=Z_2$, we get and $Z_1=Z_2=X_2= Y_2$.
Hence, with $a\in A_r\cup A_i$, using Definition~\ref{def:vpa-determinism} (2) applied to Eqs. (\ref{prop2.26a}, \ref{prop2.26b}), the determinism of $\yS$ and $\yQ$  implies  $p_1=r_1$  and $p_2=r_2$.
Again we have $(p_1,r_1)=(p_2,r_2)$ and we conclude that $\yP$ can not violate Definition~\ref{def:vpa-determinism} in any circumstance.
That is, $\yP$ is deterministic.
\end{proof}
 
It is not hard to see that when $\yS$ and $\yQ$ are deterministic and $\yeps$-moves are allowed in both of them, then the product $\yP$ need not be deterministic.
The following result links moves in the product VPA to moves in the original constituent VPAs.
\begin{prop}\label{prop:product-behavior}
Let $\yS=\yvpaA$  and $\yQ=\yvpa{Q}{Q_{in}}{A}{\yDe}{\mu}{G}$ be VPAs over the same input alphabet. 
Then,  for all $\ysi\in A^\star$, and all $i$, $k\geq 0$:
\begin{quote}
$$\ypdatrtf{((s,q),\ysi,(X_1,Y_1)\ldots (X_i,Y_i)\bot)}{\star}{\yS\times \yQ}{((p,r),\yeps,(Z_1,W_1)\ldots (Z_k,W_k)\bot)}$$ 
\begin{center}if and only if\end{center} 
$$\ypdatrtf{(s,\ysi,X_1\ldots X_i\bot)}{\star}{\yS}{(p,\yeps,Z_1\ldots Z_k\bot)}\,\,\text{and}\,\, \ypdatrtf{(q,\ysi,Y_1\ldots Y_i\bot)}{\star}{\yQ}{(r,\yeps,W_1\ldots W_k\bot)}.$$
\end{quote}
\end{prop}
\begin{proof}
Since movements in VPAs are determined by the input string, except for possible $\yeps$-moves, this result is expected.
The argument, going from the product $\yP$ to the constituent VPAs $\yS$ and $\yQ$, uses a simple induction in the number of steps in the run over $\yP$.
In the other direction, we have to allow for any of $\yS$ or $\yQ$ to make independent $\yeps$-moves.
In this case, we induct on the total number of steps that occur in both the runs over $\yS$ and $\yQ$. 
Appendix~\vref{app:prop:product-behavior} has the details. 
\end{proof}


\subsection{VPLs and closure properties}\label{subsec:properties}

We look at some closure properties involving VPLs.
Similar results appeared elsewhere~\cite{alurm-visibly-2004}, but here the VPA models are somewhat more general because they allow for $\yeps$-moves,
which can make a difference in the results.
Moreover, it will prove important to investigate if determinism, when present in the participant VPAs, can also be guaranteed for the resulting VPAs.
Further, since later on we will be analyzing the complexity of certain constructions, we also note how the sizes of the resulting models vary 
as a function of the size of the given models.

First we report on a simple result on the stack size during runs of VPAs. 
\begin{prop}\label{prop:stack-size}
	Let $\yA=\yvpaA$ and $\yB=\yvpaB$ be VPAs over a common alphabet $A$.
	Consider starting configurations $(s_1,\ysi,\yal_1\bot)$ and $(q_1,\ysi,\ybe_1\bot)$ of $\yA$ and $\yB$, respectively,
	and where $\vert \yal_1\vert=\vert\ybe_1\vert$.
	If $\ypdatrtf{(s_1,\ysi,\yal_1\bot)}{\star}{\yA}{(s_2,\yeps,\yal_2\bot)}$ and $\ypdatrtf{(q_1,\ysi,\ybe_1\bot)}{\star}{\yB}{(q_2,\yeps,\ybe_2\bot)}$, then we must have $\vert \yal_2\vert=\vert\ybe_2\vert$.
\end{prop}
\begin{proof}
A simple induction on $n=\vert\ysi\vert$. 
For details see Appendix~\ref{app:prop:stack-size}, at page~\pageref{app:prop:stack-size}.
\end{proof}

Using the product construction, we can show that VPLs are also closed under intersection.
\begin{prop}\label{prop:cap-vpa}\label{prop:determ-complement}
Let $\yS=\yvpaA$ and $\yQ=\yvpa{Q}{Q_{in}}{A}{\yDe}{\mu}{G}$ be VPAs with $n$ and $m$ states, respectively. 
Then $L(\yS)\cap L(\yQ)$ can be accepted by a VPA $\yP$ with $mn$ states.
Moreover, if $\yS$ and $\yQ$ are deterministic, then $\yP$ is also deterministic.
\end{prop}
\begin{proof}
Let $\yP=\yS\times \yQ=\yvpa{S\times Q}{S_{in}\times Q_{in}}{A}{\Gamma \times \yDe}{\nu}{F\times G}$
be the product of $\yS$ and $\yQ$.
See Definition~\ref{def:productVPA}.
It is clear that $\yP$ has $nm$ states.

For the language equivalence, assume that $\ysi\in L(\yP)$, so that 
$\ypdatrtf{((s,q),\ysi,\bot)}{\star}{\yP}{((p,r),\yeps,\yga\bot)}$ with $(s,q)\in S_{in}\times Q_{in}$ and 
$(p,r)\in F\times G$, for some $\yga\in (\yGa\times\yDe)^\star$.
Using Proposition~\ref{prop:product-behavior} we get $\ypdatrtf{(s,\ysi,\bot)}{\star}{\yS}{(p,\yeps,\yal\bot)}$ 
for some $\yal\in\yGa^\star$. 
Since $s\in S_{in}$ and $p\in F$, we conclude that $\ysi\in\ L(\yS)$.
Likewise, $\ysi\in\ L(\yQ)$, so that $\ysi\in\ L(\yS)\cap L(\yQ)$.
For the converse, assume that $s\in S_{in}$, $p\in F$ and $\ypdatrtf{(s,\ysi,\bot)}{\star}{\yS}{(p,\yeps,\yal\bot)}$ 
for some $\yal\in\yGa^\star$. Likewise, let $q\in Q_{in}$, $r\in G$ and $\ypdatrtf{(q,\ysi,\bot)}{\star}{\yQ}{(r,\yeps,\ybe\bot)}$ 
for some $\ybe\in\yDe^\star$. 
From Proposition~\ref{prop:stack-size} we get $\vert \yal\vert=\vert \ybe\vert$.
We can now apply Proposition~\ref{prop:product-behavior} and write 
$\ypdatrtf{((s,q),\ysi,\bot)}{\star}{\yP}{((p,r),\yeps,\yga\bot)}$ for some $\yga\in (\yGa\times\yDe)^\star$. 
Since $(p,r)\in F\times G$ we get $\ysi\in L(\yP)$, and the equivalence holds.

Finally, Proposition~\ref{prop:epsilon-deterministic} guarantees that $\yP$ is deterministic when $\yS$ and $\yQ$ are deterministic. 
\end{proof}


Now, we investigate the closure of VPLs under union. 
But first, in order to complete the argument for the union, we need to consider VPAs that can always read any string of input symbols when started at any state and with any stack configuration.
\begin{defi}\label{def:forward} 
Let $\yA=\yvpaA$ be a VPA.
We say that $\yA$ is a \emph{non-blocking} VPA if, for all $s\in S$, all $\ysi\in A^\star$ and all $\yal\in\yGa^\star$, there are $p\in S$ and $\ybe\in\yGa^\star$ such that $\ypdatrtf{(s,\ysi,\yal\bot)}{\star}{\yA}{(p,\yeps,\ybe\bot)}$.
\end{defi}

Any VPA can be easily turned into a non-blocking VPA, with almost no cost in the number of states.
\begin{prop}\label{prop:non-blocking}
Let $\yS=\yvpaA$ be a VPA with $n$ states. 
Then we can construct an equivalent VPA $\yQ$ with at most $n+1$ states and which is also a non-blocking VPA.
Moreover, $\yQ$ is deterministic if $\yS$ is deterministic, and $\yQ$ has no $\yeps$-moves if so does $\yS$.
\end{prop}
\begin{proof}
Let $\yQ=\yvpa{S\cup\{p\}}{S_{in}}{A}{\yGa}{\mu}{F}$ where $p\not\in S$ is a new state.
In order to construct the new transition set $\mu$ as an extension of $\rho$, first pick some stack symbol $Z\in\yGa$.
Then, for all $s\in S$ such that there is no $\yeps$-transition out of $s$, that is $(s,\yeps,\yait,q)$ is not in $\rho$ for any $q\in S$, we proceed as follows.
For any input symbol $a\in A$:
\begin{enumerate}
\item $a\in A_i$: if $(s,a,\yait,r)\not\in\rho$ for all $r\in S$, add the
transition $(s,a,\yait,p)$ to $\mu$;
\item $a\in A_c$: if $(s,a,W,r)\not\in\rho$ for all $W\in\yGa$ and all $r\in S$, add the
transition $(s,a,Z,p)$ to $\mu$;
\item $a\in A_r$: if $(s,a,W,r)\not\in\rho$ for some $W\in\yGa_\bot$ and all $r\in S$, add the
transition $(s,a,W,p)$ to $\mu$.
\end{enumerate}
Finally, add to $\mu$ the self-loops $(p,a,\yait,p)$ for all $a\in A_i$, $(p,a,Z,p)$ for all $a\in A_c$, and 
$(p,a,W,p)$ for all $a\in A_r$ and all $W\in\yGa_\bot$. 

It is clear now that $\yS$ has $n+1$ states.
Moreover,  for all $s\in S$ and all $a\in A$ the construction readily allows for a move $\ypdatrtf{(s,a,\yal\bot)}{}{}{(r,\yeps,\ybe\bot)}$ for all $\yal\in \yGa^\star$.
Hence, an easy induction on $\vert \ysi\vert\geq 0$ shows that for all $s\in S$ and all $\yal\in\yGa^\star$ there will always be a computation $\ypdatrtf{(s,\ysi,\yal\bot)}{\star}{\yS}{(r,\yeps,\ybe\bot)}$, for some $r\in S$ and some $\ybe\in\yGa^\star$.
That is, the modified version is a non-blocking VPA.

The construction adds no $\yeps$-moves, so it is  clear that $\yQ$ has no $\yeps$-moves when we start  with a VPA $\yS$ that already has  no $\yeps$-moves.

Assume that $\yS$ is deterministic.
Since the construction adds no $\yeps$-moves, the new VPA $\yQ$ does not violate condition (3) of Definition~\ref{def:vpa-determinism}.
Also, none of the self-loops added at the new state $p$ violate any of the conditions of Definition~\ref{def:vpa-determinism}.
When $a\in A_c$ a new transition $(s,a,Z,p)$ is only added to $\yQ$ when there are no transition $(s,a,X,r)$ already in $\yS$, for any $X\in \yGa$, $r\in S$. 
Hence, $\yQ$ does not violate condition (1) of Definition~\ref{def:vpa-determinism}.
Likewise, we only add $(s,a,\yait,p)$, or $(s,a,W,p)$, to $\mu$ when we find no $(s,a,\yait,r)$, respectively we find no $(s,a,W,r)$, in $\rho$ for any $r$ in $S$. 
Hence, $\yQ$ does not violate condition (2) of Definition~\ref{def:vpa-determinism}. 
We conclude that $\yQ$ is deterministic when $\yS$ is already deterministic.

If  $\ypdatrtf{(s_0,\ysi,\bot)}{}{\yS}{(f,\yeps,\yal\bot)}$, with $s_0\in S_{in}$, $f\in F$, $\yal\in\yGas$ then we also have $\ypdatrtf{(s_0,\ysi,\bot)}{}{\yQ}{(f,\yeps,\yal\bot)}$, because all transitions in $\rho$ are also in $\mu$.
Hence, $L(\yS)\ysse L(\yQ)$.
For the converse, assume $\ypdatrtf{(s_0,\ysi,\bot)}{}{\yQ}{(f,\yeps,\yal\bot)}$, with $s_0\in S_{in}$, $f\in F$, $\yal\in\yGas$.
Note that all new transitions added to $\mu$ have the new state $p$ as a target state.
Thus, since the new state $p$ is not in $S_{in}$ nor in $F$, we see that all transitions used in this run over $\yQ$ are also in $\rho$, and
we then get $\ypdatrtf{(s_0,\ysi,\bot)}{}{\yS}{(f,\yeps,\yal\bot)}$.
Thus we also have $L(\yQ)\ysse L(\yS)$, showing that $L(\yS)= L(\yQ)$.
\end{proof}
  

Now the closure of VPLs under union is at hand.
\begin{prop}\label{prop:cup-vpa}
Let $\yS$ and $\yQ$ be two VPAs over an alphabet $A$, with $n$ and $m$ states, respectively. 
Then, we can construct a non-blocking VPA $\yP$ over $A$ with at most $(n+1)(m+1)$ states and such that $L(\yP)=L(\yS)\cup L(\yQ)$. 
Moreover, if $\yS$ and $\yQ$ are deterministic, then $\yP$ is also deterministic and has no $\yeps$-moves.
\end{prop}
\begin{proof}
Let $\yS=\yvpaA$ and $\yQ=\yvpaB$.
Using Proposition~\ref{prop:non-blocking} we can assume that $\yS$ and $\yQ$ are non-blocking VPAs with $n+1$ and $m+1$ states, respectively.

Let $\yP$ be the product of $\yS$ and $\yQ$ as in Definition~\ref{def:productVPA}, except that we redefine the final states of $\yP$ as $(F\times Q)\cup (S\times G)$.
Clearly, $\yP$ has $(n+1)(m+1)$ states.

We now argue that $\yP$ is also a non-blocking VPA.
Let $\ysi\in A^\star$, $(s,q)\in S\times Q$, and let $\yga=(Z_1,W_1)\ldots(Z_k,W_k)\in (\yGa\times\yDe)^\star$.
Since $\yS$ is a non-blocking VPA, we get $n\geq 0$, $p\in S$, $\yal\in\yGa^\star$ such that  
$\ypdatrtf{(s,\ysi,Z_1\ldots Z_k\bot)}{n}{\yS}{(p,\yeps,\yal\bot)}$. 
Likewise, $\ypdatrtf{(q,\ysi,W_1\ldots W_k\bot)}{m}{\yQ}{(r,\yeps,\ybe\bot)}$, for some
$m\geq 0$, $r\in Q$, $\ybe\in\yDe^\star$.
Applying Proposition~\ref{prop:stack-size} we get $\vert \yal\vert=\vert\ybe\vert$, and then using  Proposition~\ref{prop:product-behavior} we have 
$\ypdatrtf{((s,q),\ysi,\yga\bot)}{\star}{\yP}{((p,r),\yeps,\yde\bot)}$ where $\yga=(Z_1,W_1)\ldots (Z_k,W_k)$ and  $\yde\in (\yGa\times\yDe)^\star$.
This shows that $\yP$ is a non-blocking VPA.

Now suppose that $\ysi\in L(\yP)$, that is $\ypdatrtf{((s_0,q_0),\ysi,\bot)}{\star}{\yP}{((p,r),\yeps,\yal\bot)}$, where $(s_0,q_0)\in S_{in}\times Q_{in}$,  $(p,r)\in (F\times Q)\cup(S\times G)$, and $\yal\in(\yGa\times\yDe)^\star$.
Take the case when $(p,r)\in (F\times Q)$.
We get $p\in F$ and $s_0\in S_{in}$. 
Using Proposition~\ref{prop:product-behavior} we can also write $\ypdatrtf{(s_0,\ysi,\bot)}{\star}{\yS}{(p,\yeps,\ybe\bot)}$, for some $\ybe\in \yGa^\star$.
This shows that $\ysi\in L(\yS)$.
By a similar reasoning, when $(p,r)\in (S\times G)$ we get $\ysi\in L(\yQ)$.
Thus, $L(\yP)\ysse L(\yS)\cup L(\yQ)$. 

Now let $\ysi\in  L(\yS)\cup L(\yQ)$.
Take the case $\ysi\in L(\yS)$, the case $\ysi\in L(\yQ)$ being similar.
Then we must  have $\ypdatrtf{(s_0,\ysi,\bot)}{n}{\yS}{(p,\yeps,\yal\bot)}$ for some $n\geq 0$, some $\yal\in \yGa^\star$, and some $p\in F$.
Pick any $q_0\in Q_{in}$.
Since $\yQ$ is a non-blocking VPA, Definition~\ref{def:forward} gives some $r\in Q$ and some $\ybe\in \yDe^\star$ such that $\ypdatrtf{(q_0,\ysi,\bot)}{m}{\yQ}{(r,\yeps,\ybe\bot)}$ for some $m\geq 0$, some $r\in Q$ and some $\ybe\in\yDe^\star$. 
Using Proposition~\ref{prop:stack-size} we conclude that $\vert \yal\vert=\vert\ybe\vert$.
The only-if part of Proposition~\ref{prop:product-behavior} now yields  
$\ypdatrtf{((s_0,q_0),\ysi,\bot)}{\star}{\yP}{((p,r),\yeps,\yga\bot)}$, where $\yga\in (\yGa\times\yDe)^\star$.
Clearly, $(s_0,q_0)\in S_{in}\times Q_{in}$ is an initial state of $\yP$ and $(p,r)\in  (F\times Q)\cup (S\times G)$ is a final state of $\yP$.
Hence $\ysi\in L(\yP)$. 
Thus, $ L(\yS)\cup L(\yQ)\ysse L(\yP)$, and we then have $ L(\yS)\cup L(\yQ)=L(\yP)$.

Applying Proposition~\ref{prop:epsilon-deterministic} (2) we see that when $\yS$ and $\yQ$ are deterministic, then $\yP$ is also deterministic and has no $\yeps$-moves.

The proof is now complete. 
\end{proof}


We can also show that deterministic VPLs are closed under complementation.
A related result appeared in~\cite{alurm-visibly-2004}, but here we also allow for arbitrary $\yeps$-moves in any model.
\begin{prop}\label{prop:compl-vpa}
Let $\yS=\yvpaA$ be a deterministic VPA  with $n$ states. 
Then, we can construct a non-blocking and deterministic VPA $\yQ$ over $A$ with no $\yeps$-moves,  $n + 1$ states, and such that $L(\yQ)=\ycomp{L(\yS)} = \ySis - L(\yS)$.
\end{prop}
\begin{proof}
Applying Propositions~\ref{prop:no-eps-determ} and~\ref{prop:non-blocking}, we can assume that $\yS$ is a non-blocking and deterministic VPA with $n+1$ states and no $\yeps$-moves.

Let $\yQ=\yvpa{S}{S_{in}}{A}{\yGa}{\rho}{S-F}$, that is, we switch the final states of $\yS$.
Clearly, since $\yS$ and $\yQ$ have the same set of initial states and the same transition relation, 
we see that $\yQ$ is also a non-blocking and deterministic VPA with $n+1$ states and no $\yeps$-moves.

The proof can be concluded with a simple argument to show that $L(\yQ)=\ycomp{L(\yS)} = \ySis - L(\yS)$. 
See Appendix~\ref{app:prop:compl-vpa}, at page~\pageref{app:prop:compl-vpa} for details.
\end{proof}

The next closure under a simple concatenation will prove useful in the sequel.  
\begin{prop}\avm{}\label{prop:conct-vpa} 
	Let $\yS=\yvpaA$ be a VPA with $n$ states, and let $B\ysse A$. 
	Then we can construct a non-blocking VPA $\yQ$ with at most $2n+2$ states and no $\yeps$-moves, and such that $L(\yQ)= L(\yS)B$.
\end{prop}
\begin{proof}
Using Proposition~\ref{prop:non-blocking}	we can assume that $\yS$ is non-blocking and has no $\yeps$-moves.
 
The main idea is to add new states $\hat{s}$ for each state $s$ of $\yS$, and make these new states the only final states in $\yQ$.
Next, we add transitions $(s,b,Z,\hat{r})$ to $\yQ$, where $s$ is a final state in $\yS$, and $b$ is a symbol in $B$. 
Since symbols of $B$ can also occur in strings accepted by $\yS$, we have to be careful about transitions out of the new states $\hat{s}$.
In particular, if $(t,a,Z,r)$ is a transition of $\yS$ with $a\in B$ and $t\in F$ we add $(\hat{t},a,Z,\hat{r})$ to $\yQ$, otherwise we add 
$(\hat{t},a,Z,r)$ to $\yQ$.

A detailed construction is given in Appendix~\ref{app:prop:conct-vpa}, at page~\pageref{app:prop:conct-vpa}. 
\end{proof}

%% file: vpts.tex
\section{Reactive Pushdown  Models}\label{sec:conformance}

In this section we start with the Visibly Pushdown Labeled Transition System (VPTS) and state the relationship with its associated VPA. 
Next, we discuss the notion of contracted VPTSs and then introduce the  variation of Input/Output VPTSs. 

\subsection{Visibly Pushdown Labeled Transition Systems}\label{subsec:lts}\label{subsec:pda}

A Visibly Pushdown Labeled Transition System (VPTS) extends the classical notion of Labeled Transition System (LTS),  a formal model that is convenient to express asynchronous exchange of messages between participating entities, in the sense that outputs do not have to occur synchronously with inputs, but are generated as separated events. 
Any LTS has only a finite memory, represented by its set of states. 
A VPTS, on the other hand, has a stack memory associated to it, and thus can make use of a potentially infinite memory. 

The definition of a VPTS is inspired by the close notion of a Visibly Pushdown Automaton~\cite{alurm-visibly-2004}. 
As will be apparent,  any VPTS naturally induces a VPA. 
In particular, internal transitions in the VPTS will correspond to $\yeps$-moves in the associated VPA.  
\begin{defi}\label{def:plts}
	A \emph{Visibly Pushdown Labeled Transition System} (VPTS)  over an input alphabet $L$ is a tuple $\yS=\yvptsS$, where:
	\begin{itemize}
		\item[---] $S$ is a finite set of \emph{states} or \emph{locations};
		\item[---] $S_{in}\ysse S$ is the set of \emph{initial states}; 
		\item[---] 
		There  is a special symbol $\ytau\notin L$, the \emph{internal action symbol}; 
		\item[---] $\Gamma$ is a  set of \emph{stack symbols}.
		There is a special symbol $\bot\not \in \Gamma$, the \emph{bottom-of-stack symbol}; 
		\item[---] 
		$T=T_c \cup T_r \cup T_i$, where $T_c\ysse S\times L_c \times \yGa  \times S$, $T_r\ysse S\times L_r \times \yGa_\bot  \times S$, and $T_i\ysse S\times (L_i\cup\{\ytau\}) \times \{\yait\} \times S$,  where $\yait\not\in\yGa_\bot$ is a place-holder symbol.
	\end{itemize}
\end{defi}
Let $t=(p,x,Z,q)\in T$.
If $t\in T_c$ we say that it is a \emph{push-transition}.
Its intended meaning is that $\yS$, in  state $p\in S$ and reading input $x$,  changes to state $q$ and pushes $Z$ onto the stack. 
When $t\in T_r$ we have a \emph{pop-transition}, with  the intended meaning that $\yS$, in  state $p\in S$ and reading input $x\in L_r$,  pops $Z$ from the top of the stack and changes to state $q$.
Further, when the stack is reduced to the bottom of stack symbol, $\bot$, then a pop move can be taken, leaving the stack unchanged.  
We have a \emph{simple-transition} when $t\in T_i$ and $x\in L_i$, and we have an \emph{internal-transition} when $t\in T_i$ and $x=\ytau$. 
The meaning of a simple-transition $t$ is to change from state $p$ to state $q$, while reading $x$ from the input and leaving the stack unchanged. An internal-transition also changes from state $p$ to state $q$ leaving the stack unchanged, but does not read any symbols from the input.

In order to make these notions precise, we define the set of configurations and the elementary moves of a VPTS. 
\begin{defi}\label{def:simplemove}
	Let $\yS=\yvptsS$ be a VPTS.
	A \emph{configuration} of $\yS$ is a pair $(p,\yal)\in S\times (\yGas\{\bot\})$. 
	When $p\in S_{in}$ and $\yal=\bot$, we say that  $(p,\yal)$ is an \emph{initial configuration} of $\yS$.
	The set of all configurations of $\yS$ is indicated by $\yltsconf{\yS}$.
	Let $(q,\yal)\in\yltsconf{\yS}$, and let $\ell \in L_\ytau$. 
	Then we write $\ytr{(p,\yal)}{\ell}{(q,\ybe)}$ if there is a transition $(p,\ell,Z,q)\in T$, and either:
	\begin{enumerate}
		\item $\ell \in L_c$, and $\ybe=Z\yal$; 
		\item  $\ell \in L_r$, and  either (i) $Z\neq \bot$ and $\yal=Z\ybe$, or (ii) $Z=\yal=\ybe=\bot$; 
		\item $\ell \in L_i \cup \{\ytau\}$ and $\yal=\ybe$. 
	\end{enumerate}
	We call $\ytr{(p,\yal)}{\ell}{(q,\ybe)}$ an \emph{elementary move} of $\yS$, and we say that $(p,\ell,Z,q)\in T$ is the transition \emph{used} in the move  $\ytr{(p,\yal)}{\ell}{(q,\ybe)}$.
\end{defi} 
It is clear from the definition that after any elementary move $\ytr{(p,\yal)}{\ell}{(q,\ybe)}$ we have  $(q,\ybe)\in\yltsconf{\yS}$, that is, $(q,\ybe)$ is also a configuration of $\yS$.
Moreover, it also clear that we always have $\yal=\yal'\bot$ and $\ybe=\ybe'\bot$ with $\yal'$, $\ybe'\in \yGa^\star$,

\begin{rema}\label{rem:figure}
	In figures depicting VPTSs, a push-transition $(s,x,Z,q)$ will be graphically represented by {\rm $x/\ypush{Z}$} next to the corresponding arc from $s$ to $q$. 
	Similarly, the label {\rm $x/\ypop{Z}$} next to an arc from $s$ to $q$ will indicate a pop-transition $(s,x,Z,q)$.
A simple- or internal-transition over $(s,x,\yait,q)$ will be indicate by the label {\rm $x$} next to the corresponding arc. 	
\end{rema}

\begin{exam}
	Figure~\ref{vpts1} represents a VPTS $\yS$ where the set of states is $S=\{s_0,s_1\}$, $S_{in}=\{s_0\}$.
	Also, we have $L_c=\{b\}$, $L_r=\{c,t\}$, and $L_i=\{\}$, and $\yGa=\{Z\}$.
	\input{figs/vpts1.tex}
	We have a push-transition $(s_0,b,Z,s_0)$, the pop-transitions  
	$(s_0,c,Z,s_1)$, $(s_0,t,Z,s_1)$, $(s_1,c,Z,s_1)$, $(s_1,t,Z,s_1)$, and the internal-transition $(s_1,\ytau,\yait,s_0)$. \yfim
\end{exam} 

The semantics of a VPTS is given by its traces, or behaviors. 
But first we need the notion of paths in VPTS models which are just chains of elementary moves.
\begin{defi}\label{def:path}
	Let $\yS=\yvptsS$ be a VPTS and let $(p,\yal),(q,\ybe)\in\yltsconf{\yS}$ . 
	\begin{enumerate}
		\item Let $\ysi=l_1,\ldots,l_n$ be a word in $L_\ytau^\star$. We say that $\ysi$ is a \emph{path} from  $(p,\yal)$ to $(q,\ybe)$ if there are configurations  $(r_i, \yal_i)\in\yltsconf{\yS}$ , $0\leq i\leq n$, such that $\ytr{(r_{i-1},\yal_{i-1})}{l_i}{(r_i,\yal_i)}$, $1\leq i\leq n$,  with $(r_0,\yal_0)=(p,\yal)$ and $(r_n,\yal_n)=(q,\ybe)$.
		\item Let $\ysi\in L^\star$. We say that $\ysi$ is an \emph{observable path}  from $(p,\yal)$ to $(q,\ybe)$ in $\yS$ if
		there is a path $\mu$ from $(p,\yal)$ to $(q,\ybe)$ in $\yS$ such that $\ysi=h_\ytau(\mu)$.
	\end{enumerate}
	In both cases we also say that the path starts at $(p,\yal)$ and ends at $(q,\ybe)$, and we say that the configuration $(q,\ybe)$ is \emph{reachable} from $(p,\yal)$.
	We also say that $(q,\ybe)$ is \emph{reachable in} $\yS$ if it is reachable from an initial configuration of $\yS$.
\end{defi}


Clearly, moves labeled by the internal symbol $\ytau$ can occur in a path.
An observable path is just a path from which $\ytau$-labels  were removed. 
If $\ysi$ is a path from $(p,\yal)$ to $(q,\ybe)$, this can also be indicated  by writing $\ytr{(p,\yal)}{\ysi}{(q,\ybe)}$.
When $\vert\ysi\vert=1$ this has exactly the same meaning as indicated in Definition~\ref{def:simplemove}, so that no confusion can arise with this notation.
We may also write $\ytr{(p,\yal)}{\ysi}{}$ to indicate that there is some $(q,\ybe)\in \yltsconf{\yS}$ such that  $\ytr{(p,\yal)}{\ysi}{(q,\ybe)}$; likewise, $\ytr{(p,\yal)}{}{(q,\ybe)}$ means that there is some $\ysi\in L_\ytau^*$ such that $\ytr{(p,\yal)}{\ysi}{(q,\ybe)}$. 
Also $\ytr{(p,\yal)}{}{}$ means $\ytr{(p,\yal)}{\ysi}{(q,\ybe)}$ for some $(q,\ybe)\in \yltsconf{\yS}$ and some $\ysi\in L_\ytau^*$. 
When $\ysi$ is an observable path from $(p,\yal)$ to $(q,\ybe)$ we may write $\ytrt{(p,\yal)}{\ysi}{(q,\ybe)}$, with similar
shorthand notation also carrying over to the $\ytrt{}{}{}$ relation.
When we want to emphasize that the underlying VPTS is $\yS$, we write $\ytru{(p,\yal)}{\ysi}{\yS}{(q,\ybe)}$, or $\ytrut{(p,\yal)}{\ysi}{\yS}{(q,\ybe)}$.

Paths starting at a given configuration $(p,\yal)$ are also called the traces of $(p,\yal)$, or the traces starting at $(p,\yal)$.
The semantics of a VPTS is related to traces starting at an initial configuration. 
\begin{defi}\avm{}\label{def:trace}
	Let $\yS=\yvptsS$ be a VPTS and let $(p,\yal)\in \yltsconf{\yS}$. 
	\begin{enumerate}
		\item The set of \emph{traces} of $(p,\yal)$ is   $tr(p,\yal)=\{\ysi\yst \ytr{(p,\yal)}{\ysi\,\,}{}\}$.
		The set of \emph{observable traces} of $(p,\yal)$ is $otr(p,\yal)= \{\ysi\yst \ytrt{(p,\yal)}{\ysi\,\,}{}\}$.
		\item The \emph{semantics} of $\yS$ is $\bigcup\limits_{q\in S_{in}}\!\!\!\!tr(q,\bot)$, and the \emph{observable semantics} of $\yS$ is 
$\bigcup\limits_{q\in S_{in}}\!\!\!\!otr(q,\bot)$.
	\end{enumerate}
\end{defi}
We will also indicate the semantics and, respectively, the observable semantics, of $\yS$ by
$tr(\yS)$ and $otr(\yS)$.

	In any VPTS $\yS=\yvptsS$, if $\ytrt{(s,\yal)}{}{(p,\ybe)}$ then we also have $\ytr{(s,\yal)}{}{(p,\ybe)}$ in $\yS$, for all $(s,\yal)$, $(p,\ybe)\in \yltsconf{\yS}$.
	Moreover, $otr(\yS)=h_\ytau(tr(\yS))$.
	When $\yS$ has no internal transitions we also have $otr(\yS)=tr(\yS)$.

We can also restrict the syntactic description of VPTS models somewhat,  without loosing any descriptive capability, by removing states that are not reachable from any initial state, since these states 
will play no role when considering any system behaviors.  
Moreover, we can also eliminate $\ytau$-labeled self-loops.
We formalize these observations in the following remark.
\begin{rema}\label{rema:lte-finite}
	Let $\yS=\yvptsS$ be a VPTS.
	For any $s\in S$ we postulate that  $\ytr{(s_0,\bot)}{\ysi}{(s,\yal\bot)}$, for some $\yal\in\yGas$,
	$\ysi\in L_\ytau^\star$,  and $s_0\in S_{in}$. 
	Also, if $(s,\ytau,\yait,q)\in T$ then $s\neq q$.
\end{rema}

The intended meaning for  $\ytau$-moves in VPTSs is similar to that given for LTS models~\cite{bonifacio2020cleiej}, that is,
a VPTS can autonomously move along $\ytau$-transitions, without consuming any input symbol.
However, in some situations such moves may not be desirable, or  simply we might want no observable behavior leading to two distinct states.
This motivates the notion of determinism in VPTS models.

\begin{defi}\label{def:vpts-determinism}
	Let $\yS=\yvptsS$ be a VPTS. 
We say that $\yS$ is \emph{deterministic} if, for all   $s$, $p\in S_{in}$, $s_1$, $s_2\in S$, $\ybe_1,\ybe_2 \in \yGas$, and $\ysi\in L^\star$,  
we have that 	$\ytrut{(s,\bot)}{\ysi}{}{(s_1,\ybe_1\bot)}$ and $\ytrut{(p,\bot)}{\ysi}{}{(s_2,\ybe_2\bot)}$ imply	$s_1=s_2$ and $\ybe_1=\ybe_2$.
\end{defi}
As a consequence, deterministic VPTSs do not have internal moves.
\begin{prop}\label{prop:vpts-deterministic}
	Let $\yS=\yvptsS$ be a deterministic VPTS.
	Then $\yS$ has no $\ytau$-labeled transitions.
\end{prop}
\begin{proof}
	By contradiction, assume that $(s,\ytau,\yait,q)\in T$. 
	From Remark~\ref{rema:lte-finite} we get $s\neq q$ and we also get $\yal\in\yGas$, $\ysi\in\ L^\star$ such that $\ytr{(s_0,\bot)}{\ysi}{(s,\yal\bot)}$, with $s_0\in S_{in}$.
	Hence, $\ytr{(s_0,\bot)}{\ysi}{(s,\yal\bot)}\ytr{}{\ytau}{(q,\yal\bot)}$. 
	Using Definition~\ref{def:trace} we get $\ytrt{(s_0,\bot)}{\mu}{(s,\yal\bot)}$ and $\ytrt{(s_0,\bot)}{\mu}{(q,\yal\bot)}$, where $\mu=h_\ytau(\ysi)$. 
	Since $s\neq q$, this contradicts Definition~\ref{def:vpts-determinism}. 
\end{proof}


\subsection{Contracted VPTSs}\label{subsec:contracted}

It will also be useful, later on, to count on a guarantee that every transition in a VPTS can be exercised by some trace of the model.
Since every transition, except possibly for pop transitions, can always be taken, we concentrate on the pop transitions.
\begin{defi}\label{def:vpts-reduced}
	We say that a VPTS $\yS=\yvpts{S}{S_{in}}{L}{\yGa}{T}$ is \emph{contracted} if for every transition $(p,b,Z,r) \in T$ with $b\in L_r$, there are some $s_0\in S_{in}$, $\yal\in\yGa^\star$ and $\ysi\in L ^\star$ such that $\ytrt{(s_0,\bot)}{\ysi}{(p,\yal\bot)}$, where either (i) $\yal=Z\ybe$ for some $\ybe\in\yGa^\star$, or (ii) $\yal=\yeps$ and $Z=\bot$.
\end{defi}

We  can obtain contracted VPTSs using the next Proposition~\ref{prop:contracted-vpts}. 
The idea is to construct a context free grammar (CFG) based on the given VPTS, in such a way that the CFG generates strings where terminals represent VPTS transitions.
The productions of the CFG will indicate the set of transitions that can be effectively used in a trace over the VPTS.
\begin{prop}\avm{}\label{prop:contracted-vpts} 
	Let $\yS=\yvpts{S}{S_{in}}{L}{\yGa}{T}$ be a VPTS.
	We can effectively construct a contracted VPTS $\yQ=\yvpts{Q}{Q_{in}}{L}{\yGa}{R}$ with $\vert Q\vert\leq \vert S\vert$, 
and such that $tr(\yS)=tr(\yQ)$.
	Moreover, if $\yS$ is deterministic, then $\yQ$ is also deterministic.
\end{prop}
\begin{proof}
	First we construct a context-free grammar $G$ whose terminals represent transitions of $\yS$. 
Non-terminals are of the form $[s,Z,p]$ where $s, p\in S$ are states of $\yS$ and $Z\in\yGa_\bot$ is a stack symbol. 
The main idea  can be grasped as follows.
Let $t_i=[s_i,a_i,Z_i,p_i]$, $1\leq i\leq n$ be transitions of $\yS$ and let $\ysi=a_1a_2\cdots a_n$ be an input string.
If  $G$ has a leftmost derivation
	$$\ycfgtrtfl{[s_0,\bot,-]}{\star}{}{t_1\cdots t_n[r_1,W_1,r_2][r_2,W_2,r_3]\cdots[r_m,W_m,r_{m+1}][r_{m+1},\bot,-]}$$
	it must be the case that $\yS$, starting at the initial configuration $(s_0,\bot)$, can move along the transitions $t_1, \ldots, t_n$., in that order, to reach the configuration $(r_1,W_1W_2\cdots W_m\bot)$.
	That is, $\ytr{(s_0,\bot)}{\ysi}{(r_1,W_1W_2\cdots W_m\bot)}$, where $\ysi=a_1a_2\cdots a_n$.
And vice-versa.
	We then show that leftmost derivations of $G$ faithfully simulate traces of $\yS$ and, conversely, that any trace of $\yS$ can be simulated by a leftmost derivation of $G$.
That done, we can easily extract from $G$ a contracted VPTS $\yQ$.
A simple argument then proves that $L(\yQ)=tr(\yS)$.

The complete construction and detailed proofs can be found in Appendix~\vref{app:prop:contracted-vpts}. 	
\end{proof}


\subsection{Relating VPTS and VPA models}\label{subsec:vpts-vpa}

Now we show that any VPTS $\yS$  gives rise to an associated VPA $\yS_\yA$ in a natural way.
We convert any $\ytau$-transition of $\yS$  into a $\yeps$-transition of $\yS_\yA$.
The set of final states of $\yS_\yA$ is just the set of all locations of $\yS$.
Conversely, we can associate a VPA to any given VPTS, provided that all states in the given VPA are final states.
\begin{defi}\label{def:plts-pda}
	We have the following two associations:
	\begin{enumerate}
		\item Let $\yS=\yvptsS$ be a VPTS. 
		The VPA \emph{induced} by $\yS$ is  $\yA_\yS=\yvpa{S}{S_{in}}{L}{\yGa}{\rho}{S}$ where, for all $p$, $q\in S$,  $Z\in \yGa$,  $\ell\in L$, we have: 
		\begin{enumerate}
			\item $(p,\ell,Z,q)\in \rho$ if and only if $(p,\ell,Z,q)\in T$;
			\item $(p,\yeps,\yait,q)\in \rho$ if and only if $(p,\ytau,\yait,q)\in T$.
		\end{enumerate}
		\item Let $\yA=\yvpa{S}{S_{in}}{L}{\yGa}{\rho}{S}$ be a VPA.
		The VPTS \emph{induced} by $\yA$ is $\yS_\yA=\yvpts{S}{S_{in}}{L}{\yGa}{T}$ where:
		\begin{enumerate}
			\item $(p,\ell,Z,q)\in T$ if and only if $(p,\ell,Z,q)\in \rho$;
			\item $(p,\ytau,\yait,q)\in T$ if and only if $(p,\yeps,\yait,q)\in \rho$.
		\end{enumerate}
	\end{enumerate}
\end{defi}
The relationship between the associated models is given by the following result.
\begin{prop}	\label{prop:lang-vpa-vlpts}
	$\yS=\yvptsS$ be a VPTS and  $\yA=\yvpa{S}{S_{in}}{L}{\yGa}{\rho}{S}$ a VPA.
	Assume that either $\yA$ is the VPA induced $\yS$, or $\yS$ is the VPTS induced by $\yA$.
	Then, the following are equivalent, where $\ysi\in L^\star$, $\mu\in (L_\ytau)^\star$,  $s,p\in S$, $\yal,\ybe \in\yGa^\star$, $n\geq 0$:
	\begin{enumerate}
		\item $\ytrut{(s,\yal\bot)}{\ysi}{\yS}{(p,\ybe\bot)}$
		\item  $\ytrtf{(s,\yal\bot)}{\mu}{\yS}{(p,\ybe\bot)}$, $h_\ytau(\mu)=\ysi$
		\item $\ypdatrtf{(s,\ysi,\yal\bot)}{n}{\yA}{(p,\yeps,\ybe\bot)}$,  $n=\vert\mu\vert$
	\end{enumerate}
\end{prop}
\begin{proof}
	From Definition~\ref{def:path} we see that (1) and (2) are equivalent, for any VPTS $\yS$.
	
	If $\yA$ is the VPA induced by $\yS$ then, using Definition~\ref{def:plts-pda}(1), an easy induction on $\vert\mu\vert\geq 0$ shows that if (2) holds then (3) also holds with $n=\vert\mu\vert$.
	Likewise, an easy induction on $n\geq 0$ shows that if (3) holds, then we get some $\mu$ satisfying (2) and with $\vert\mu\vert=n$.	
	If $\yS$ is the VPTS induced by the VPA $\yA$ then the reasoning is very similar, now using  Definition~\ref{def:plts-pda}(2). 
\end{proof}

The observable semantics of $\mathcal{S}$ is just the language accepted by $A_\mathcal{S}$.
We note this as the next proposition.
\begin{prop}	\label{prop:plts-pda}
	Let $\yS$ be a VPTS and  $\yA_{\yS}$ the VPA induced by $\yS$.
	Then $otr(\yS)=L(\yA_{\yS})$ and, further, if $\yS$ is deterministic and contracted
	then $\yA_\yS$ is also deterministic.
	Conversely, let $\yA$ be a VPA and $\yS_\yA$ the VPTS induced by $\yA$.
	Then $L(\yA)=otr(\yS_\yA)$
	and, also, 
	if $\yA$ is deterministic and has no $\yeps$-moves, then $\yS_\yA$ is deterministic.
\end{prop}
\begin{proof}
	Let $\yS=\yvptsS$ and  $\yA_\yS=\yvpa{S}{S_{in}}{L}{\yGa}{\rho}{S}$.
	Then, $otr(\yS)=L(\yA_\yS)$ follows directly from Proposition~\ref{prop:lang-vpa-vlpts}.
		
	Now assume that $\yS$ is deterministic and contracted. 
	We show that $\yA_\yS$  satisfies Definition~\ref{def:vpa-determinism}.
	First, take $p$, $q\in S_{in}$. 
	Then, we have $\ytrt{(p,\bot)}{\yeps}{(p,\bot)}$ and $\ytrt{(q,\bot)}{\yeps}{(q,\bot)}$ in $\yS$.
	Since $\yS$ is deterministic, Definition~\ref{def:vpts-determinism} gives $p=q$.
	This shows that $\vert S_{in}\vert\leq 1$, as required by Definition~\ref{def:vpa-determinism}.
	
	Next, we look at the three assertions at Definition~\ref{def:vpa-determinism}.
	First, let $(p,\ell,Z_i,q_i)\in\rho_c$,  $i=1,2$.
	Since $\ell\in L_c$, Definition~\ref{def:plts-pda} gives $(p,\ell,Z_i,q_i)\in T$ for $i=1,2$. From Remark~\ref{rema:lte-finite} we obtain some $\ysi\in L^\star$ and $s_0\in S_{in}$ such that $\ytrut{(s_0,\bot)}{\ysi}{\yS}{(p,\yal\bot)}$ for some $\yal\in\yGa^\star$. 
	Hence,
	$\ytrut{(s_0,\bot)}{\ysi\ell}{\yS}{(q_i,Z_i\yal\bot)}$ for $i=1,2$. 
	Since $\yS$ is deterministic, Definition~\ref{def:vpts-determinism} gives $q_1=q_2$ and $Z_1=Z_2$.
	We conclude that $A_\yS$ satisfies condition 1 at Definition~\ref{def:vpa-determinism}. 	
	Now suppose we have $(p,\ell,Z,q_i)\in\rho_r\cup \rho_i$, $i=1,2$.
	Since $\ell\in L_r\cup L_i$, Definition~\ref{def:plts-pda} gives $(p,\ell,Z,q_i)\in T$, $i=1,2$. 
	When $\ell\in L_i$ we get $Z=\yait$, and proceed exactly as in the first case.
	When $\ell\in L_r$, since $\yS$ is contracted, Proposition~\ref{prop:contracted-vpts}  gives  $s_0\in S_{in}$, $\ysi\in L^\star$ and $\yal\in\yGa^\star$ such that  $\ytrut{(s_0,\bot)}{\ysi}{\yS}{(p,\yal\bot)}$.
	Further, $\yal= Z\ybe$ when $Z\neq \bot$, or $\yal=\yeps$ when $Z=\bot$.
	In the first case, $\ytrut{(s_0,\bot)}{\ysi\ell}{\yS}{(q_i,\ybe\bot)}$ and, in the second case,  $\ytrut{(s_0,\bot)}{\ysi\ell}{\yS}{(q_i,\bot)}$, $i=1,2$.
	Thus, because $\yS$ is deterministic, Definition~\ref{def:vpts-determinism} forces
	gives $q_1=q_2$.
	We conclude that $A_\yS$ satisfies condition 2 at Definition~\ref{def:vpa-determinism}. 
	Finally suppose we have $(p,x,Z,q_1)\in\rho$ and $(p,\yeps,Z,q_2)\in\rho$ with $x\neq \yeps$.
	Then, Definition~\ref{def:plts-pda} gives $(p,\ytau,\yait,q_2)\in T$
	and $(p,x,Z,q_1)\in T$. 
	Since $\yS$ is deterministic, Proposition~\ref{prop:vpts-deterministic} says that it has no $\ytau$-transitions, and we reached a contradiction.
	Hence,  $A_\yS$ satisfies condition 3 at Definition~\ref{def:vpa-determinism}, and the proof is complete. 
	
	For the converse, let $\yA=\yvpa{S}{S_{in}}{L}{\yGa}{\rho}{S}$ and $\yS_\yA=\yvptsS$. 
	Again, $otr(\yA)=otr(\yS_\yA)$ is immediate from Proposition~\ref{prop:lang-vpa-vlpts}.
	Now assume that $\yA$ is deterministic and has no $\yeps$-moves.
	Clearly, from the construction, we know that $\yS_\yA$ has no $\ytau$-moves. 
	Let $\ytrut{(s,\bot)}{\ysi}{\yS_\yA}{(s_1,\ybe_1\bot)}$ and 
	$\ytrut{(p,\bot)}{\ysi}{\yS_\yA}{(s_2,\ybe_2\bot)}$ with $p,s\in S_{in}$ and $\ysi\in L^\star$.
	Since $\yA$ is deterministic, we get $\vert S_{in}\vert\leq 1$, so that $p=s$.
	Using Definition~\ref{def:path}, and since $\yS_\yA$ has no $\ytau$-moves, we get 
	$\ytrtf{(s,\bot)}{\ysi}{\yS_\yA}{(s_1,\ybe_1\bot)}$ and 
	$\ytrtf{(s,\bot)}{\ysi}{\yS_\yA}{(s_2,\ybe_2\bot)}$.
	From Proposition~\ref{prop:lang-vpa-vlpts} we get 
	$\ypdatrtf{(s,\ysi,\bot)}{n}{\yA}{(s_1,\yeps,\ybe_1\bot)}$ and 
	$\ypdatrtf{(s,\ysi,\bot)}{n}{\yA}{(s_2,\yeps,\ybe_2\bot)}$, where $n=\vert \ysi\vert$.
	Now, since $\yA$ is deterministic, we can use Proposition~\ref{prop:vpa-determ} and conclude that $s_1=s_2$ and $\ybe_1=\ybe_2$.
	This shows that $\yS_\yA$ satisfies  Definition~\ref{def:vpts-determinism}, completing the proof.
\end{proof}
Lemma~\ref{prop:plts-pda} also says that $otr(\yS)$ is a visibly pushdown language, and that
for any given VPTS $\yS$, we can easily construct a VPA $\yA$ with $L(\yA)=otr(\yS)$.

\subsection{Input Output Pushdown Transition Systems}\label{subsec:ioplts}

The VPTS formalism can be used to model systems with a potentially infinite memory and with a capacity to interact asynchronously with an external environment.
In such situations, we may want to treat some action labels as symbols that the VPTS ``receives'' from the environment, and some other action labels as symbols that the VPTS ``sends back'' to the environment. 
The next VPTS variation differentiates between input action symbols and output action symbols.
\begin{defi}\label{def:iolts}
	An Input/Output Visibly Pushdown Transition System (IOVPTS) over an alphabet $L$ is a tuple $\yI=\yiovptsS$, where
	\begin{itemize}
		\item $L_I$ is a finite set of \emph{input actions}, or \emph{input labels}; 
		\item $L_U$ is a finite set of \emph{output actions}, or \emph{output labels}; 
		\item $L_I\cap L_U=\emptyset$, and $L=L_I\cup L_U $ is the set of \emph{actions} or \emph{labels}; and
		\item $\yvptsS$ is an  \emph{underlying VPTS} over $L$, which is associated to $\yI$.
	\end{itemize}
\end{defi}

We 
denote the class of all IOVPTSs with input alphabet $L_I$ and output alphabet $L_U$ by $\yiovp{L_I}{L_U}$. 

\begin{rema}\label{rem:iovpts-vpts}
	In order to keep the number of definitions under control, we agree that in any reference to a notion based on IOVPTSs, and  that has not been explicitly defined at some point, we substitute the IOVPTS model by its underlying VPTS.
	As a case in point, if $\yS$ is any IOVPTS with $\yV$ as its underlying VPTS, then the VPA induced by $\yS$ is simply the VPA induced by $\yV$, according to Definition~\ref{def:plts-pda}. 
	Likewise for any formal assertion involving IOVPTS models.
	For example, using Proposition~\ref{prop:plts-pda} we can just say that $otr(\yS)=L(\yA_\yS)$ without any explicit mention to the underlying VPTS $\yV$.
\end{rema}

The semantics of an IOVPTS is just the set of its observable traces, that is, observable traces of its underlying VPTS.  
\begin{defi}\label{def:iolts-semantics}
	Let $\yI=\yiovptsS$ be an IOVPTS.
	The \emph{semantics} of $\yI$ is the set $otr(\yI)=otr(\yS_\yI)$,
	where $\yS_\yI$ is the underlying VPTS associated to $\yI$.
\end{defi}
Also, when referring an IOVPTS $\yI$, the notation $\underset{\yI}{\rightarrow}$ and $\underset{\yI}{\Rightarrow}$ are to be understood as $\underset{\yS}{\rightarrow}$ and $\underset{\yS}{\Rightarrow}$, respectively, where $\yS$ is the underlying VPTS associated to $\yI$. 

\begin{exam}
	Figure~\ref{iovpts1} represents an IOVPTS that describes a machine that dispenses drinks.
	Here we have $L_I=\{b\}$ and $L_U=\{c,t\}$. From the context we can see that $L_c=\{b\}$, $L_r=\{c,t\}$ and $L_i=\yemp$.
	\input{figs/iovpts1.tex}
	The start state is $s_0$.
Symbol $b$ stands for button an user can press when asking for a drink, namely a cup of coffee or a cup of tea, with corresponding buttons represented by the labels $c$ and $t$, respectively.
	The user can hit the $b$ button while the machine stays at state $s_0$.
	Each time the $b$ button is activated, the machine pushes the symbol $Z$ on the stack, so that the stack is used to count how many times the $b$ button was hit by the user.
	
	At any instant, after the user has activated the $b$ button at least once, the machine moves to state $s_1$ and starts dispensing either coffee or tea, indicated by the $c$ and $t$ buttons.
	It decrements the stack each time a drink is dispensed, so that it will never deliver more drinks than the user asked for.
	
	A move back to state $s_0$, over the internal label $\ytau$, may interrupt the delivery of drinks, so that the user can, possibly, receive less drinks than originally asked for.
	In this case, when the next user will operate  the machine with  residual number of $Z$ symbols in the stack he could, eventually collect more drinks than asked for.
	But the machine will never dispense more drinks than the total number of solicitations. 
	An alternative could be to use one more state $s_2$ to interrupt the transition from $s_1$ to $s_0$ and install a self-loop at $s_2$ that empties the stack.  
	A more complex and complete version of a drink dispensing machine is illustrated in Subsection~\ref{subsec:drink}. \yfim
\end{exam}
We register one more example which will be used later.
\begin{exam}
Figure~\ref{iovpts2} depicts another IOVPTS, where $L_I=\{a,b\}$, $L_U=\{x\}$, $L_c=\{a\}$, $L_r=\{b,x\}$ and $L_i=\emptyset$.
	Also, $S_{in}=\{s_0\}$ and $\yGa=\{A\}$. \yfim
	\input{figs/iovpts2.tex}
\end{exam}

%% file: figs/vpts1.tex
\begin{figure}[tb]
\center

\begin{tikzpicture}[node distance=1cm, auto,scale=.6,inner sep=1pt]
  \node[ initial by arrow, initial text={}, punkt] (q0) {$s_0$};
  \node[punkt, inner sep=1pt,right=2.5cm of q0] (q1) {$s_1$};  
\path (q0)    edge [ pil, left=50]
                	node[pil,above]{$c/\ypop{Z}$} (q1);
\path (q0)    edge [ pil, right=50]
                	node[pil,below]{$t/\ypop{Z}$} (q1);
                	
\path (q0)    edge [loop above] node   {$b/\ypush{Z}$} (q0);
\path (q1)    edge [loop above] node   {$c/\ypop{Z},t/\ypop{Z}$} (q1);

\path (q1)    edge [ pil, bend left=50]
                	node[pil]{$\ytau$} (q0);


\end{tikzpicture}
\caption{A VPTS $\yS_1$, with $L_c=\{b\}$, $L_r=\{c,t\}$, $L_i=\yemp$.}

\label{vpts1}
\end{figure}

%% file: figs/iovpts1.tex
\begin{figure}[tb]
\center

\begin{tikzpicture}[node distance=1cm, auto,scale=.6,inner sep=1pt]
  \node[ initial by arrow, initial text={}, punkt] (q0) {$s_0$};
  \node[ initial by arrow, initial text={}, punkt,  inner sep=1pt,right=2.5cm of q0] (q1) {$s_1$};  
\path (q0)    edge [ pil, left=50]
                	node[pil,above]{$c/\ypop{Z}$} (q1);
\path (q0)    edge [ pil, right=50]
                	node[pil,below]{$t/\ypop{Z}$} (q1);
                	
\path (q0)    edge [loop above] node   {$b/\ypush{Z}$} (q0);
\path (q1)    edge [loop above] node   {$c/\ypop{Z},t/\ypop{Z}$} (q1);
  
\path (q1)    edge [ pil, bend left=50]
                	node[pil,below]{$\ytau$} (q0);

\end{tikzpicture}
\caption{An IOVPTS with $L_I=\{b\}$ e $L_U=\{c,t\}$.}

\label{iovpts1}
\end{figure}
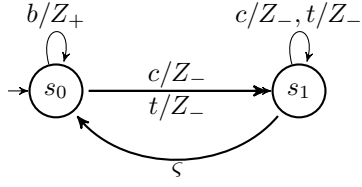

%% file: figs/iovpts2.tex
\begin{figure}[tb]
\center

\begin{tikzpicture}[node distance=1cm, auto,scale=.6,inner sep=1pt]
  \node[ initial by arrow, initial text={}, punkt] (q0) {$s_0$};
  \node[punkt, inner sep=1pt,right=2.5cm of q0] (q1) {$s_1$};  
  \node[punkt, inner sep=1pt,below right=1.5cm of q0] (q2) {$s_2$};  

\path (q0)    edge [ pil, left=50]
                	node[pil,above]{$b/\ypop{A}$} (q1);

\path (q2)    edge [ pil, left=50]
                 	node[pil,right]{$b/\ypop{A}$} (q1);               
\path (q1)    edge [ pil, bend left=50]
                	node[pil]{$b/\ypop{A}$} (q2);
                	                	
\path (q0)    edge [loop above] node   {$a/\ypush{A}$} (q0);
\path (q1)    edge [loop above] node   {$a/\ypush{A}$} (q1);

\path (q0)    edge [ pil, bend right=50]
                	node[below left]{$x/\ypop{A}$} (q2);
\path (q2)    edge [ pil,  right=90]
                	node[pil,above right]{$a/\ypush{A}$} (q0);


\end{tikzpicture}
\caption{An  IOVPTS specification $\yS$ with $L_I=\{a,b\}$ and $L_U=\{x\}$.}

\label{iovpts2}
\end{figure}
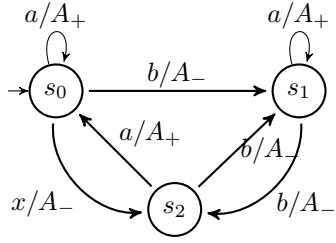

%% file: our-psuites.tex
\section{Conformance Checking and Visibly Pushdown Languages}\label{sec:cfl-suites}

In this section we define a more general conformance relation based on  Visibly Pushdown Languages~\cite{alurm-visibly-2004}, a proper subset of the more general class of context-free languages~\cite{hopcu-introduction-1979}, but a proper superset of the regular languages.
Next we study the notion of test suite completeness 
and give a polynomial time complexity method to check conformance between an IUT and its specification, both based on VPTS models, and using the more general conformance relation over VPLs.

\input{vconf}
\input{our-psuites-notion}
\input{our-psuites-complexity}

%% file: vconf.tex
\subsection{A General Conformance Relation for VPTS models }\label{subsec:conf}

The more general conformance relation is defined on subset of words specified by a tester. 
Informally, consider a language $D$, the set of ``desirable'' behaviors, and a language $F$, the set of ``forbidden'' behaviors, of a system.
If we have a specification VPTS $\yS$ and an implementation VPTS $\yI$ we want to say that $\yI$  \emph{conforms} to $\yS$ according to $(D,F)$ if no undesired behavior in $F$ that is observable in $\yI$  is specified in $\yS$, and all desired behaviors in $D$ that are observable in $\yI$ are specified in $\yS$.
This leads to the following definition.
\begin{defi}\label{def:conf}
Let $L$ be an alphabet, and let $D, F\ysse L^\star$.
Let $\yS$ and $\yI$ be VPTSs over $L$.
We say that \emph{$\yI$ $(D,F)$-visibly conforms to $\yS$}, written $\yI\yconf{D}{F} \yS$, if and only if
\begin{enumerate}
\item $\ysi\in otr(\yI)\cap F$, then $\ysi\not\in otr(\yS)$;
\item $\ysi\in otr(\yI)\cap D$, then $\ysi\in otr(\yS)$.
\end{enumerate}
\end{defi}
We note an equivalent way of expressing these conditions that may also be useful.
Recall that the complement of $otr(\yS)$ is $\ycomp{otr}(\yS)=L^\star-otr(\yS)$.
\begin{prop}\label{prop:equiv-conf}
Let $\yS$ and $\yI$ be VPTSs over $L$ and let $D, F\ysse L^\star$.
Then $\yI\yconf{D}{F} \yS$ if and only if
$otr(\yI)\cap \big[(D\cap\ycomp{otr}(\yS))\cup(F\cap otr(\yS))\big]=\yemp$.
\end{prop}
\begin{proof}
From Definition~\ref{def:conf} we readily get $\yI\yconf{D}{F} \yS$ if and only if  $otr(\yI)\cap F\cap otr(\yS)=\yemp$ and 
$otr(\yI)\cap D\cap \ycomp{otr}(\yS)=\yemp$.
And this holds if and only if 
$$\yemp=\big[otr(\yI)\cap F\cap otr(\yS)\big]\cup\big[otr(\yI)\cap D\cap \ycomp{otr}(\yS)\big]=otr(\yI)\cap \big[(D\cap\ycomp{otr}(\yS))\cup(F\cap otr(\yS))\big],$$
as desired.
\end{proof}

\begin{exam}
\label{exemplo3}
Let $\yS$ be a specification depicted in Figure~\ref{iovpts2}. 
Take the languages $D=\{a^nb^nx: n\geq 0\}$ and $F=\{a^nb^{n+1}:n\geq 0\}$. 
This says that any behavior consisting of a block of $a$s followed by an equal length block of $b$s and terminating by an $x$, is a desirable behavior. 
Any block of $a$s followed by a lengthier block of $b$s is undesirable.
We  want to check whether the implementation $\yI$ conforms to the specification $\yS$ with respect 
to the  sets of behaviors described by $D$ and $F$.
That is, we want to check whether $\yI\yconf{D}{F}\yS$.  
\input{figs/impl}

First, we obtain the  VPA $\ycomp{\yS}$ depicted in Figure~\ref{compiovpts2}. 
Since $\yS$ is  deterministic and all its states are final, we just add a new state $err$ to $\ycomp{\yS}$, and for any missing transitions in $\yS$ we add corresponding transitions ending at $err$ in  $\ycomp{\yS}$. 
It is not hard to see that the language accepted by  $\ycomp{\yS}$ is $\ycomp{otr}(\yS)$. 
\input{figs/compiovpts2}
Again, it is easy to see from Figure~\ref{iovpts2} that  $a^nb^{n+1}\not\in otr(\yS)$, for all $n\geq 0$.
So, $F\cap otr(\yS) = \emptyset$. 
Also we see that the VPA $\yD$, depicted at Figure~\ref{vpaD}, accepts the language $D$ and that $D \subseteq \ycomp{otr}(\yS)$. 
Then the VPA $\yD$ accepts the language $T=D\cap \ycomp{otr}(\yS)=(D\cap\ycomp{otr}(\yS))\cup(F\cap otr(\yS))$.
\input{figs/vpaD}

Now  let  $\yI$ be the implementation  depicted in Figure~\ref{impl}. 
A simple inspection also shows that $aabbx$ is accepted by $\yI$, and we also have $aabbx\in D$.
Hence,  $otr(\yI) \cap D\cap\ycomp{otr}(\yS) =ort(\yI)\cap T\neq \yemp$, and Proposition~\ref{prop:equiv-conf}  implies that $\yI\yconf{D}{F} \yS$ does not hold.

On the other hand, if we assume an implementation $\yI$  that is isomorphic to $\yS$, $\yI$ would not have the transition $\ytr{q_2}{x/\bot}{q_1}$ and then $aabbx$ would not be an observable behavior of $\yI$. 
Actually, in this case,  $otr(\yI) \cap D\cap\ycomp{otr}(\yS) = \yemp$. 
So that now 
$otr(\yI)\cap \big[(D\cap\ycomp{otr}(\yS))\cup(F\cap otr(\yS))\big]=\yemp$, 
and therefore $\yI\yconf{D}{F} \yS$, as expected.  
\yfim
\end{exam} 

Depending on languages $D$ and $F$, we can test  conformance involving several distinct classes of behaviors:
\begin{enumerate}
\item  All observable behaviors of $\yI$ are of interest, and we are not concerned with any undesirable behaviors of $\yI$.
Then, let $D\subseteq L^\star$ and $F=\yemp$.
We get $\yI\yconf{D}{F} \yS$ if and only if $otr(\yI)\ysse otr(\yS)$.

\item 
Let $\yI=\yvptsS$ and let $C\ysse S$ be a subset of the locations of $\yI$.
Behaviors of interest are all observable traces $\ysi$ of $\yI$ that take the initial configuration $(s_0,\bot)$ of $\yI$ to a configuration $(q,\yal)$
and where $q\in C$.
Also, let $E$ be another set of locations of $\yI$ with $E\cap C=\yemp$, and the undesired behaviors of $\yI$ are observable traces  $\ysi$ that lead to a configuration whose location is in $E$.
Then, $\yI\yconf{C}{E} \yS$ if and only if undesired observable traces of $\yI$ are not observable in $\yS$ and all desirable observable traces of $\yI$ are also observable in $\yS$.

\item Let $\yS=\yvptsS$ be a specification and let $\yI=\yvptsQ$ be an implementation. Let $H\ysse L$ be a subset of $L$.
The desirable behaviors of $\yI$ are those observable traces that end in a label in  $H$, and we are not interested in undesirable traces of $\yI$.
In this case, choose $F=\yemp$ and $D=L^\star H$.
\end{enumerate}
Note that $D$  and $F$ are VPLs in all cases listed above.

%

%% file: figs/impl.tex
\begin{figure}[tb]
\center

\begin{tikzpicture}[node distance=1cm, auto,scale=.6,inner sep=1pt]
  \node[ initial by arrow, initial text={}, punkt] (q0) {$q_0$};
  \node[punkt, inner sep=1pt,right=2.5cm of q0] (q1) {$q_1$};  
  \node[punkt, inner sep=1pt,below right=1.5cm of q0] (q2) {$q_2$};  

\path (q0)    edge [ pil, left=50]
                	node[pil,above]{$b/\ypop{A}$} (q1);

\path (q2)    edge [ pil, bend right=50]
                 	node[pil,right]{$b/\ypop{A},x/\ypop{\bot}$} (q1);               
\path (q1)    edge [ pil, right=50]
                	node[pil]{$b/\ypop{A}$} (q2);
                	                	
\path (q0)    edge [loop above] node   {$a/\ypush{A}$} (q0);
\path (q1)    edge [loop above] node   {$a/\ypush{A}$} (q1);

\path (q0)    edge [ pil, bend right=50]
                	node[pil,left]{$x/\ypop{A}$} (q2);
\path (q2)    edge [ pil,  right=50]
                	node[pil,right]{$a/\ypush{A}$} (q0);


\end{tikzpicture}
\caption{An implementation IOVPTS $\yI$ with $L_I=\{a,b\}$ and $L_U=\{x\}$.}

\label{impl}
\end{figure}
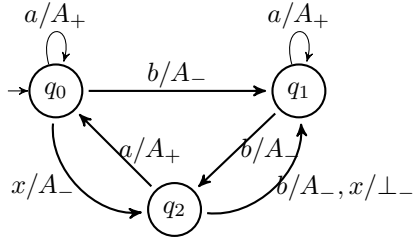

%% file: figs/compiovpts2.tex
\begin{figure}[tb]
\center

\begin{tikzpicture}[node distance=1cm, auto,scale=.6,inner sep=1pt]
  \node[ initial by arrow, initial text={}, punkt] (q0) {$\bar{s_0}$};
  \node[punkt, inner sep=1pt,right=2.5cm of q0] (q1) {$\bar{s_1}$};  
  \node[punkt, inner sep=1pt,below right=1.5cm of q0] (q2) {$\bar{s_2}$};  
  \node[punkt, accepting, inner sep=1pt,below=3cm of q1] (q3) {$err$};  
  
\path (q0)    edge [ pil, bend left=30]
                	node[pil,above]{$b/\ypop{A}$} (q1);

\path (q2)    edge [ pil, bend left=30]
                 	node[pil,right]{$b/\ypop{A}$} (q1);               
\path (q1)    edge [ pil, bend left=30]
                	node[pil,right]{$b/\ypop{A}$} (q2);
                	                	
\path (q0)    edge [loop above] node   {$a/\ypush{A}$} (q0);
\path (q1)    edge [loop above] node   {$a/\ypush{A}$} (q1);

\path (q0)    edge [ pil, bend right=40]
                	node[pil,above right]{$x/\ypop{A}$} (q2);
\path (q2)    edge [ pil,  bend right=40]
                	node[pil,above right]{$a/\ypush{A}$} (q0);

\path (q2)    edge [ pil, right=30]
                 	node[pil,pos=.35,above right]{$x/\ypop{A}$} node[pil,below left]{$b/\ypop{\bot}$} (q3);  
\path (q2)    edge [ pil, right=30]
					node[pil,pos=.65,above right]{$x/\ypop{\bot}$} node[pil,below left]{$b/\ypop{\bot}$} (q3);

\path (q1)    edge [ pil, bend left=50]
                 	node[pil,pos=.3,right]{$x/\ypop{A}$} (q3);         
\path (q1)    edge [ pil, pos=.5,bend left=50]
					node[pil,right]{$x/\ypop{\bot}$} (q3);         
\path (q1)    edge [ pil, pos=.7,bend left=50]
					node[pil,right]{$b/\ypop{\bot}$} (q3);

\path (q0)    edge [ pil, bend right=50]
                 	node[pil,pos=.45, below  left]{$b/\ypop{\bot}$} (q3);         
\path (q0)    edge [ pil, bend right=50]
					node[pil,pos=.55,below left]{$x/\ypop{\bot}$} (q3);

  \path (q3)    edge [ loop below]
                   	node[pil]{$a/\ypush{A},b/\ypop{A},x/\ypop{A},b/\ypop{\bot},x/\ypop{\bot}$}  (q3);                        	
                	  

\end{tikzpicture}
\caption{The VPA accepting $\ycomp{otr}(\yS)$ for the IOVPTS $\yS$ of Figure~\ref{iovpts2}.}

\label{compiovpts2}
\end{figure}
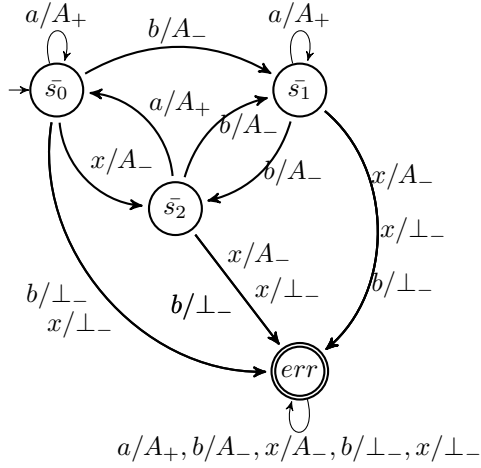

%% file: figs/vpaD.tex
\begin{figure}[tb]
\center

\begin{tikzpicture}[node distance=1cm, auto,scale=.6,inner sep=1pt]
  \node[ initial by arrow, initial text={}, punkt] (q0) {$d_0$};
  \node[punkt, inner sep=1pt,right=2cm of q0] (q1) {$d_1$};  
    \node[punkt, accepting, inner sep=1pt,right=2cm of q1] (q2) {$d_2$};  
    
\path (q0)    edge [ pil, left=50]
                	node[pil,above]{$b/\ypop{A}$} (q1);
\path (q1)    edge [ pil, left=50]
                	node[pil,above]{$x/\ypop{\bot}$} (q2);
                	               	
\path (q0)    edge [loop above] node   {$a/\ypush{A}$} (q0);
\path (q1)    edge [loop above] node   {$b/\ypop{A}$} (q1);


\end{tikzpicture}
\caption{The VPA $\yD$ accepting $D=a^nb^nx$.}

\label{vpaD}
\end{figure}
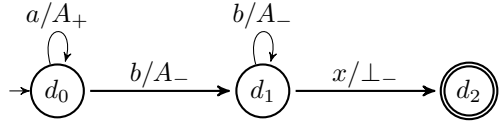

%% file: our-psuites-notion.tex
\subsection{Test Suite Completeness over VPLs}\label{sec:test-suites}

We first state the general definition of a test suite.
\begin{defi}\label{def:testsuite}
Let $L$ be a set of symbols.
A test suite $T$ over $L$ is a language over $L$, \emph{i.e.} $T\ysse L^\star$. 
Each $\ysi\in T$ is called a \emph{test case}.
\end{defi}
A test suite $T$ should be engineered  to detect faulty observable behavior of any given IUT, when compared to what has been determined by a specification.
In this case, $T$ can be seen as specifying a fault model, in the sense that test cases in $T$ represent faulty observable behaviors. 
In particular, if $T$ is a VPL, then it can be specified by a VPA $\yA$.
Alternatively, we could specify $T$ by a finite set of VPAs, so that
the union of all the undesirable behaviors specified by these VPAs comprise the fault model.

Next we say that an implementation $\yI$ satisfies, or adheres, to a test suite $T$ when 
no observable behavior of $\yI$ is a harmful behavior present in $T$.
\begin{defi}\label{def:adherence}
Let $T$ be a test suite over an alphabet $L$.
A VPTS $\yQ$ over $L$ \emph{adheres} to $T$  if  $\ysi\not\in T$ for all $\ysi\in otr(\yQ)$.
Further, an IOVPTS $\yI$ over $L$ adheres to $T$ if its underlying VPTS adheres to $T$.
\end{defi}
In general, we want test suites to be sound, in the sense that adherence always implies visual conformance.
Moreover, the converse is also desirable, that is, when we have visual conformance, then we also have adherence.
\begin{defi}\label{def:complete}
Let $L$ be an alphabet and let $T$ be a test suite over $L$.
Let $\yS$ be a VPTS over $L$, and let $D, F\ysse L ^\star$ be languages over $L$.
We say that:
\begin{enumerate}
\item  $T$ is \emph{sound} for $\yS$ and $(D,F)$  if 
$\yI$ adheres to $T$ implies $\yI \yconf{D}{F} \yS$, for all VPTS $\yI$ over $L$.
\item $T$ is \emph{exhaustive} for $\yS$ and $(D,F)$ if $\yI \yconf{D}{F} \yS$ implies that $\yI$ adheres to $T$, for all VPTS $\yI$ over $L$.
\item $T$ is \emph{complete} for $\yS$ and $(D,F)$ if it is both sound and exhaustive 
for $\yS$ and $(D,F)$.
\end{enumerate}
\end{defi}
Note that, for convenience, adherence is defined in the negative, that is, $\yI$ adheres to $T$ if we have $otr(\yI)\cap T=\yemp$. 
Of course we could reverse the notions of soundness and exhaustiveness by substituting $L^\star-T$ for $T$ throughout.

It comes as no surprise that the test suite we can extract from Proposition~\ref{prop:equiv-conf} is always complete.
But, furthermore, we will also show that it is unique.

\begin{lemm}\label{lemm:always-complete}
Let $\yS$ be a specification VPTS over $L$, and let $D$, $F\ysse L^\star$ be a pair of languages over $L$.
Then, the set 
$\big[(D\cap\ycomp{otr}(\yS))\cup(F\cap otr(\yS))\big]$ is the only complete test suite for $\yS$ and $(D,F)$. 
\end{lemm}
\begin{proof}
Write $T=\big[(D\cap\ycomp{otr}(\yS))\cup(F\cap otr(\yS))\big]$,
and let $\yI$ be any implementation VPTS over $L$.
From Definition~\ref{def:adherence}, we know that $\yI$ adheres to $T$ if and only if $otr(\yI)\cap T=\yemp$.
From Proposition~\ref{prop:equiv-conf} we get that $\yI\yconf{D}{F} \yS$ if and only if $otr(\yI)\cap T=\yemp$.
Hence, $\yI$ adheres to $T$ if and only if $\yI\yconf{D}{F} \yS$.
Since $\yI$ was arbitrary, from Definition~\ref{def:complete} we conclude that $T$ is a complete test suite for $\yS$ and $(D,F)$.

Now, take another test suite $Z\ysse L^\star$, with $Z\neq T$.
For the sake of contradiction, assume that $Z$ is also complete for $\yS$ and $(D,F)$.
Fix any implementation $\yI$.
Since $Z$ is complete, Definition~\ref{def:complete} says that 
$\yI$ adheres to $Z$ if and only if $\yI\yconf{D}{F} \yS$.
Using Proposition~\ref{prop:equiv-conf} we know that $\yI\yconf{D}{F} \yS$ if and only if $otr(\yI)\cap T=\yemp$.
Hence, $\yI$ adheres to $Z$ if and only if $otr(\yI)\cap T=\yemp$.
From Definition~\ref{def:adherence} we know that $\yI$ adheres to $Z$  if and only if  $otr(\yI)\cap Z=\yemp$.
We conclude that $otr(\yI)\cap Z=\yemp$ if and only if $otr(\yI)\cap T=\yemp$.
But $Z\neq T$ gives some $\ysi\in L^\star$ such that $\ysi\in T$ and $\ysi\not\in Z$.
The case $\ysi\not\in T$ and $\ysi\in Z$ is entirely analogous. 
We now have $\ysi\in T\cap \ycomp{Z}$.
If we can construct an implementation VPTS $\yQ$ over $L$  with $\ysi\in otr(\yQ)$, then we have reached a contradiction 
because we would have $\ysi\in otr(\yQ)\cap T$ and $\ysi\not\in otr(\yQ)\cap Z$.
But that is simple. 
Let $\ysi=x_1x_2\ldots x_k$, with $k\geq 0$ and $x_i\in L$ ($1\leq i\leq k$). 
Define $L_c=L_r=\yemp$ and $L_i=L$, and let $\yQ=\yvpts{Q}{\{q_0\}}{L}{\yemp}{R}$, where $Q=\{q_i\yst 0\leq i\leq k\}$, and $R=\{(q_{i-1},x_i,\yait,q_i)\yst 1\leq i\leq k\}$.
Clearly, $\ysi\in otr(\yQ)$, concluding the proof. 
\end{proof}
Lemma~\ref{lemm:always-complete} says that the test suite $T=\big[(D\cap\ycomp{otr}(\yS))\cup(F\cap otr(\yS))\big]$ is complete for the specification $\yS$ and the pair of languages $(D,F)$.
So, given an implementation $\yI$, checking if it $(D,F)$-visibly conforms to $\yS$ is equivalent  to checking if $\yI$ adheres to $T$ and, by Definition~\ref{def:adherence},  the latter is equivalent to checking that we have $otr(\yI)\cap T=\yemp$.

We also note that, in Lemma~\ref{lemm:always-complete}, in order to construct $\yQ$ with $\ysi\in otr(\yQ)$ it was crucial that we had no restrictions on the size of $\yQ$, since we have no control over the size of the witness $\ysi$.


%% file: our-psuites-complexity.tex
\subsection{Checking Visual Conformance for VPTS models}\label{sec:suites-complexity}

When testing conformance one important issue is the size of test suites, relatively to the size of the given specification.
Let $\yS=\yvptsS$  be a VPTS.
A reasonable measure of the size of $\yS$ would be the number of symbols required to write down a complete syntactic description of $\yS$.
Assume that $\yS$ has $m=\vert T\vert$ transitions, $n=\vert S\vert$ states, $\ell=\vert L\vert$ action symbols, and $p=\vert\yGa\vert$ stack symbols. 
Since any transition can be written using $\yoh{\ln (n\ell p)}$ symbols, 
the size of $\yS$ is  $\yoh{m\ln (n\ell p)}$.
From Remark~\ref{rema:lte-finite}, we see that $n$ is $\yoh{m}$ and, clearly, so are $\ell$ and $p$.
Thus, the size of $\yS$ is bounded by $\yoh{m\ln m}$.
If we fix the stack and action alphabets, then the size of the VPTS will be bounded by 
$\yoh{m}$.
In what follows, and with almost no prejudice, we will ignore the small logarithmic factor.%
\footnote{It is also customary to write $\yoh{m\ln m}$ as $\widetilde{\mathcal{O}}(m)$.
In the sequel, we can always replace $\yoh{\cdot}$ by $\widetilde{\mathcal{O}}(\cdot)$.} 

Given visibly pushdown languages $D$ and $F$ over $L$, and given a specification $\yS$ over $L$, Lemma~\ref{lemm:always-complete} 
 says that the fault model $T$ is complete for $\yS$ and $(D,F)$, where
$T=\big[(D\cap\ycomp{otr}(\yS))\cup(F\cap otr(\yS))\big]$.
Assume that $L(\yA_D)=D$ and $L(\yA_F)=F$ where $\yA_D$ and $\yA_F$ are deterministic VPAs with  $n_D$ and $n_F$ states, respectively. 
Also, assume that $\yS$ is deterministic with $n_S$ states.
Proposition~\ref{prop:contracted-vpts} says that we may as well take $\yS$ as a contracted and deterministic VPTS.
Now, Proposition~\ref{prop:plts-pda} 
gives a deterministic VPA $\yA_1$  with $n_S$ states and such that $L(\yA_1)=otr(\yS)$. 
Using Proposition~\ref{prop:compl-vpa},
we can construct a deterministic VPA $\yB_1$ with $n_S +1$ states and such that $L(\yB_1)=\ycomp{L(\yA_1)}=\ycomp{otr}(\yS)$.
Using Proposition~\ref{prop:determ-complement} we can obtain a deterministic VPA $\yA_2$ with at most $n_S n_F$ states 
and such that $L(\yA_2)=L(\yA_F)\cap L(\yA_1)=F\cap otr(\yS)$, and also a 
deterministic VPA $\yB_2$ with $(n_S +1) n_D$ states 
such that $L(\yB_2)=L(\yA_D)\cap L(\yB_1)=D\cap \ycomp{otr}(\yS)$.
Proposition~\ref{prop:cup-vpa}  gives a deterministic VPA $\yT$ with $(n_Sn_F+1)(n_Sn_D+n_D+1)$  states 
and such that $L(\yT)=L(\yA_2)\cup L(\yB_2)=T$.
Proposition~\ref{prop:cup-vpa} also says that $\yT$ is non-blocking and has no $\yeps$-moves. 
Lemma~\ref{lemm:always-complete} says that $L(\yT)$ is a complete test suite for $\yS$ and $(D,F)$.
\begin{prop}\label{prop:conf-poli}
Let $L$ be an alphabet with $|L|=n_L$.
Let $\yS$ and $\yI$ be deterministic IOVPTSs over $L$ with $n_S$ and $n_I$ states, respectively.
Let $\yA_D$ and $\yA_F$ be  deterministic VPAs over $L$ with $n_D$ and $n_F$ states, respectively,
with $L(\yA_D)=D$, $L(\yA_F)=F$.
Then, we can construct a deterministic, non-blocking VPA $\yT$ with at most $(n_Sn_F+1)(n_Sn_D+n_D+1)$ states and no $\yeps$-moves, 
and such that $L(\yT)$ is a complete test suite for $\yS$ and $(D,F)$.
Moreover, there is an algorithm with polynomial time complexity $\yoh{(n_In_S^2n_Fn_D)^3}$ and that checks if $\yI \yconf{D}{F} \yS$. 
\end{prop}
\begin{proof}
The preceding discussion gives a  deterministic and non-blocking VPA $\yT$ with at most $n_T=(n_Sn_F+1)(n_Sn_D+n_D+1)$ states and no $\yeps$-moves, and  
such that $L(\yT)=T=\big[(D\cap\ycomp{otr}(\yS))\cup(F\cap otr(\yS))\big].$
Since $\yI$ is deterministic, using Propositions~\ref{prop:contracted-vpts} and~\ref{prop:plts-pda}, we can  get a  deterministic VPA $\yA$ with at most $n_I$ states, 
and  such that $otr(\yI)=L(\yA)$. 
From Proposition~\ref{prop:determ-complement}
we can construct a deterministic VPA $\yB$ with at most $n_In_T$ states, 
and such that $L(\yB)=L(\yA)\cap L(\yT)=otr(\yI)\cap T$.

The emptiness problem for a VPA is decidable in time $\yoh{n^3}$, where $n$ is the number of states in the VPA~\cite{alurm-visibly-2004,brevccr-multi-1996,lang-p-2011,caromp-2-2007}.
Hence, we can check whether  $\yI \yconf{D}{F} \yS$ in asymptotic time $\yoh{(n_In_S^2n_Fn_D)^3)}$. 
\end{proof}
Note that, if the specification $\yS$ and the model languages $D$ and $F$ are fixed, so that $n_S$, $n_D$ and $n_F$ are constants, the algorithm runs in polynomial time $\yoh{n^3}$, where $n$ is the number of states in the implementation being tested for visual conformance.

%% file: testing-iovpts.tex
\section{An \emph{ioco-like} Conformance Checking for IOVPTS Models}\label{sec:tretma-suites}

In this section we investigate whether IUTs, also described as IOVPTS models, conform to a given specification IOVPTS model.
Here we use a new notion of an {\bf ioco-like} conformance relation for IOVPTSs. 
That notion captures the same idea as the standard notion of conformance as studied by Tretmans~\cite{tret-model-2008}, but using Labeled Transition Systems (LTSs) which do not have access to a pushdown stack memory.  
The idea is that, given a specification $\yS$ and an IUT $\yI$, we say whether $\yI$ {\bf ioco-like} conforms to $\yS$ when, for any valid observable behavior $\ysi$ of $\yS$, any output symbol that $\yI$ can emit after running over $\ysi$ is, necessarily, among the output symbols that $\yS$  can also emit after it runs over the same $\ysi$. 

\input{ioco-like}
\input{iovpts-fault-model}

\input{iovpts-algorithm}

\input{iovpts-complexity}
\input{iovpts-example}

%% file: ioco-like.tex
\subsection{An ioco-like Conformance  Relation  for IOVPTS models}\label{subsec:plts-ioco} 

Here we define an ioco-like conformance relation~\cite{tret-model-2008}  for IOVPTS models.
Let $\yS$ be a specification IOVPTS and  let $\yI$ be an implementation IOVPTS.
The ioco-like relation essentially requires that any observable trace $\ysi$ of $\yI$ is also an observable trace of $\yS$ and, further, if $\ysi$ leads $\yI$ to a configuration from which $\yI$ can emit the output label $\ell$, then $\yS$ must also end up in a configuration from which the label $\ell$ can also be output. 
Next we formalize these ideas.
\begin{defi}\label{def:out-after}
Let $\yS=\yiovptsS$ and $\yI=\yiovptsQ$ be  IOVPTSs, with $L=L_I\cup L_U$.
\begin{enumerate}
\item Define the function $\yafter\!\!\!: \yltsconf{\yS}\times  L^\star\rightarrow \ypow{\yltsconf{\yS}}$ by letting $$(s,\yal) \yafter \ysi = \big\{(q,\ybe) \yst \ytrt{(s,\yal)}{\ysi}{(q,\ybe)}\big\},\quad \text{for all $(s,\yal)\in \yltsconf{\yS}$, $\ysi\in L^\star$}.$$
\item We say that $\yI \yiocolike \yS$ if for all $\ysi\in otr(\yS)$,  $q_0\in Q_{in}$, $\ell\in\yout((q_0,\bot) \yafter \ysi)$ there is some $s_0\in S_{in}$ such that  
$\ell\in\yout((s_0,\bot) \yafter \ysi)$.
\end{enumerate}   
\end{defi}

Now we show that the $\!\!\yiocolike \!\!$ conformance relation can be seen as an instance of the conformance relation based on languages given in Definition~\ref{def:conf}. 
\begin{lemm}\label{lemm:ioco-free}
Let $\yS=\yiovptsS$ be a specification IOVPTS and let  $\yI=\yiovptsQ$ be an implementation IOVPTS.
Then  $D=otr(\yS) L_U$ is a VPL, and we have that $\yI \yiocolike \yS$ if and only if
$\yI \yconf{D}{\yemp} \yS$.
\end{lemm}
\begin{proof}
From Definition~\ref{def:iolts-semantics} we know that the semantics of an IOVPTS is given by the semantics of its underlying VPTS.
So, for the remainder of this proof, when we write $\yS$ and $\yI$ we will be referring to the underlying VPTSs 
of the given IOVPTSs $\yS$ and $\yI$, respectively.
By Proposition~\ref{prop:plts-pda} we see that $L(\yA)=otr(\yS)$, where $\yA$ is obtained using  Definition~\ref{def:plts-pda}, is a VPL. 
Using Proposition~\ref{prop:conct-vpa} we conclude that $D$ is also a VPL.

Next, we argue that $\yI\yiocolike \yS$ if and only if $\yI \yconf{D}{\yemp} \yS$.
First assume that we have $\yI \yconf{D}{\yemp} \yS$. 
Because $otr(\yI) \cap \yemp \cap otr(\yS)=\yemp$, it is clear from Definition~\ref{def:conf} that $\yI \yconf{D}{\yemp} \yS$ is equivalent  
to $otr(\yI)\cap D\ysse otr(\yS)$.
Let $\ysi\in otr(\yS)$ and let $\ell\in \yout((q_0,\bot) \yafter \ysi)$ for some $q_0\in Q_{in}$.
We must show that $\ell\in \yout((s_0,\bot)\yafter \ysi)$ for some $s_0\in S_{in}$.
Because  $\ell\in \yout((q_0,\bot) \yafter \ysi)$ we get 
$\ysi, \ysi\ell\in otr(\yI)$.
Since $\ell\in L_U$, we get $\ysi\ell\in otr(\yS) L_U$ and so $\ysi\ell\in D$.
We conclude that $\ysi\ell\in otr(\yI)\cap D$.
Since we already know that $otr(\yI)\cap D\ysse otr(\yS)$, we now have 
$\ysi\ell \in otr(\yS)$.
So, $\ell\in \yout((s_0,\bot) \yafter \ysi)$ for some $s_0\in S_{in}$, as desired.

Next, assume that $\yI \yiocolike \yS$ and we want to show that $\yI \yconf{D}{\yemp} \yS$ holds.
Since $otr(\yI)\cap \yemp\cap otr(\yS)=\yemp$, 
the first condition of Definition~\ref{def:conf} is immediately verified.
We now turn to the second condition of Definition~\ref{def:conf}.
In order to show that $otr(\yI)\cap D\ysse otr(\yS)$,
let $\ysi\in otr(\yI)\cap D$.
Then, $\ysi\in D$ and so $\ysi=\yal\ell$ with $\ell\in L_U$ and $\yal\in otr(\yS)$, because $D=otr(\yS) L_U$.
Also, $\ysi\in otr(\yI)$ gives $\yal\ell\in otr(\yI)$, and so $\yal\in otr(\yI)$.
Then, because $\ell\in L_U$, we get $\ell\in \yout((q_0,\bot) \yafter \yal)$ for some $q_0\in Q_{in}$.
Because we assumed $\yI \yiocolike \yS$ and we have $\yal\in otr(\yS)$, we also get $\ell\in \yout((s_0,\bot) \yafter \yal)$, for some $s_0\in S_{in}$.
So $\yal\ell\in otr(\yS)$. Because $\ysi=\yal\ell$, we have $\ysi\in otr(\yS)$.
We have, thus, showed that $otr(\yI)\cap D\ysse otr(\yS)$, as desired.
\end{proof}

\begin{exam}
We illustrate the relationship  between $\yI \yiocolike \yS$ and $\yI \yconf{D}{\yemp} \yS$, 
using the specification IOVPTS $\yS$  depicted in Figure~\ref{iovpts2} and the implementation $\yI$  depicted in Figure~\ref{impl}.

We  want to check whether $\yI\yconf{D}{F} \yS$ holds. 
Let $\ysi=aabb$.
From Figure~\ref{iovpts2} it is apparent that  $\ytrt{(s_0,\bot)}{\ysi}{(s_2,\bot)}$ and that 
$(s_0,\bot) \yafter \ysi=\{(s_2,\bot)\}$.
From Figure~\ref{impl} we get   $(q_0,\bot) \yafter \ysi=\{(q_2,\bot)\}$.
Also, $x\in \yout((q_2,\bot))$, but $x\not\in \yout((s_2,\bot))$.
So, by Definition~\ref{def:out-after}, $\yI\yiocolike \yS$ does not hold.

Now take $\ysi=aabbx$.
Since $aabb\in otr(\yS)$, we get $aabbx\in otr(\yS)L_U=D$.
Also, $aabbx\in otr(\yI)$ and $aabbx\not\in otr(\yS)$, so that
$aabbx\in otr(\yI)\cap D\cap \ycomp{otr}(\yS)$.
Using Proposition~\ref{prop:equiv-conf}, we conclude that $\yI\yconf{D}{\yemp} \yS$ does not hold also, as expected. \yfim
\end{exam} 

%

We can also characterize the {\bf ioco-like} relation as follows.
\begin{coro}\label{coro:ioco-charac}
Let $\yS$ be a specification IOVPTS and let  $\yI$ be an implementation IOVPTS.
Then $\yI \yiocolike \yS$ if and only if $otr(\yI)\cap T= \yemp$, where $T=\ycomp{otr}(\yS)\cap \big[otr(\yS) L_U\big]$.
\end{coro}
\begin{proof}
From Lemma~\ref{lemm:ioco-free} we have that $\yI \yiocolike \yS$ if and only if $\yI \yconf{D}{\yemp} \yS$, where 
$D=otr(\yS)L_U$.
From Proposition~\ref{prop:equiv-conf} we know that the latter holds if and only if 
$otr(\yI)\cap (D\cap\ycomp{otr}(\yS))=\yemp$.
\end{proof}

\begin{exam}
	Let the IOVPTS $\yS$ of Figure~\ref{iovpts2} be the specification model, with $L_I=\{a,b\}$, and $L_U=\{x\}$. 
	It is not hard to see that the VPA $\yTS$, depicted in Figure~\ref{capSD}, is such that $L(\yTS)=D\cap \ycomp{otr}(\yS)$, where $D=otr(\yS)L_U$.
	\input{figs/capSD}
	According to  Lemma~\ref{lemm:always-complete} and Corollary~\ref{coro:ioco-charac}, the language accepted by  $\yTS$, $L(\yTS)$, is a complete test suite for $\yS$ and $(D,\yemp)$.  \yfim
\end{exam}

%% file: figs/capSD.tex
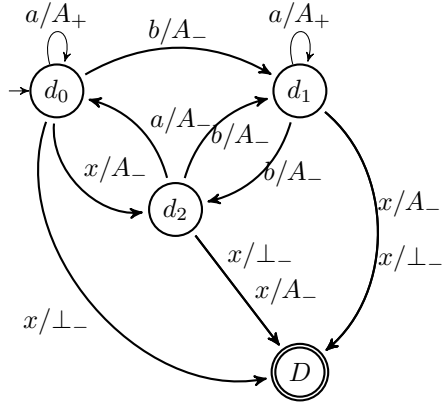
\begin{figure}[tb]
\center

\begin{tikzpicture}[node distance=1cm, auto,scale=.6,inner sep=1pt]
  \node[ initial by arrow, initial text={}, punkt] (q0) {$d_0$};
  \node[punkt, inner sep=1pt,right=2.5cm of q0] (q1) {$d_1$};  
  \node[punkt, inner sep=1pt,below right=1.5cm of q0] (q2) {$d_2$};  
  \node[punkt, accepting, inner sep=1pt,below=3cm of q1] (q3) {$D$};  
  
\path (q0)    edge [ pil, bend left=30]
                	node[pil,above]{$b/\ypop{A}$} (q1);

\path (q2)    edge [ pil, bend left=30]
                 	node[pil, pos=.4,right]{$b/\ypop{A}$} (q1);               
\path (q1)    edge [ pil, bend left=30]
                	node[pil, right]{$b/\ypop{A}$} (q2);
                	                	
\path (q0)    edge [loop above] node   {$a/\ypush{A}$} (q0);
\path (q1)    edge [loop above] node   {$a/\ypush{A}$} (q1);

\path (q0)    edge [ pil, bend right=50]
                	node[pil, above right]{$x/\ypop{A}$} (q2);
\path (q2)    edge [ pil,  bend right=30]
                	node[pil,pos=0.4,above right]{$a/\ypop{A}$} (q0);

\path (q2)    edge [ pil, right=30]
                 	node[pil, pos=0.65, above right]{$x/\ypop{A}$} (q3);         
\path (q2)    edge [ pil, right=30]
					node[pil, pos=0.35,above right]{$x/\ypop{\bot}$} (q3);         

\path (q1)    edge [ pil, bend left=50]
                 	node[pil,pos=.4,right]{$x/\ypop{A}$} (q3);      
\path (q1)    edge [ pil, bend left=50]
					node[pil,pos=.6,right]{$x/\ypop{\bot}$} (q3);

\path (q0)    edge [ pil, bend right=60]
                 	node[pil,below left]{$x/\ypop{\bot}$} (q3);         
                                 	  

\end{tikzpicture}
\caption{The VPA that accepts the language  $\ycomp{otr}(\yS)\cap\big[otr(\yS)L_U\big]$ for the IOVPTS $\yS$ of Figure~\ref{iovpts2}.}
\label{capSD}
\end{figure}

%% file: iovpts-fault-model.tex
\subsection{IOVPTS Fault Models and Test Suite Completeness}\label{subsec:fault-models}

Here we introduce the notion of an external tester environment that can also be formalized using an Input/Output Visibly Pushdown Transition System with a special \yfail\ state. 
We call this model a \emph{Visibly Pushdown Fault Model}. 
Such a model, $\yT$,  operates in conjunction with an IUT, $\yI$.
Their joint behavior can be interpreted as $\yT$ sending symbols to $\yI$, and receiving symbols from it. 
Therefore, the sets of input and output symbols in $\yT$ and $\yI$ are reversed.

\begin{defi}\label{def:test-purpose}
	Let $L_I$ and $L_U$ be sets of input and output symbols, respectively, with $L=L_I\cup L_U$.
	A \emph{visibly pushdown fault model} (\emph{VPFM}) over $L$ is any  IOVPTS $\yT\in\yiovp{L_U}{L_I}$ with a  distinguished $\yfail$ state.
\end{defi}

Given a VPFM $\yT$ and an IUT $\yI$, in order to formally describe the exchange of action symbols between $\yT$ and $\yI$, we define their cross-product  VPTS $\yT\times \yI$.
The cross-product is easily constructed by taking the product of  the corresponding VPAs.
Recall Definitions~\ref{def:productVPA} and~\ref{def:plts-pda}.
\begin{defi}\label{def:cross}
Let $\yS$ and $\yI$ be two IOVPTSs over an alphabet $L$.
Their cross-product is the VPTS $\yS\times \yI$ induced by the product VPA $\yA_\yS\times \yA_\yI$,
where $\yA_\yS$ and $\yA_\yI$ are the VPAs induced by $\yS$ and $\yI$, respectively.
\end{defi}

Having an implementation $\yI$ and a VPFM $\yT$ as a tester, we need to say  when a test run is successful with respect to a given specification $\yS$.
Recalling that  $\yT$ signals an unsuccessful run when it reaches a {\bf fail} state, the  following definition formalizes the verdict of a test run. 
Recall Definition~\ref{def:out-after}.
Importantly, given a specification $\yS$, when the test run is successful  we need a guarantee that $\yI$ does  {\bf ioco-like} conform to  $\yS$, and conversely, that the test run surely fails when $\yI$ does not {\bf ioco-like} conform to $\yS$, for any implementation $\yI$.
That is, we need a property of completeness.
It is also customary to specify that  IUTs of interest are only taken from particular subclasses of IOVPTS models, and completeness is then defined when taking implementations from these subclasses only.
\begin{defi}\label{def:passes}
	Let $\yT=\yiovptsU$ be a VPFM and let $\yI=\yiovptsS$ be an implementation over the same alphabet
	$L=L_I\cup L_U$.
	We say that $\yI$ \emph{passes} $\yT$ if, for all  $\ysi\in L^\star$ and all initial configurations $((t_0,q_0),\bot)$ of their cross-product $\yT\times \yI$ we do not have 
	$\ytrut{((t_0,q_0),\bot)}{\ysi}{\yT\times \yI}{((\yfail,q),\yal\bot)}$ 
	for any configuration $((\yfail,q),\yal\bot)$ of $\yT\times \yI$.
\end{defi}

Now we construct a VPFM which is complete relatively to a given specification $\yS$.
\begin{lemm}\label{lemm:ioco-complete} 
	Let $\yB\in \yiovp{L_I}{L_U}$ be a deterministic specification with $n$ states. 
	We can effectively construct a  fault model $\yT$  which is {\bf ioco-like} complete for $\yB$.
	Moreover, $\yT$ is deterministic  and  has $n+1$ states.
\end{lemm}
\begin{proof}
From Corollary~\ref{coro:ioco-charac} we know that, for any IUT $\yI$, $\yI \yiocolike \yB$ if and only if 
$otr(\yI)\cap T=\yemp$, where $T=\ycomp{otr}(\yB)\cap [otr(\yB) L_U]$.
The desired  fault model $\yT$ can be constructed as follows.
If $\ell$ is a push symbol in $L_U\cap L_c$ add to $\yS$ a transition $(s,\ell,Z,\yfail)$ if we do not have in $\yS$ any other transitions like $(s,\ell,W,p)$ in $\yS$ for any  $p$ and any $W$.
Likewise, when $\ell$ is a pop symbol in $L_U$ add a transition  $(s,\ell,W,\yfail)$
given that we do not have any other transitions $(s,\ell,W,p)$, for any $p$.
Finally, for any simple symbol $\ell$ in $L_U$, add $(s,\ell,\yait,\yfail)$, if we do not in $\yS$ have a transition  $(s,\ell,\yait,p)$, for any $p$.
Then, an argument can show that $\yI$ passes $\yT$ if and only if $\yI \yiocolike \yB$, for any IUT $\yI$.

Details can be found in the complete proof in Appendix~\ref{app:lemm:ioco-complete}, at page~\pageref{app:lemm:ioco-complete}. 
\end{proof}

%% file: iovpts-algorithm.tex
\subsection{Testing IUT models for {\bf ioco-like} conformance}\label{subsec:testing-iuts}

Once we have a VPFM which is complete for a given specification, we can test whether IUTs {\bf ioco-like} conform to that specification. 
In order to do so, we introduce the notion of \emph{balanced run}. 
Let $\yV$ be any VPTS over $L$.
We say that a string $\ysi\in L^\star$ induces a \emph{balanced run}  from $p$ to $q$ in $\yV$ if we have $\ytrtf{(p,\bot)}{\ysi}{\yV}{(q,\bot)}$.
The next theorem gives a decision procedure for testing {\bf ioco-like} conformance.
\begin{theo}\label{lemm:vpts-test}
Let $\yS=\yiovpts{S_\yS}{\{s_0\}}{L_I}{L_U}{\yDe_\yS}{T_\yS}\in\yiovp{L_I}{L_U}$ be a deterministic specification, and let $\yI=\yiovpts{S_\yI}{I_{in}}{L_I}{L_U}{\yDe_\yI}{T_\yI}\in\yiovp{L_I}{L_U}$ be an IUT.
Then we can effectively decide whether $\yI \yiocolike S$ holds.
Further,  if  $\yI \yiocolike S$ does not hold, we can find $\ysi\in otr(\yS)$, $\ell\in L_U$ that verify this condition, \emph{i.e.}, $\ell\in\yout((q_0,\bot) \yafter \ysi)$  for some $q_0\in I_{in}$, and $\ell\not\in \yout((s_0,\bot) \yafter \ysi)$. 
\end{theo}
\begin{proof}
Let $L=L_I\cup L_U$.
From Corollary~\ref{coro:ioco-charac} we know that $\yI \yiocolike \yS$ does not hold if and only if  there is some $\ysi\in otr(\yS)$, $\ell\in L_U$ such that $\ysi\ell\in otr(\yI)\cap T$, where $T=\ycomp{otr(\yS)}\cap\big[ otr(\yS)L_U\big]$.
The proof of Lemma~\ref{lemm:ioco-complete}, and Proposition~\ref{prop:vpts-deterministic}, indicate how to obtain a deterministic fault model $\yT=\yiovpts{S_\yT}{\{t_0\}}{L_U}{L_I}{\yDe_\yT}{T_\yT}$, with no $\ytau$-moves, and such that $\yI \yiocolike \yS$ if and only if $\yI$ does not pass $\yT$, that is, 
$\ytrut{((t_0,q_0),\bot)}{\ysi}{\yP}{((\yfail,q),\yal\bot)}$ for some $\ysi\in L^\star$,
where $\yP=\yT\times \yI$ is the product IOVPTS, and  $(t_0,q_0)$ is an initial state of $\yP$.

Write  $\yA=\yvpts{A}{I}{L}{\yGa}{\rho}$ for the underlying VPTS associated to $\yP$.
In order to check for $\yiocolike$ conformance, and using Proposition~\ref{prop:lang-vpa-vlpts}, it suffices to check whether a configuration $((\yfail,q),\yal\bot)$ is reachable from some initial configuration of $\yA$, and where $q$ can be any state of $\yI$. 
We will modify $\yA$ in a simple way in order to make this reachability problem more amenable.
\begin{description}
\item[\sf Emptying the stack.]
First we move to a pop state after reaching $\yfail$ in $\yT$.
For all states  $(\yfail,q)$ in $\yA$, add the internal transition $((\yfail,q),\ytau,\yait,f_1)$ to $\rho$, where $f_1$ is a new state added to $A$.
Then, for all $W\in \yGa$ add the self-loops $(f_1,b_1,W,f_1)$ to $\rho$, where $b_1$ is a new pop symbol added to $L$.
Next, add the transition $(f_1,b_1,\bot,f_2)$ to $\rho$, where $f_2$ is another new state added to $A$.
Let $\yA_1=\yvpts{A_1}{I}{L_1}{\yGa}{\rho_1}$ be the resulting VPTS obtained after these modifications to $\yA$.
Since $(\yfail,q)$ is a sink state in $\yA$, it is easy to see that we have $\ytrtf{((t_0,q_0),\bot)}{\ysi}{}{((\yfail,q),\yal\bot)}$ in $\yA$ if and only if  $\ytrtf{((t_0,q_0),\bot)}{\mu}{}{(f_2,\bot)}$ in $\yA_1$, where $\mu=\ysi \ytau b_1^k$ with $k=\vert\yal\vert+1$.

\item[\sf Eliminating moves on an empty stack.]
Let $s_0$ be a new state added to $A_1$, $a_2$ a new push symbol added to $L_1$, and $Z_2$ a new stack symbol added to $\yGa$.
We make $s_0$ the new (unique) initial state and add to $\rho_1$ the self-loop $(s_0,a_2,Z_2,s_0)$.
Next we connect $s_0$ to all initial states of $A_1$, by adding internal transitions $(s_0,\ytau,\yait,s)$ for all $s\in I$. 
Finally, we replace any pop transition on an empty stack $(p,c,\bot,q)$ by the new pop transition $(p,c,Z_2,q)$.
Let $\yA_2=\yvpts{A_2}{\{s_0\}}{L_2}{\yGa_2}{\rho_2}$ be the new VPTS after these modifications to $\yA_1$.

Suppose we have $\ytrtf{(s,\bot)}{\ysi}{\yA_1}{(q,\bot)}$, where $s\in I$ is an initial state, and assume that we have $0\leq k\leq \vert \ysi\vert$ pop moves on the empty stack on this run of $\yA_1$. 
Then, a simple induction on $k$ shows that in $\yA_2$ we have 
$\ytrtf{(s_0,\bot)}{a_2^k}{}{(s_0,Z_2^k\bot)}\ytrtf{}{\ytau}{}{(s,Z_2^k\bot)}\ytrtf{}{\ysi}{}{(q,\bot)}$.
Conversely, if we have a run  $\ytrtf{(s_0,\bot)}{\ysi}{\yA_2}{(q,\bot)}$, then $\ysi=a_2^k\ytau\mu$ for some $k\geq 0$, and a simple induction on $k$ shows that in $\yA_1$ we have  $\ytrtf{(s,\bot)}{\mu}{\yA_1}{(q,\bot)}$ where $s$ is an initial state and in this run over $\yA_1$ we made $k$ pop moves on an empty stack.  
\end{description}
After these transformations, we see that $\ytrtf{((t_0,q_0),\bot)}{\ysi}{\yA}{((\yfail,q),\yal\bot)}$, with $(t_0,q_0)$ initial in $\yA$, if and only if 
we have $\ytrtf{(s_0,\bot)}{\mu}{\yA_2}{(f_2,\bot)}$, where $\mu=a_2^k\ytau\ysi\ytau b_1^n$ for some $k\geq 1$ and $m\geq 0$.
Now, from the definition, we have $\ytrut{((t_0,q_0),\bot)}{\eta}{\yP}{((\yfail,q),\yal\bot)}$ if and only if $\ytrtf{((t_0,q_0),\bot)}{\ysi}{\yA}{((\yfail,q),\yal\bot)}$ where $\eta=h_\ytau(\ysi)$.
Thus, $\ytrut{((t_0,q_0),\bot)}{\eta}{\yP}{((\yfail,q),\yal\bot)}$ if and only if  we have a balanced run $\mu$ from $s_0$ to $f_2$ in $\yA_2$ and 
$\eta=h_{\{a_2,b_1,\ytau\}}(\mu)$, that is, $\eta$ is obtained from $\mu$ by erasing from $\mu$ all occurrences of $a_2$, $b_1$ and $\ytau$.
Putting it together, we have: $\yI \yiocolike \yS$ does not hold if and only if $\ytrtf{(s_0,\bot)}{\mu}{\yA_2}{(f_2,\bot)}$.

We have reduced the $\yiocolike$ conformance test to the following problem: given two states $p$ and $q$ of a VPTS, find a string $\ysi$ that induces a balanced run from $p$ to $q$, or indicate that such a string does not exist.
Next, we describe how to solve this problem. 

The following construction was also inspired from ideas in~\cite{FinkelWW97,sipser}.
Let $\yP=\yvpts{Q}{Q_{in}}{L}{\yGa}{\rho}$ be a VPTS
given by an incidence vector $P$, indexed by $Q$, where $P[p]$ points to a list of all transitions $(p,x,Z,q)\in \rho$ where $p$ is the source state. 
We assume that $\yP$ has no pop transitions on the empty stack, that is, of the form $(p,x,\bot,q)$ where $x$ is a pop symbol.
Algorithm~\ref{alg1} shows the pseudo-code.

We will use two vectors of pointers, $In$ and $Out$, both indexed by $Q$, and a queue $V$.
The entry $In[p]$ will point to a list of triples $(s,a,Z)$ corresponding to a transitions $(s,a,Z,p)$ where $a\in L_c$, that is, $p$ as the target state of a push transition.
Likewise, and entry in $Out[p]$ will point to a list of triples $(a,Z,s)$ corresponding to transitions $(p,a,Z,s)$ where  $a\in L_r$, that is, $p$ is the source state of a pop transition.
We will also need a square matrix $R$, indexed by $Q\times Q$, where $R[p,q]$ will contain: (i) $[a,p,q,b]$, or (ii) $[p,c,q]$, or (iii) $[p,s,q]$, or (iv) $0$, where $a\in L_c$, $b\in L_r$, $c\in L_i\cup\{\ytau\}$, and $p$, $q$, $s$ are states.
The general idea is that, when $R[p,q]\neq 0$ then it will code for a string $\ysi$ that induces a balanced run from $p$ to $q$.

We now examine at Algorithm~\ref{alg1}.
Lines 1--10 initialize $V$, $In$, $Out$ and $R$.
At line 6, note that a transition $(p,a,\yait,q)$ immediately induces a balanced run from $p$ to $q$.
At line 10, we collect a simple balanced run from $p$ to $r$ that is induced by a push transition $(p,a,Z,q)$ and a pop transition $(q,b,Z,r)$.
In the main loop, lines 11--19, removing $(p,q)$ from $V$ indicates that we already have a string, say $\ysi$, that induces a balanced run from $p$ to $q$.
At lines 13--14, a string $\mu$ that induces balanced run from $s$ to $p$ is encoded in $R[s,p]$.
Hence, $\mu\ysi$ induces a balanced run from $s$ to $q$.
If we still do not have 
a balanced run from $s$ to $q$,  we can now encode the string $\mu\ysi$ in $R[s,q]$ and move the pair$(s,q)$ to $V$ so that it can be examined later.
Lines 15--16 do the same, but now we encode in $R[p,t]$ a string that induces a balanced run from $p$ to $t$.
At lines 17--19, we search for a push transition $(s,a,Z,p)$ and a matching pop transition $(q,b,Z,t)$ and, when successful, we encode $a\ysi b$ as a string that induces a balanced run from $s$ to $t$. 
The cycle repeats until saturation when $V=null$, or until we find a balanced run from $s_i$ to $s_e$, as requested.
Lines 22--25 list a simple recursive procedure that extracts the string encoded in $R[p,q]\neq 0$. 
\end{proof}

{
	\begin{algorithm}[!ht]
		\fontsize{9}{9}\selectfont
		\SetAlgoLined
		\DontPrintSemicolon
		\SetNoFillComment
		\KwData{\emph{Given}: a vector $P$, where $P[p]$ is a list of all $(p,a,Z,q)\in \rho$;  states $s_i$, $s_e$.}
		
		\KwData{\emph{Assumptions}: $\yP$ has no transition on the empty stack and $s_i\neq s_e$.}

  \KwData{\emph{Uses}: vectors $In$, $Out$ indexed by $Q$; matrix $R$ indexed by $Q\times Q$; queue $V$.}
		\KwResult{Check for a balanced run (BR) from $s_i$ to $s_e$; if there is one find a string that induces it.}
		\SetKwFunction{FAux}{getstring}
		\SetKwProg{Fn}{Function}{:}{}
		
		\tcp*[l]{Initialize $V$, $R$, $In$ and $Out$}
  $V=null$

  \lForAll{$p,q$ in $Q$}{\,\{$R[p,q]=0$\,\}}
		\lForAll{$p$ in $Q$}{  \{\,\,$In[p]=null$; $Out[p]=null$\,\,\}}

		\ForAll{$p$ in $Q$}{ 
				\ForAll{$(p,a,Z,q)$ in $P[p]$}{ 
						\lIf(\tcp*[f]{simp \& inter trans.}){$a\in L_i\cup\{\ytau\}$ and $p\neq q$ and $R[p,q]=0$} {\{$R[p,q]=[p,a,q]$; add $(p,q)$ to $V$\}}
						\lElseIf(\tcp*[f]{pop transition}) {$a\in L_r$ } {add $(a,Z,q)$ to $Out[p]$}
						\uElse(add \mbox{$(p,a,Z)$ to $In[q]$}\tcp*[f]{push transition}){
         \ForAll(\tcp*[f]{BR from push \& pop transitions}){$(q,b,W,r)\in P[q]$}{\lIf{$W=Z$ and $b\in L_r$ and $p\neq r$ and $R[p,r]=0$}{\,\{$R[p,r]=[a,q,q,b]$; add $(p,r)$ to $V$\,\}} }
      	}
   	}
}

  \tcp*[l]{\mbox{ }}		
		\tcp*[l]{Main loop}
		\While{$V\neq null$ and $R[s_i,s_e]=0$}{
			Remove $(p,q)$ from $V$ \tcp*[f]{We have a BR from $p$ to $q$}
		
			\ForAll(\tcp*[f]{new BR from s to q}){$s$ in $Q$}{\lIf{$R[s,p]\neq 0$ and $s\neq q$ and $R[s,q]=0$}{\{\,	$R[s,q]=[s,p,q]$; add $(s,q)$ to $V$\,\} } }
   \ForAll(\tcp*[f]{new BR from p to t}){$t$ in $Q$}{\lIf{$R[q,t]\neq 0$ and $p\neq t$ and $R[p,t]=0$}{\{\,	$R[p,t]=[p,q,t]$; add $(p,t)$ to $V$\,\}} }
			\ForAll{$(s,a,Z)$ in $In[p]$}{
  			\ForAll(\tcp*[f]{push $Z$ from $s$ to $p$, pop $Z$ from $q$ to $t$}){$(b,W,t)$ in $Out[q]$}{
     \lIf{$W=Z$ and $s\neq t$ and $R[s,t]=0$}{\,\{ $R[s,t]=[a,p,q,b]$; add $(s,t)$ to $V$ \}}
     }
			}
		} 

  \tcp*[l]{\mbox{ }}				
		\tcp*[l]{Issue the verdict}
		\lIf{$R[s_i,s_e]=0$} { Print \textsc{There are no balanced runs from $s_i$ to $s_e$}}
		\lElse{ Print \textsc{A string that induces a balanced run from $s_i$ to $s_e$ (between $\vert\,\,\vert)$: $\vert\,$} \FAux($s_i,s_e$)\,\textsc{$\,\vert$}}

  \tcp*[l]{\mbox{ }}				
		\Fn(\tcp*[f]{Print the string}){\FAux{$p,q$}}{ 
    \Switch{$R[p,q]$}{
     \lCase{$[p,a,q]$}{   Print ``$a$'' }
     \lCase{$[p,s,q]$}{  \,\{ 	\FAux($p,s$);  	\FAux($s,q$)\,\} }
     \lCase{$[a,p,q,b]$}{\leIf{$p\neq q$}{ \,\{  	Print ``$a$'';  \FAux($s,r$);	Print ``$b$''\,\}}{ \,\{  	Print ``$a$''; Print ``$b$''\,\}} }
		 }
		}
		\caption{Checking for balanced runs in a VPTS $\yP=\yvpts{Q}{Q_{in}}{L}{\yGa}{\rho}$}\label{alg1}
	\end{algorithm}
}

Now we argue for the correctness of Algorithm~\ref{alg1}. 
\begin{theo}\label{teo:alg1}
Let $\yP=\yvpts{Q}{Q_{in}}{L}{\yGa}{\rho}$ be a VPTS with no transitions  of the form $(p,a,\bot,q)$ in $\rho$.
Also let $s_i$, $s_e\in Q$, with $s_i\neq s_e$.
Suppose $\yP$, $s_i$, $s_e$ are input to Algorithm~\ref{alg1}.
Then it stops and  returns a string $\ysi\in L^\star$ such that  $\ytrtf{(s_i,\bot)}{\ysi}{\yP}{(s_e,\bot)}$, or
it indicates that such a string does not exist.
\end{theo}
\begin{proof}
From the structure of the code, it is simple to show that if $R[s,t]\neq 0$ during an execution of Algorithm~\ref{alg1}, then there it codes for a string that induces a balanced run from $s$ to $t$.
When the algorithm terminates, the simple recursive call \texttt{getstring($s_i,s_e$)} extracts a witness string.

For the other direction, assume that the main loop terminates with $R[s,t]=0$ and we have a string $\ysi$ that induces a balanced run from $s$ to $t$, that is, $\ytrtf{(s,\bot,)}{\ysi}{}{(t,\bot)}$.
A simple argument will show that we can always reach a contradiction by considering two cases, namely, an intermediate configuration $(r,\bot)$ occurs, or does not occur, 
in the run $\ytrtf{(s,\bot,)}{\ysi}{}{(t,\bot)}$.

For details see the complete proof in Appendix~\ref{app:teo:alg1} at  page~\pageref{app:teo:alg1}.
\end{proof}

%% file: iovpts-complexity.tex
Now we examine the complexity of our testing approach.
\begin{theo}\label{te:ioco-complexity}
Let $\yS\in\yiovp{L_I}{L_U}$ be a deterministic specification  with $n_s$ states and $m_s$ transitions, and let $\yI\in\yiovp{L_I}{L_U}$ be an IUT with $n_i$ states and $m_i$ transitions. 
Then there is a procedure, with  worst case asymptotic polynomial time complexity bounded by $\yoh{n_s^3n_i^3+n_s^2m_s^4m_i^2}$, that verifies whether $\yI \yiocolike \yS$. Moreover,  if  $\yI \yiocolike S$ does not hold, the procedure finds an input string that proves this condition.
\end{theo}
\begin{proof}
Write $L=L_I\cup L_U$, $\vert L\vert=\ell$ and let $g_s$, $g_i$ be the number of stack symbols in $\yS$ and $\yI$, respectively.
We now follow the argument in the proof of Theorem~\ref{lemm:vpts-test}.

First, the fault model $\yT$ is constructed in Lemma~\ref{lemm:ioco-complete}.
From Eqs. (\ref{eq:testpurpose1:Lc})--(\ref{eq:testpurpose1:Li}) we see that $n_t=n_s+1$ and $m_t\leq m_s+n_s g_s \ell$, where $n_t$ and $m_t$ are the number of states and transitions in $\yT$, respectively. 
It is easy to see that $\yT$ can be effectively constructed from $\yS$ by an algorithm with worst case time complexity bounded by $\yoh{m_s+n_s g_s \ell}$. 
We can safely assume $m_s\geq g_s$ and $m_s\geq \ell$.
Hence, $m_t$ and $n_t$ can be bounded by $\yoh{n_sm_s^2}$ and $\yoh{n_s}$, respectively, and the worst case time complexity to obtain $\yT$ can also be bounded by  $\yoh{n_sm_s^2}$.

Next, we construct the product $\yP=\yT\times\yI$.
Let $n_p$, $m_p$ and $g_p$ be the number of states, transitions and stack symbols in $\yP$, respectively.
Using Definition~\ref{def:productVPA}, and since $\yT$ has no $\ytau$-moves, we  see that $n_p=n_t n_i$, $m_p\leq m_t m_i+n_tm_i$, and $g_p=g_sg_i$.
As before, a simple algorithm with worst case time complexity bounded by  $\yoh{m_t m_i+n_tm_i}$ can construct $\yP$ given $\yT$ and $\yI$.
Hence, $n_p$, $m_p$ and $g_p$  can be bounded by $\yoh{n_sn_i}$, $\yoh{n_sm_s^2 m_i}$ and $\yoh{m_sm_i}$, respectively, and the worst case time complexity to obtain $\yP$ can  be bounded by  $\yoh{n_sm_s^2m_i}$.

Finally, Theorem~\ref{lemm:vpts-test} requires  that we construct a VPTS $\yA=\yvpts{S_a}{\{s_0\}}{L_a}{\yGa_a}{\rho_a}$ from $\yP$, and guarantees
that $\yI \yiocolike S$ does not hold if and only if we have $\ytrtf{(s_0,\bot)}{\mu}{\yA}{(f,\bot)}$ for some $\mu\in (L_a\cup\{\ytau\})^\star$, where $f$ is a specific state in $S_a$ with $s_0\neq f$.
Moreover, if such is the case, the final steps in the proof of Theorem~\ref{lemm:vpts-test} indicate how to obtain the desired string $\ysi$ that proves that $\yI \yiocolike S$ fails.
Let $n_a=\vert S_a\vert$ and $m_a=\vert \rho_a\vert$. 
From the proof of Theorem~\ref{lemm:vpts-test}, it is easy to get $n_a=n_p+3$ and $m_a=m_p+2n_i+g_p+2$.
Also, a simple procedure, running in worst case time complexity $\yoh{m_a+n_a}$, can be used construct the VPTS $\yA$ given the product $\yP$.
Then, $n_a$ can be bounded by $\yoh{n_s n_i}$, $m_a$ can be bounded by $\yoh{n_sm_s^2m_i}$, and 
the worst case time complexity to construct $\yA$ can be  bounded by $\yoh{n_sm_s^2m_i}$.  

The final step is to submit $\yA$ and the two states $s_0$, $f$ to Algorithm~\ref{alg1}.
Theorem~\ref{teo:alg1} guarantees that  Algorithm~\ref{alg1} correctly produces the desired string $\mu$ or indicates that no such string exists. 

We now derive an asymptotic upper bound on the number of steps required for Algorithm~\ref{alg1} in the worst case.
Clearly, the number of steps for lines 1--3 can be bounded by $\yoh{n_a^2}$.

For each state $p\in S_a$, let $s_p$, $t_q$ be the number of transitions in $\yA$ that have $p$ as a source and $q$ as a target state, respectively.
Thus, $\sum_{p\in Q} s_p\leq m_a$ and $\sum_{p\in Q}t_p\leq m_a$ for all $p\in S_a$.
The number of steps relative to lines 4--10 can be bound by 
$$\sum_{p\in S_a} \big(s_p\sum_{q\in S_a} t_q\big)\leq \sum_{p\in S_a} (s_pm_a)=m_a\sum_{p\in S_a} s_p \leq m_a^2.$$

From the proof of Theorem~\ref{teo:alg1}, we have that each pair $(p,q)$ can enter the queue $V$ at most once.
Hence, a state $p$ will appear in a pair $(p,q)$ removed from $V$ at most $n_a$ times. 
Since the number of steps at each execution of lines 13--14 can be bounded by $\yoh{n_a^2}$, the total effort spent for lines 13--14 is bound by  $\yoh{n_a^3}$.
Likewise for lines 15--16.
Now we bound the total number of steps for lines 17--19.
For each pair $(p,q)$ removed from $V$, the cost relative to lines 17--19 is $\yoh{t_ps_q}$.
Since each pair of  $(p,q)$ can enter $V$ at most once, the total cost is bound by 
$$\sum_{p,q\in Q}t_ps_q=\sum_{p\in Q}t_p\big(\sum_{q\in Q} s_q\big)\leq \sum_{p\in Q}(t_pm_a)=m_a\sum_{p\in Q}t_p\leq m_a^2.$$
Hence, the total number of steps to execute the main loop at lines 12--19 is bound by
$\yoh{n_a^3+m_a^2}$.      
We can now conclude that the total number of steps to execute  Algorithm~\ref{alg1} is bound by $\yoh{n_a^3+m_a^2}$.
Using the  previously computed values, we see that a worst case asymptotic time complexity  for Algorithm~\ref{alg1} is bounded by  $\yoh{n_s^3n_i^3+n_s^2m_s^4m_i^2}$.
Since this bound dominates all the preprocessing steps needed to construct $\yT$, $\yP$ and $\yA$, the overall worst case time complexity of the {\bf ioco-like} checking procedure is $\yoh{n_s^3n_i^3+n_s^2m_s^4m_i^2}$.
\end{proof}
In some practical situations we may assume that the number of stack and alphabet symbols can be assumed to be a constant, for specifications and IUTs models that will be considered. 
In these situations, the number of transitions of the fault model $\yT$ can be bounded by $\yoh{m_s}$, in the proof of Theorem~\ref{te:ioco-complexity}.
As a consequence, the worst case time complexity for Algorithm~\ref{alg1} can be seen to be bounded by $\yoh{n_s^3n_i^3+m_s^2m_i^2}$. 

%% file: iovpts-example.tex
\subsection{A Drink Dispensing Machine}\label{subsec:drink}

Now we want to apply the results from previous sections in a real-world system, a drink dispensing machine. 
because the overall testing procedure is not yet fully implemented, the example has to be somewhat contrived.
In the following subsections we describe the drink machine and its specification IOVPTS model. 
Next we construct the fault model for the given specification and, in the sequel, we test some possible IUT models for {\bf ioco-like} conformance. 

\subsubsection{A Drink Dispensing Machine}

In a typical dispensing machine, a customer puts in  some money and then order the desire beverage. 
After choosing the beverage, the right amount of money will be charged and the machine will dispense the chosen drink. 
If the amount of money was in excess, the customer can ask for the balance.
If the amount of money already in the machine is not enough, the customer has to add more money or the customer can decide to get a full refund.
Usually, real machines accept several payment methods such as coins, cash and credit card. 
In order to ease our modeling we specify only unit coins as payment method in our models. 

The complete IOVPTS specification model $\yS\yS=\yiovpts{S_\yS}{S_{in}}{L_I}{L_U}{\yDe_\yS}{T_\yS}$ is depicted in Figure~\ref{drink}, where
\input{figs/drinkc.tex}
$L_I=\{coi, rch, crd, wtr, tea, cof, deb\}$ is the set of input events and $L_U =\{chg, dwt, dte, dco\}$ is the set of output events. 
The alphabet $L=L_I \cup L_U $ is partitioned into the push events $L_c =\{coi\}$, pop events
$L_r=\{crd, chg, deb, dwt, dte, dco\}$ and simple events $L_i =\{rch, wtr, tea, cof\}$.
We have split state $s_1$ to make the figure clearer.
Recall Remark~\ref{rem:figure} for the notation.
The underlying VPTS is $\yA_\yS=\yvpts{S_\yS}{S_{in}}{L}{\yDe_\yS}{T_\yS}$. 

The system starts at state $s_1$ where the customer can either insert coins  into the machine --- event labeled $coi$ ---, request his change --- event labeled $rch$ ---, or ask for a drink, namely, label $wtr$ for water, label $tea$ for tea, or label $cof$ for coffee. 
Inserting coins is represented by the self-loop labeled  $coi/\ypush{C}$ at state $s_1$.
Note that the pushdown stack keeps track of the number of coins inserted into the machine. 
At state $s_1$ the customer can also request a refund, or the remaining  change, by activating the $rch$ event. 
The machine will then return the correct balance via the pop self-loop  $crd/\ypop{C}$ at state $s_3$. 

The customer orders a drink by pushing the bottom for water, tea, or coffee, moving  the machine to states $s_2$, $s_4$ and $s_6$, respectively. 
The price associated to water is one coin, for tea it is two coins and for coffee it is three coins. 
When the customer asks for water, the pop transition $dwt/\ypop{C}$ is taken, returning to state $s_1$, and the correct charge is applied subtracting one coin from the total amount.
The event $dwt$ indicates that water have been dispensed.
However, if not enough coins have been inserted, the transition from state $s_2$ back to state $s_1$ is blocked.
The customer can proceed by inserting more coins using the self-loop  $coi/\ypush{C}$ at state $s_2$.
For simplicity, once the customer has made a commit to order some of the beverages, the machine will wait until enough coins have been inserted to pay for the chosen drink.

The behavior when ordering tea or coffee follows similar paths. 

Figure~\ref{tp-drinkc} depicts the fault model $\yT$  that is constructed for the specification $\yS$ using Lemma~\ref{lemm:ioco-complete}. 
\input{figs/tp-drinkc}
In the figure, We have split the $\yfail$ in order keep the figure unclutered. 
For the same reason, the sets $D_\ell$, for $\ell\in\{a,r,w,t,c\}$, collect the label of several transitions, as indicated in the figure caption.
Note that $L_U\cap L_c=\yemp=L_U\cap L_i$, so that the only symbols we need to check for transitions going into the $\yfail$ state are over the symbols of $L_U\cap L_r=\{chg,dwt,dte,dco\}$ together with all stack symbols $\{C,\bot\}$.
For example, the set $D_a$ denotes transitions to the $\yfail$ state 
with pairs $(x,y)$ for all $x\in L_U\cap L_r$ and all $y\in \{C,\bot\}$.

\subsubsection{Testing Some Implementations to the Drink Dispensing Machine}

In this subsection, we examine some example implementations, and test
them for {\bf ioco-like} conformance against the specification $\yS$ depicted in Figure~\ref{drink}, and whose fault model $\yT$ is shown in Figure~\ref{tp-drinkc}.

\subsubsection*{Missing change}

Our first example is the IUT $\yI_a$, depicted in Figure~\ref{IUTa}. 
Notice that in this IUT we have $(s_3,chg/\ypop{C},s_1)$ instead of $(s_3,chg/\bot,s_1)$ as in the specification model. 
This error may cause the customer to receive less than the correct balance for the money inserted in machine.   
\input{figs/IUT-buga}

Consider the sequence of events $\ysi=coi\, coi\, coi\, rch\, crd\, crd\, chg$.
Note that in the IUT $\yI_a$, according to this sequence of events, the customer has inserted three coins and then requested the remaining change, but only two coins get credited.
However, on the specification $\yS$, all three coins would be credited.
It is easy to see that $\ysi$ leads $\yT$ to the $\yfail$ state, while $\yI_a$ reaches state $s_1$.
According to Lemma~\ref{lemm:ioco-complete}, $\yI_a \yiocolike \yS$ does not hold.

More formally, write $\ysi=\mu\, chg$, where $\mu=coi\, coi\, coi\, rch\, crd\, crd$.
Then, from the figures, it follows easily that $\ytrt{(s_1,\bot)}{\mu}{}{(s_3,C\bot)}\ytrt{}{chg}{}{}$ in $\yI_a$. Since $chg\in L_U$, we have $chg\in \yout((s_1,\bot) \yafter \ysi)$ in $\yI_a$.
However, in $\yS$ we also have $\ytrt{(s_1,\bot)}{\mu}{}{(s_3,C\bot)}$, but then it is not the case that $\ytrt{(s_3,C\bot)}{chg}{}{}$ in $\yS$, so that 
 $chg\not\in \yout((s_1,\bot) \yafter \ysi)$ in $\yS$.
From Definition~\ref{def:out-after}, it follows that $\yI_a \yiocolike \yS$ does not hold.
Observe that, in the product $\yT\times \yI_a$, we get the corresponding run
$$ \ytrt{((s_1,s_1),\bot)}{\mu}{}{((s_3,s_3),C\bot)}\ytrt{}{chg}{}{((\yfail,s_1),\bot)},$$
as precognized by Lemma~\ref{lemm:ioco-complete}.

\subsubsection*{Problems with the prices for drinks}

Now we test the IUT $\yI_b$, depicted in Figure~\ref{IUTb}, where coffee is wrongly charged at two coins only. 
Here we notice that the state $s_8$ is missing in the IUT, and so 
we have a self-loop transition $(s_7,deb/\ypop{C},s_7)$ instead of $(s_7,deb/\ypop{C},s_8)$ as in the specification model. 
In this case a fault can occur and the machine may charge less for a cup of coffee.  
\input{figs/IUT-bugb}\\

Consider the sequence of events $\eta=\mu\,  dco$, where $\mu=coi\, coi\, cof\, deb$. 
The customer has inserted only two coins and ordered coffee, and still the machine may deliver a cup of coffee. 
It is easy to see that $\eta$ leads $\yT$ to the $\yfail$ state, while $\yI_b$ reaches state $s_1$.
Again, from Lemma~\ref{lemm:ioco-complete} we obtain that $\yI_b \yiocolike \yS$ does not hold.

As in the previous example, we have $dco\in L_U$ and $dco\in \yout((s_1,\bot) \yafter \mu)$ in $\yI_b$, while $dco \not\in \yout((s_1,\bot) \yafter \mu)$ in $\yS$.
From Definition~\ref{def:out-after}, we can declare that $\yI_b \yiocolike \yS$ does not hold.
In the product $\yT\times \yI_b$, we get the corresponding run
$$ \ytrt{((s_1,s_1),\bot)}{\eta}{}{}{((\yfail,s_1),\bot)}.$$

Notice that, in this same IUT, coffee could be charged more than three coins, \emph{i.e.}, the machine may subtract more than  three coins before dispensing a cup of coffee, when the user has inserted more than three coins before asking for the cup of coffee. 

\subsubsection*{Dispensing an unwanted drink}

In the IUT $\yI_c$ depicted in Figure~\ref{IUTc}, the machine can dispense coffee even when the customer has ordered tea and, moreover, the customer can get a cup of coffee for the price of only two coins. 
In  $\yI_c$ we have $(s_5,coi/\ypush{C},s_8)$ instead of $(s_5,coi/\ypush{C},s_5)$ as in the specification model. 
\input{figs/IUT-bugc}
Consider the sequence $\ysi=\mu\, dco$, where $\mu=coi\, coi\, tea\, deb\, coi$.
The customer has inserted two coins and ordered tea. 
After the machine charges one coin, the customer decides to insert one more coin.
Then machine dispenses  a cup of coffee instead, and charges only two coins for it. 
We can see in the figure that $\ysi$ leads $\yT$ to the $\yfail$ state, while $\yI_c$ reaches state $s_1$.
According to Lemma~\ref{lemm:ioco-complete} then $\yI_c \yiocolike \yS$ does not hold.

Again 
we have $dco\in \yout((s_1,\bot) \yafter \mu)$ in $\yI_c$, and $dco\in L_U$. 
On the other hand, 
 we do not have  $\ytrt{(s_5,C\bot)}{dco}{}{}$ in $\yS$, and so $dco \not\in \yout((s_1,\bot) \yafter \mu)$ in $\yS$.
Using Definition~\ref{def:out-after} we see that $\yI_c \yiocolike \yS$ does not hold.
In the product $\yT\times \yI_b$ we observe the corresponding run
$$ \ytrt{((s_1,s_1),\bot)}{\ysi}{}{}{((\yfail,s_1),\bot)}.$$

\subsubsection*{An faulty implementation that conforms to the specification}

Lastly, we turn to IUT $\yI_d$, as depicted in Figure~\ref{IUTd}.
Note that $\yI_d$ is not isomorphic to the specification model, since it has the extra self-loop $(s_5,deb/\ypop{C},s_5)$.
This allows the machine to subtract any number of extra coins after the customer has orderd a drink, given that more than enough coins have been inserted. 
\input{figs/IUT-nobug}
Consider the sequence of events $coi\, coi\, coi\, tea\, deb\, deb\,dte\,rch\,chg$, signaling  that
the customer initially has inserted three coins, then decided to order tea.
According to the IUT $\yI_d$, however, when requesting the remaining change the customer get no coins back, and the net effect was that the customer was charged three coins for a cup of tea.  
Even in face of that, we show below that $\yI_d$  does conform to the specification $\yS$.

Recall the specification $\yS$ shown in Figure~\ref{drink}.
We note that $\yI_d$ differs from $\yS$ only by the extra transition at state $s_5$.
Further, for each symbol $x \in L_U$ and state $s_i$, there is at most one transition out of $s_i$ on $x$, both  on $\yI_d$ and on $\yS$.
Reasoning  more formally, 
it is easy to see that for any sequence of events $\ysi$ a simple induction on $\vert\ysi\vert\geq 0$ shows that if we have  $(s_i,\yal\bot) \in (s_1,\bot) \yafter \ysi$ in $\yS$  and $(s_j,\ybe\bot) \in (s_1,\bot) \yafter \ysi$ in $\yI_d$, then $i=j$ and $\yal=\ybe$. 
According to Definition~\ref{def:out-after}, if $\yI_d \yiocolike \yS$ did not hold, we would need $\ell\in L_U$ and a sequence $\ysi$ such that $\ell \in \yout((s_i,\yal\bot))$ where $(s_i,\yal\bot) \in (s_1,\bot) \yafter \ysi$ in $\yI_d$ and $\ell \not \in \yout( (s_j,\ybe\bot))$ where $(s_j,\ybe\bot) \in (s_1,\bot) \yafter \ysi$ in $\yS$.
Since the transitions of $\yS$ and $\yI_d$ are identical, except at state $s_5$, we conclude that $s_i=s_j=s_5 $ and $\ell=deb$. 
Because $deb\not\in L_U$ we reached a contradiction, and must conclude that $\yI_d \yiocolike \yS$ does hold.

We note that $deb\not\in L_U$ was crucial to the preceding argument. 
In fact, if we move $deb$ from $L_I$ to $L_U$, then we clearly would get that $\yI_d \yiocolike \yS$ fails.
This is because the nature of the {\bf ioco-like} relation checks only that the IUT  \emph{may not} emit any \emph{output symbol} that  was not enabled in the specification, after they both experience any sequence of events that runs on the specification.
On the other hand, the definition of the {\bf ioco-like} relation says nothing about \emph{input symbols} that may be emitted by the IUT and the specification after a common run in both models.

%% file: figs/drinkc.tex
\begin{figure}[htb]
\center

\begin{tikzpicture}[node distance=1cm, auto,scale=.6,inner sep=1pt]

  \node[ initial by arrow, initial text={}, punkt] (s1) {$s_1$};
  \node[punkt, dashed, inner sep=1pt,right=2.0cm of s1] (s1b) {$s_1$};
  \node[punkt, inner sep=1pt, above=2cm of s1b] (s2) {$s_2$};

  \node[punkt, inner sep=1pt,right=2.0cm of s1b] (s4) {$s_4$};  
   \node[punkt, inner sep=1pt,right=1.5cm of s4] (s5) {$s_5$};  

  \node[punkt, inner sep=1pt,below=2cm of s1b] (s6) {$s_6$};  
  \node[punkt, inner sep=1pt,right=1.5cm of s6] (s7) {$s_7$};  
  \node[punkt, inner sep=1pt,right=1.5cm of s7] (s8) {$s_8$};  

  \node[punkt, inner sep=1pt,below=2cm of s1] (s9) {$s_3$};  
      

\path (s1)    edge [loop above=60] node   {$coi/\ypush{C}$} (s1);
\path (s1b)    edge [ pil, bend right=30]
                	node[pil,right]{$wtr$} (s2);
\path (s2)    edge [ pil, bend right=30]
                	node[pil,left]{$dwt/\ypop{C}$} (s1b);

\path (s2)    edge [loop right=60] node   {$coi/\ypush{C}$} (s2);

\path (s1b)    edge [ pil, left=50]
                	node[pil,below]{$tea$} (s4);

\path (s4)    edge [ pil, left=50]
                	node[pil,below]{$deb/\ypop{C}$} (s5);
\path (s4)    edge [loop below=60] node[pil, below right =-0.25]   {$coi/\ypush{C}$} (s4);

\path (s5)    edge [ pil, bend right=30]
                	node[pil,above]{$dte/\ypop{C}$} (s1b);
\path (s5)    edge [loop above=60] node   {$coi/\ypush{C}$} (s5);

\path (s1b)    edge [ pil, left=50]
                	node[pil,left]{$cof$} (s6);
\path (s6)    edge [ pil, left=50]
                	node[pil,below]{$deb/\ypop{C}$}(s7);
\path (s7)    edge [ pil, left=50]
                	node[pil,below]{$deb/\ypop{C}$} (s8);
\path (s8)    edge [ pil, left=50]
                	node[pil,below left]{$dco/\ypop{C}$} (s1b);

\path (s6)    edge [loop below=60] node {$coi/\ypush{C}$} (s6);
\path (s7)    edge [loop below=60] node {$coi/\ypush{C}$} (s7);
\path (s8)    edge [loop below=60] node {$coi/\ypush{C}$} (s8);

\path (s1)    edge [ pil, bend right=30]
node[pil, left]{$rch$} (s9);
\path (s9)    edge [loop below] node   {$crd/\ypop{C}$} (s9);
\path (s9)    edge [ pil, bend right=30]
node[pil,right]{$chg/\bot$} (s1);

\end{tikzpicture}
\caption{A drink dispensing machine $\yS$.}
\label{drink}
\end{figure}

%% file: figs/tp-drinkc.tex
\begin{figure}[htb]
\centering

\begin{tikzpicture}[node distance=1cm, auto,scale=.6,inner sep=1pt]

  \node[ initial by arrow, initial text={}, punkt] (s1) {$s_1$};
  \node[punkt, dashed, inner sep=1pt,right=2.0cm of s1] (s1b) {$s_1$};
  \node[punkt, inner sep=1pt, above=2cm of s1b] (s2) {$s_2$};

  \node[punkt, inner sep=1pt,right=2.0cm of s1b] (s4) {$s_4$};  
   \node[punkt, inner sep=1pt,right=1.5cm of s4] (s5) {$s_5$};  

  \node[punkt, inner sep=1pt,below=2cm of s1b] (s6) {$s_6$};  
  \node[punkt, inner sep=1pt,right=1.5cm of s6] (s7) {$s_7$};  
  \node[punkt, inner sep=1pt,right=1.5cm of s7] (s8) {$s_8$};  

  \node[punkt, inner sep=1pt,below=2cm of s1] (s9) {$s_3$};  
      

  \node[punkt, densely dotted, inner sep=1pt,below left=3cm of s1] (fail1) {$\yfail$};
  \node[punkt, densely dotted, inner sep=1pt,above=2cm of s4] (fail2) {$\yfail$};
  \node[punkt, densely dotted, inner sep=1pt,below right =1.5cm of s5] (fail3) {$\yfail$};
  \node[punkt, densely dotted, inner sep=1pt,below=2cm of s6] (fail4) {$\yfail$};  
  
      

\path (s1)    edge [loop above=60] node   {$coi/\ypush{C}$} (s1);
\path (s1b)    edge [ pil, bend right=30]
                	node[pil,right]{$wtr$} (s2);
\path (s2)    edge [ pil, bend right=30]
                	node[pil,left]{$dwt/\ypop{C}$} (s1b);

\path (s2)    edge [loop left=60] node   {$coi/\ypush{C}$} (s2);

\path (s1b)    edge [ pil, left=50]
                	node[pil,below]{$tea$} (s4);

\path (s4)    edge [ pil, left=50]
                	node[pil,below]{$deb/\ypop{C}$} (s5);
\path (s4)    edge [loop below=60] node[pil, below right =-0.25]   {$coi/\ypush{C}$} (s4);

\path (s5)    edge [ pil, bend right=30]
                	node[pil,above]{$dte/\ypop{C}$} (s1b);
\path (s5)    edge [loop above=60] node   {$coi/\ypush{C}$} (s5);

\path (s1b)    edge [ pil, left=50]
                	node[pil,left]{$cof$} (s6);
\path (s6)    edge [ pil, left=50]
                	node[pil,below]{$deb/\ypop{C}$}(s7);
\path (s7)    edge [ pil, left=50]
                	node[pil,below]{$deb/\ypop{C}$} (s8);
\path (s8)    edge [ pil, left=50]
                	node[pil,below left =-0.15]{$dco/\ypop{C}$} (s1b);

\path (s6)    edge [loop left=60] node {$coi/\ypush{C}$} (s6);
\path (s7)    edge [loop below=60] node[pil, below=-0.15] {$coi/\ypush{C}$} (s7);
\path (s8)    edge [loop below=60] node[pil, below=-0.15] {$coi/\ypush{C}$} (s8);

\path (s1)    edge [ pil, bend right=30]
node[pil, left]{$rch$} (s9);
\path (s9)    edge [loop below] node[pil, below=-0.15]   {$crd/\ypop{C}$} (s9);
\path (s9)    edge [ pil, bend right=30]
node[pil,right]{$chg/\bot$} (s1);


\path (s1)    edge [ pil, bend right=30]
node[pil,red,above left]{$D_a$} (fail1);
\path (s9)    edge [ pil, bend left=30]
node[pil,red,below]{$D_r$} (fail1);

\path (s2)    edge [ pil, right=30]
node[pil,red,above]{$D_w$} (fail2);
\path (s4)    edge [ pil, bend right=30]
node[pil,red,above right]{$D_a$} (fail2);

\path (s5)    edge [ pil, bend left=30]
node[pil,red,right]{$D_t$} (fail3);
\path (s8)    edge [ pil, bend right=30]
node[pil,red,below]{$D_c$} (fail3);

\path (s6)    edge [ pil, right=30]
node[pil,red,left]{$D_a$} (fail4);
\path (s7)    edge [ pil, right=30]
node[pil,red,below right]{$D_a$} (fail4);


\end{tikzpicture}

\vspace*{2ex}
{\fontsize{9}{9}
{\raggedright
\hspace*{15ex}
$D_a=\{dwt/\ypop{C}, dte/\ypop{C}, dco/\ypop{C}, chg/\ypop{C}, dwt/\bot, dte/\bot, dco/\bot, chg/\bot \}$\\
\hspace*{15ex}
$D_r=\{dwt/\ypop{C}, dte/\ypop{C}, dco/\ypop{C}, chg/\ypop{C}, dwt/\bot, dte/\bot, dco/\bot \}$\\
\hspace*{15ex}
$D_w=\{dte/\ypop{C}, dco/\ypop{C}, chg/\ypop{C}, dwt/\bot, dte/\bot, dco/\bot, chg/\bot \}$\\
\hspace*{15ex}
$D_t=\{dwt/\ypop{C}, dco/\ypop{C}, chg/\ypop{C}, dwt/\bot, dte/\bot, dco/\bot, chg/\bot \}$\\
\hspace*{15ex}
$D_c=\{dwt/\ypop{C}, dte/\ypop{C}, chg/\ypop{C}, dwt/\bot, dte/\bot, dco/\bot, chg/\bot \}$\\
}}

\caption{The fault model $\yT$ for $\yS$.}
\label{tp-drinkc}
\end{figure}

%% file: figs/IUT-buga.tex
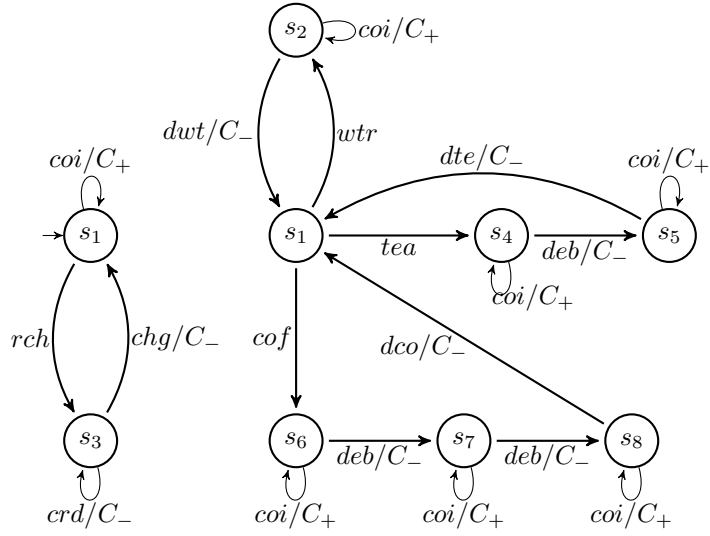
\begin{figure}[htb]
\center

\begin{tikzpicture}[node distance=1cm, auto,scale=.6,inner sep=1pt]

  \node[ initial by arrow, initial text={}, punkt] (s1) {$s_1$};
\node[punkt, inner sep=1pt,right=2.0cm of s1] (s1b) {$s_1$};
\node[punkt, inner sep=1pt, above=2cm of s1b] (s2) {$s_2$};

  \node[punkt, inner sep=1pt,right=2.0cm of s1b] (s4) {$s_4$};  
   \node[punkt, inner sep=1pt,right=1.5cm of s4] (s5) {$s_5$};  

  \node[punkt, inner sep=1pt,below=2cm of s1b] (s6) {$s_6$};  
  \node[punkt, inner sep=1pt,right=1.5cm of s6] (s7) {$s_7$};  
  \node[punkt, inner sep=1pt,right=1.5cm of s7] (s8) {$s_8$};  

  \node[punkt, inner sep=1pt,below=2cm of s1] (s9) {$s_3$};  
      

\path (s1)    edge [loop above=60] node   {$coi/\ypush{C}$} (s1);
\path (s1b)    edge [ pil, bend right=30]
                	node[pil,right]{$wtr$} (s2);
\path (s2)    edge [ pil, bend right=30]
                	node[pil,left]{$dwt/\ypop{C}$} (s1b);

\path (s2)    edge [loop right=60] node   {$coi/\ypush{C}$} (s2);

\path (s1b)    edge [ pil, left=50]
                	node[pil,below]{$tea$} (s4);

\path (s4)    edge [ pil, left=50]
                	node[pil,below]{$deb/\ypop{C}$} (s5);
\path (s4)    edge [loop below=60] node[pil, below right =-0.25]   {$coi/\ypush{C}$} (s4);

\path (s5)    edge [ pil, bend right=30]
                	node[pil,above]{$dte/\ypop{C}$} (s1b);
\path (s5)    edge [loop above=60] node   {$coi/\ypush{C}$} (s5);

\path (s1b)    edge [ pil, left=50]
                	node[pil,left]{$cof$} (s6);
\path (s6)    edge [ pil, left=50]
                	node[pil,below]{$deb/\ypop{C}$}(s7);
\path (s7)    edge [ pil, left=50]
                	node[pil,below]{$deb/\ypop{C}$} (s8);
\path (s8)    edge [ pil, left=50]
                	node[pil,below left=-0.15]{$dco/\ypop{C}$} (s1b);

\path (s6)    edge [loop below=60] node {$coi/\ypush{C}$} (s6);
\path (s7)    edge [loop below=60] node {$coi/\ypush{C}$} (s7);
\path (s8)    edge [loop below=60] node {$coi/\ypush{C}$} (s8);

\path (s1)    edge [ pil, bend right=30]
node[pil, left]{$rch$} (s9);
\path (s9)    edge [loop below] node   {$crd/\ypop{C}$} (s9);
\path (s9)    edge [ pil, bend right=30]
node[pil,right]{$chg/\ypop{C}$} (s1);

\end{tikzpicture}
\caption{An IUT $\yI_a$ giving wrong change.}
\label{IUTa}
\end{figure}

%% file: figs/IUT-bugb.tex
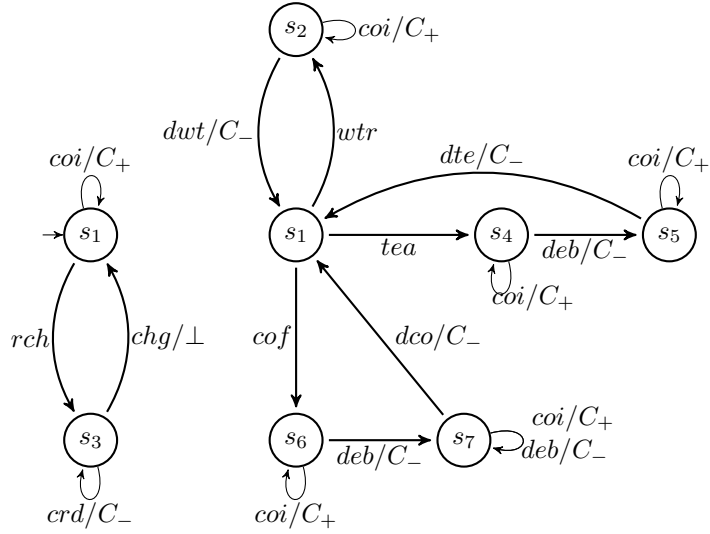
\begin{figure}[htb]
\center

\begin{tikzpicture}[node distance=1cm, auto,scale=.6,inner sep=1pt]

  \node[ initial by arrow, initial text={}, punkt] (s1) {$s_1$};
  \node[punkt, inner sep=1pt,right=2.0cm of s1] (s1b) {$s_1$};
  \node[punkt, inner sep=1pt, above=2cm of s1b] (s2) {$s_2$};

  \node[punkt, inner sep=1pt,right=2.0cm of s1b] (s4) {$s_4$};  
   \node[punkt, inner sep=1pt,right=1.5cm of s4] (s5) {$s_5$};  

  \node[punkt, inner sep=1pt,below=2cm of s1b] (s6) {$s_6$};  
  \node[punkt, inner sep=1pt,right=1.5cm of s6] (s7) {$s_7$};  

  \node[punkt, inner sep=1pt,below=2cm of s1] (s9) {$s_3$};  
      

\path (s1)    edge [loop above=60] node   {$coi/\ypush{C}$} (s1);
\path (s1b)    edge [ pil, bend right=30]
                	node[pil,right]{$wtr$} (s2);
\path (s2)    edge [ pil, bend right=30]
                	node[pil,left]{$dwt/\ypop{C}$} (s1b);

\path (s2)    edge [loop right=60] node   {$coi/\ypush{C}$} (s2);

\path (s1b)    edge [ pil, left=50]
                	node[pil,below]{$tea$} (s4);

\path (s4)    edge [ pil, left=50]
                	node[pil,below]{$deb/\ypop{C}$} (s5);
\path (s4)    edge [loop below=60] node[pil, below right =-0.25]   {$coi/\ypush{C}$} (s4);

\path (s5)    edge [ pil, bend right=30]
                	node[pil,above]{$dte/\ypop{C}$} (s1b);
\path (s5)    edge [loop above=60] node   {$coi/\ypush{C}$} (s5);

\path (s1b)    edge [ pil, left=50]
                	node[pil,left]{$cof$} (s6);
\path (s6)    edge [ pil, left=50]
                	node[pil,below]{$deb/\ypop{C}$}(s7);
\path (s7)    edge [ pil, right=20]
node[pil,  right = 0.15]{$dco/\ypop{C}$} (s1b);

\path (s6)    edge [loop below=60] node {$coi/\ypush{C}$} (s6);
\path (s7)    edge [loop right=60] node[pil, above right = 0.1] 
{$coi/\ypush{C}$} (s7);
\path (s7)    edge [loop right=60] node[pil, below right= -0.1] 
{$deb/\ypop{C}$} (s7);

\path (s1)    edge [ pil, bend right=30]
node[pil, left]{$rch$} (s9);
\path (s9)    edge [loop below] node   {$crd/\ypop{C}$} (s9);
\path (s9)    edge [ pil, bend right=30]
node[pil,right]{$chg/\bot$} (s1);

\end{tikzpicture}
\caption{An IUT $\yI_b$ charges wrong. }
\label{IUTb}
\end{figure}

%% file: figs/IUT-bugc.tex
\begin{figure}[htb]
\center

\begin{tikzpicture}[node distance=1cm, auto,scale=.6,inner sep=1pt]

  \node[ initial by arrow, initial text={}, punkt] (s1) {$s_1$};
  \node[punkt, inner sep=1pt,right=2.0cm of s1] (s1b) {$s_1$};
  \node[punkt, inner sep=1pt, above=2cm of s1b] (s2) {$s_2$};

  \node[punkt, inner sep=1pt,right=2.0cm of s1b] (s4) {$s_4$};  
   \node[punkt, inner sep=1pt,right=1.5cm of s4] (s5) {$s_5$};  

  \node[punkt, inner sep=1pt,below=2cm of s1b] (s6) {$s_6$};  
  \node[punkt, inner sep=1pt,right=1.5cm of s6] (s7) {$s_7$};  
  \node[punkt, inner sep=1pt,right=1.5cm of s7] (s8) {$s_8$};  

  \node[punkt, inner sep=1pt,below=2cm of s1] (s9) {$s_3$};  
      

\path (s1)    edge [loop above=60] node   {$coi/\ypush{C}$} (s1);
\path (s1b)    edge [ pil, bend right=30]
                	node[pil,right]{$wtr$} (s2);
\path (s2)    edge [ pil, bend right=30]
                	node[pil,left]{$dwt/\ypop{C}$} (s1b);

\path (s2)    edge [loop right=60] node   {$coi/\ypush{C}$} (s2);

\path (s1b)    edge [ pil, left=50]
                	node[pil,below]{$tea$} (s4);

\path (s4)    edge [ pil, left=50]
                	node[pil,below]{$deb/\ypop{C}$} (s5);
\path (s4)    edge [loop below=60] node[pil, below right =-0.25]   {$coi/\ypush{C}$} (s4);

\path (s5)    edge [ pil, bend right=30]
                	node[pil,above]{$dte/\ypop{C}$} (s1b);
\path (s5)    edge [ pil, bend left=30]
node[pil, above left]{$coi/\ypush{C}$} (s8);

\path (s1b)    edge [ pil, left=50]
                	node[pil,left]{$cof$} (s6);
\path (s6)    edge [ pil, left=50]
                	node[pil,below]{$deb/\ypop{C}$}(s7);
\path (s7)    edge [ pil, left=50]
                	node[pil,below]{$deb/\ypop{C}$} (s8);
\path (s8)    edge [ pil, left=50]
                	node[pil,below left =-0.15]{$dco/\ypop{C}$} (s1b);

\path (s6)    edge [loop below=60] node {$coi/\ypush{C}$} (s6);
\path (s7)    edge [loop below=60] node {$coi/\ypush{C}$} (s7);
\path (s8)    edge [loop below=60] node {$coi/\ypush{C}$} (s8);

\path (s1)    edge [ pil, bend right=30]
node[pil, left]{$rch$} (s9);
\path (s9)    edge [loop below] node   {$crd/\ypop{C}$} (s9);
\path (s9)    edge [ pil, bend right=30]
node[pil,right]{$chg/\bot$} (s1);

\end{tikzpicture}
\caption{An IUT $\yI_c$ that delivers the wrong drink.}
\label{IUTc}
\end{figure}
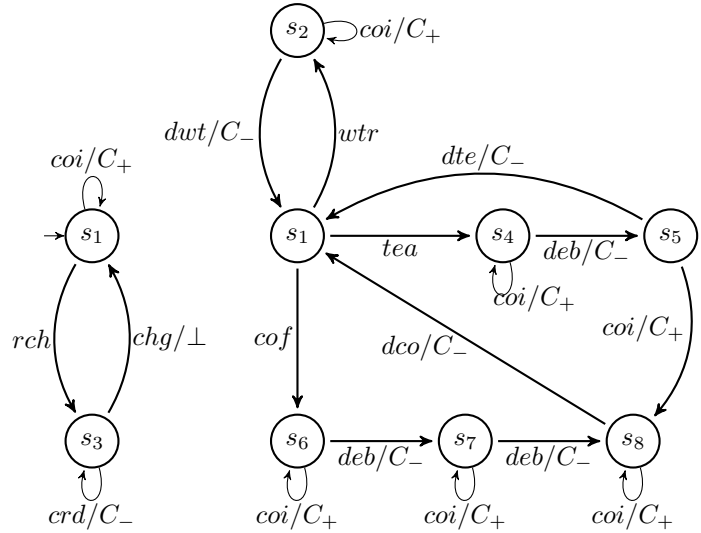

%% file: figs/IUT-nobug.tex
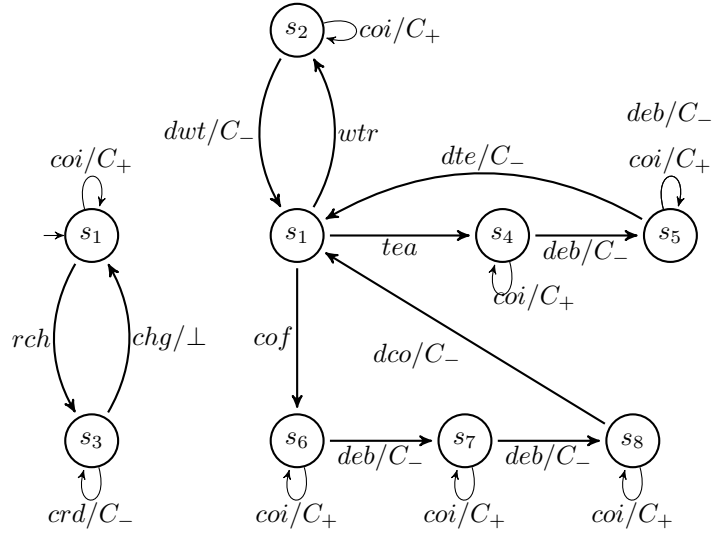
\begin{figure}[htb]
\center

\begin{tikzpicture}[node distance=1cm, auto,scale=.6,inner sep=1pt]

  \node[ initial by arrow, initial text={}, punkt] (s1) {$s_1$};
  \node[punkt, inner sep=1pt,right=2.0cm of s1] (s1b) {$s_1$};
  \node[punkt, inner sep=1pt, above=2cm of s1b] (s2) {$s_2$};

  \node[punkt, inner sep=1pt,right=2.0cm of s1b] (s4) {$s_4$};  
   \node[punkt, inner sep=1pt,right=1.5cm of s4] (s5) {$s_5$};  

  \node[punkt, inner sep=1pt,below=2cm of s1b] (s6) {$s_6$};  
  \node[punkt, inner sep=1pt,right=1.5cm of s6] (s7) {$s_7$};  
  \node[punkt, inner sep=1pt,right=1.5cm of s7] (s8) {$s_8$};  

  \node[punkt, inner sep=1pt,below=2cm of s1] (s9) {$s_3$};  
      

\path (s1)    edge [loop above=60] node   {$coi/\ypush{C}$} (s1);
\path (s1b)    edge [ pil, bend right=30]
                	node[pil,right]{$wtr$} (s2);
\path (s2)    edge [ pil, bend right=30]
                	node[pil,left]{$dwt/\ypop{C}$} (s1b);

\path (s2)    edge [loop right=60] node   {$coi/\ypush{C}$} (s2);

\path (s1b)    edge [ pil, left=50]
                	node[pil,below]{$tea$} (s4);

\path (s4)    edge [ pil, left=50]
                	node[pil,below]{$deb/\ypop{C}$} (s5);
\path (s4)    edge [loop below=60] node[pil, below right =-0.25]   {$coi/\ypush{C}$} (s4);

\path (s5)    edge [ pil, bend right=30]
                	node[pil,above]{$dte/\ypop{C}$} (s1b);
\path (s5)    edge [loop above=60] node   {$coi/\ypush{C}$} (s5);
\path (s5)    edge [loop above=60] node[pil, above =0.6]     {$deb/\ypop{C}$} (s5);

\path (s1b)    edge [ pil, left=50]
                	node[pil,left]{$cof$} (s6);
\path (s6)    edge [ pil, left=50]
                	node[pil,below]{$deb/\ypop{C}$}(s7);
\path (s7)    edge [ pil, left=50]
                	node[pil,below]{$deb/\ypop{C}$} (s8);
\path (s8)    edge [ pil, left=50]
                	node[pil,below left]{$dco/\ypop{C}$} (s1b);

\path (s6)    edge [loop below=60] node {$coi/\ypush{C}$} (s6);
\path (s7)    edge [loop below=60] node {$coi/\ypush{C}$} (s7);
\path (s8)    edge [loop below=60] node {$coi/\ypush{C}$} (s8);

\path (s1)    edge [ pil, bend right=30]
node[pil, left]{$rch$} (s9);
\path (s9)    edge [loop below] node   {$crd/\ypop{C}$} (s9);
\path (s9)    edge [ pil, bend right=30]
node[pil,right]{$chg/\bot$} (s1);

\end{tikzpicture}
\caption{An IUT $\yI_d \yiocolike \yS$.}
\label{IUTd}
\end{figure}

%% file: conclusions.tex
\section{Concluding Remarks}\label{sec:conclusion}

Testing conformance of reactive systems implementations is usually a hard task, due to the intrinsic nature of reactive systems, which allows for the asynchronous interactions of messages with an external environment. 
In such situations, rigorous approaches capable of generating test suites for these models, are indicated.
Several methods have been developed in order to generate test suites for reactive systems that can be modeled by formal systems which  have access only to a finite memory, or a finite number of states.

Here we study a more powerful class of reactive systems, those can make use of a potentially infinite memory, the form of a pushdown stack. 
We defined a corresponding more general conformance relation for these systems.
We also showed the existence, and how to generate, test suites that can be used to verify whether this more general conformance relation holds between a given specification and a proposed implementation.
Further, we argued that these test suites can be generated in polynomial time, and that they are also complete, that is, 
a given implementation does conform to a given specification if and only if the test suite verdict indicates that this is the case.  

For such reactive systems with a pushdown stack, we also defined a conformance relation, called {\bf ioco-like}, that embodies the same idea as the classic {\bf ioco} conformance relation that is used for testing simpler systems, namely, those cannot make use of a potentially infinite memory.
Essentially, these conformance relations say that the implementation can only emit an output signal that is already present in the  specification, after a common exchange of symbols has taken place in both models.
We showed that  {\bf ioco-like} conformance is a special case of the more general conformance relation we treat in this work.
We also developed, and proved correct, a polynomial time algorithm that can be used to test for {\bf ioco-like} conformance between an implementation and a given specification. 
In common practical situations the algorithm exhibits an asymptotic worst case time complexity that can be bounded by $\yoh{n^3+m^3}$ where $n$ and $m$ 
are proportional to the cube of the product of the number of states and transitions, respectively, present in the implementation and in the specification.  
%

%% file: apendice.tex

\appendix 

\section{Constructions and Proofs}\label{sec:apendice}

In this appendix we present complete constructions and detailed proofs. 

\input{sec2-apendice}
\input{sec3-apendice}
\input{sec5-apendice}

%% file: sec2-apendice.tex

\subsection{Proof of Proposition~\ref{prop:no-eps}}\label{app:prop:no-eps}
\begin{proof}
	Let $\yA=\yvpaA$.
	Define the mapping $\yfun{E}{S}{\ypow{S}}$ by
	$$E(s)=\{p\yst \ypdatrtf{(s,\yeps,\bot)}{n}{\yA}{(p,\yeps,\bot)}, n\geq 0\},$$ 
	that is, $E(s)$ is the set of all states that can be reached from $s$ through $\yeps$-moves.
	A simple inductive algorithm guarantees that we can effectively compute $E$.
	
	Consider the VPA $\yB=\yvpa{S}{Q_{in}}{A}{\yGa}{\mu}{F}$ where 
	$Q_{in}=\bigcup\limits_{s\in S_{in}} E(s)$, and 
	the set of transitions $\mu$ is obtained from $\rho$ by removing all $\yeps$-transitions from $\rho$, and then adding to $\mu$ all transitions $(r,a,Z,p)$ where $a\neq \yeps$, $(s,a,Z,q)$ is a transition in $\rho$ and, by means of $\yeps$-transitions alone, we can reach $s$ from $r$ and  $p$ from $q$, that is,
	$$\mu=\left[\strut\rho-\{(s,\yeps,\yait,p)\yst s,p\in S\}\right]\bigcup \left\{\strut(r,a,Z,p)\yst (s,a,Z,q)\in \rho, a\neq \yeps, s\in E(r), p\in E(q) \right\}.$$
	Clearly, $\yB$ has no $\yeps$-moves, and has the same number of states as $\yA$.
	
	The next two claims will be used to show that $\yA$ and $\yB$ are equivalent VPAs. 
	\begin{description}
		\item[\sc Claim 1:] If $\ysi\neq \yeps$ and $\ypdatrtf{(s,\ysi,\yal)}{\star}{\yA}{(t,\yeps,\ybe)}$, then 
		$\ypdatrtf{(s,\ysi,\yal)}{\star}{\yB}{(t,\yeps,\ybe)}$.\\
		{\sc Proof:} Let $\ypdatrtf{(s,\ysi,\yal)}{n}{\yA}{(t,\yeps,\ybe)}$, $n\geq 1$ and $\ysi=a\yde$ where $a\in A$. 
		When $n=1$, we have $\yde=\yeps$ and there is a transition $(s,a,Z,t)$ in $\rho$. 
		Since $s\in E(s)$ and $t\in E(t)$, the construction gives $(s,a,Z,t)$ in $\mu$ and the result follows immediately. 
		When $n>1$ we may write
		$$(s,\ysi,\yal)=\ypdatrtf{(s,a\yde,\yal)}{k}{\yA}{(r,a\yde,\yal)}\ypdatrtf{}{}{\yA}{(q,\yde,\yga)}\ypdatrtf{}{m}{\yA}{(t,\yeps,\ybe)},$$
		where $k\geq 0$ and $k+1+m=n$, that is, the first $k$ moves of $\yA$ are $\yeps$-moves, while the next one is a move over $a$, through a transition $(r,a,Z,q)\in\rho$, for some $Z\in\yGa_\bot$. 
		Then, $r\in E(s)$. 
		If $\yde=\yeps$, the last $m$ moves of $\yA$ are also $\yeps$-moves, and we get $\yga=\ybe$ and $t\in E(q)$.
		In this case, the construction gives $(s,a,Z,t)\in\mu$.
		Since $\ysi=a\yde=a$ we have  
		$(s,\ysi,\yal)=\ypdatrtf{(s,a,\yal)}{}{\yB}{(t,\yeps,\yga)}=(t,\yeps,\ybe)$, as desired.
		Now let $\yde\neq \yeps$. Since $q\in E(q)$ and we already have $r\in E(s)$, the construction gives $(s,a,Z,q)\in \mu$, so that $(s,\ysi,\yal)=\ypdatrtf{(s,a\yde,\yal)}{}{\yB}{(q,\yde,\yga)}$.
		But $m<n$, and since $\yde\neq \yeps$,  the induction gives $\ypdatrtf{(q,\yde,\yga)}{\star}{\yB}{(t,\yeps,\ybe)}$.
		Hence, $\ypdatrtf{(s,\ysi,\yal)}{}{\yB}{(t,\yeps,\yga)}$ and the claim holds.
	\end{description}
	Now let $\ysi\in L(\yA)$.
	Then, $\ypdatrtf{(s_0,\ysi,\bot)}{\star}{\yA}{(f,\yeps,\yga)}$ where $s_0\in S_{in}$ and $f\in F$.
	When $\ysi=\yeps$ we get $\yga=\bot$ and $f\in E(s_0)$.
	Since $s_0\in S_{in}$ we get $f\in Q_{in}$.
	Clearly, $\ypdatrtf{(f,\yeps,\bot)}{0}{\yB}{(f,\yeps,\bot)}$ and so $\ysi=\yeps\in L(\yB)$.
	When $\ysi\neq \yeps$,  Claim 1 gives $\ypdatrtf{(s_0,\ysi,\bot)}{\star}{\yB}{(f,\yeps,\yga)}$.
	Since $s_0\in E(s_0)$, the construction gives $s_0\in Q_{in}$.
	Hence, $\ysi\in L(\yB)$.
	We may now conclude that $L(\yA)\ysse L(\yB)$.
	
	For the converse, we state
	\begin{description}
		\item[\sc Claim 2:] If $\ypdatrtf{(s,\ysi,\yal)}{n}{\yB}{(t,\yeps,\ybe)}$ with $n\geq 0$, then 
		$\ypdatrtf{(s,\ysi,\yal)}{\star}{\yA}{(t,\yeps,\ybe)}$.\\
		{\sc Proof:} When $n=0$ the result follows immediately. Let $n>0$.
		Since $\yB$ has no $\yeps$-moves, we must have $\ysi=a\yde$ with $a\in A$, and 
		$$(s,\ysi,\yal)=\ypdatrtf{(s,a\yde,\yal)}{}{\yB}{(r,\yde,\yga)}\ypdatrtf{}{n-1}{\yB}{(t,\yeps,\ybe)},$$
		where the first move was through a transition $(s,a,Z,r)\in \mu$. 
		By the construction, we must have $(p,a,Z,q)\in \rho$ with $p\in E(s)$ and $r\in E(q)$.
		This gives $\ypdatrtf{(s,a\yde,\yal)}{\star}{\yA}{(p,a\yde,\yal)}$ and $\ypdatrtf{(q,\yde,\yga)}{\star}{\yA}{(r,\yde,\yga)}$.
		Composing we get 
		$$(s,\ysi,\yal)=\ypdatrtf{(s,a\yde,\yal)}{\star}{\yA}{(p,a\yde,\yal)}\ypdatrtf{}{}{\yA}{(q,\yde,\yga)}\ypdatrtf{}{\star}{\yA}{(r,\yde,\yga)}.$$
		Using the induction hypothesis we get $\ypdatrtf{(r,\yde,\yga)}{\star}{\yA}{(t,\yeps,\ybe)}$.
		So, $\ypdatrtf{(s,\ysi,\yal)}{\star}{\yA}{(t,\yeps,\ybe)}$ and the claim holds.
	\end{description}
	Now let $\ysi\in L(\yB)$, so that $\ypdatrtf{(q_0,\ysi,\bot)}{\star}{\yB}{(t,\yeps,\yga)}$, with $q_0\in Q_{in}$ and $t\in F$.
	Using Claim 2, we may write $\ypdatrtf{(q_0,\ysi,\bot)}{\star}{\yA}{(t,\yeps,\yga)}$.
	Also, the construction gives some $s_0\in S_{in}$ with $q_0\in E(s_0)$. 
	Hence,  $\ypdatrtf{(s_0,\ysi,\bot)}{\star}{\yA}{(q_0,\ysi,\bot)}$. 
	Composing we get $\ypdatrtf{(s_0,\ysi,\bot)}{\star}{\yA}{(t,\yeps,\yga)}$, that is $\ysi\in L(\yA)$.
	This shows that $L(\yB)\ysse L(\yA)$.
	
	We now have $L(\yA)=L(\yB)$, and the proof is complete. 
\end{proof}

\subsection{Proof of Proposition~\ref{prop:no-eps-determ}}\label{app:prop:no-eps-determ}
\begin{proof}
	Let $\yA=\yvpaA$.
	First we eliminate cycles of $\yeps$-transitions, and in a second step we eliminate the remaining $\yeps$-transitions.
	Let 
	\begin{align}\label{eq:2-23a}
		(s_1,\yeps,\yait,s_{2}),   (s_2,\yeps,\yait,s_{3}), \ldots, (s_{k-1},\yeps,\yait,s_{k}),  (s_k,\yeps,\yait,s_{k+1}),
	\end{align}
	with $s_{k+1}=s_1$ and $k\geq 1$, be a $\yeps$-cycle in $\yA$.
	We construct the new VPA $\yB=\yvpa{Q}{Q_{in}}{A}{\yGa}{\mu}{E}$, mapping the cycle into a state that occurs in it, say into $s_1$.
	Let $T=\{ (s_j,\yeps,\yait,s_{j+1}): j=1,\ldots, k\}$, $J=\{s_j: j=1,\ldots, k\}$.
	Start with $\yB=\yA$ and transform $\yB$ as follows:
	\begin{enumerate}
		\item[(a)] Transitions: (i) remove $T$ from $\mu$; (ii) for all $(p,x,Z,q)\in\rho$ with $p\not\in J$ and $q\in J$
		remove  $(p,x,Z,q)$ from $\mu$ and add $(p,x,Z,s_1)$ to $\mu$.
		\item[(b)] States: (i) remove $s_j$ from $Q$ and $E$, $j=2,\ldots, k$; (ii) if $F\cap J\neq \yemp$, then add $s_1$ to $E$; (iii) 
		if $S_{in}\cap J\neq \yemp$, then let $Q_{in}=\{s_1\}$.
	\end{enumerate}
	We have to show that $\yB$ is equivalent to $\yA$, and that $\yB$ is deterministic.
\begin{description}
\item[\sf Claim 1:] $\yB$ is deterministic.

{\sc Proof.}
	Since $\yA$ is deterministic, $\vert S_{in}\vert\leq 1$.
	Hence, it is clear that $\vert Q_{in}\vert\leq 1$.
	We now  look at each conditions in Definition~\ref{def:vpa-determinism}.
	
	Suppose that $t_i=(p,x,Z_i,q_i)\in \mu_c$, $i=1,2$.
	If $t_1, t_2\in\rho$, the determinism of $\yA$ immediately gives $Z_1=Z_2$, $q_1=q_2$ and so condition (1) in Definition~\ref{def:vpa-determinism} holds.
	With no loss of generality assume $t_1\not\in \rho$.
	Then item (a) of the construction forces $p\not\in J$, $q_1=s_1$, and $t'_1=(p,x,Z_1,s_j)\in \rho$ for some $s_j\in J$.  
	If $t_2\in\rho$ then, because $\yA$ is deterministic, condition (1) in Definition~\ref{def:vpa-determinism} implies $Z_1=Z_2$ and $q_2=s_j$, so that $t_2=(p,x,Z_2,s_j)\in \mu_c$.
	If $j\geq 2$, we see that $(p,x,Z_1,s_j)\in \mu_c$ contradicts item (a) of the construction.
	Hence, $j=1$ and then $q_2=s_j=s_1=q_1$.
	Together with $Z_1=Z_2$ we see that  condition (1) in Definition~\ref{def:vpa-determinism} holds.
	Assume now $t_2\not\in \rho$.
	We obtain again $q_2=s_1$ and $t'_2=(p,x,Z_2,s_\ell)\in\rho$ with $s_\ell\in J$.
	Now $t'_1,t'_2\in\rho$ and the determinism of $\yA$ forces $Z_1=Z_2$.
	We now have $q_1=s_1=q_2$, $Z_1=Z_2$ and condition (1) in Definition~\ref{def:vpa-determinism} holds again. 
	
	Suppose that $t_i=(p,x,Z,q_i)\in \mu_r\cup\mu_i$, $i=1,2$.
	An argument entirely similar to the one in the preceding paragraph shows that  condition (2) in Definition~\ref{def:vpa-determinism} holds. 
	
	Now let $t_1=(p,x,Z,q_1)\in \mu$ and $t_2=(p,\yeps,\yait,q_2)\in \mu$, with $x\neq \yeps$.
	If $t_1,t_2\in\rho$ we get an immediate contradiction to the determinism of $\yA$.
	Assume $t_1\not\in\rho$. 
	As before the construction gives $t'_1=(p,x,Z,s_j)\in \rho$ for some $s_j\in J$.
	Since  $x\neq \yeps$, if $t_2\in\rho$ we get a contradiction to the determinism of $\yA$.
	Hence, $t_2\not\in\rho$ implies, by the construction, that $t'_2=(p,\yeps,\yait,s_\ell)\in \rho$ for some $s_\ell\in J$.
	Again, $x\neq \yeps$ leads to a contradiction to the determinism of $\yA$.
	Finally, when $t_1\in\rho$ the determinism of $\yA$ will force $t_2\not\in\rho$, and the construction will give $t'_2=(p,\yeps,\yait,s_\ell)\in \rho$
	for some $s_\ell\in J$.
	Then, $t_1, t'_2\in\rho$ contradicts the determinism of $\yA$ again because $x\neq \yeps$.
	We conclude that  $t_1=(p,x,Z,q_1)\in \mu$ and $t_2=(p,\yeps,\yait,q_2)\in \mu$, with $x\neq \yeps$, cannot happen and condition (3) in Definition~\ref{def:vpa-determinism} also holds. 
\end{description}

Next, we want to argue for language equivalence.	
First we show that $\yA$ can imitate runs of $\yB$.
\begin{description}
\item[\sf Claim 2:] If $\ypdatrtf{(q,\ysi,\yal_1\bot)}{\star}{\yB}{(p,\yeps,\yal_2\bot)}$ then we also have $\ypdatrtf{(q,\ysi,\yal_1\bot)}{\star}{\yA}{(p,\yeps,\yal_2\bot)}$.

{\sc Proof:}	Write $\ypdatrtf{(q,\ysi,\yal_1\bot)}{n}{\yB}{(p,\yeps,\yal_2\bot)}$, where $n\geq 0$.
		If all transitions used in the run over $\yB$ are in $\rho$, there is nothing to prove.
	Let $t=(u,x,Z,v)$, with $x \in A$ and $Z\in\yGa\cup\{\yeps\}$, be the first transition not in $\rho$  used in the run over $\yB$.
	By item (a) of the construction we get $v=s_1$ and $t'=(u,x,Z,s_j)\in\rho$, where $s_j\in J$.
	Then, $\ysi=\ysi_1x\ysi_2$ and we may write
	$\ypdatrtf{(q,\ysi_1x\ysi_2,\yal_1\bot)}{\star}{\yA}{}{(u,x\ysi_2,\ybe\bot)}\ypdatrtf{}{1}{\yB}{(s_1,\ysi_2,\yga\bot)}\ypdatrtf{}{k}{\yB}{(p,\yeps,\yal_2\bot)}$, with $0\leq k<n$.
	Using the transitions in the $\yeps$-cycle at Eq.(\ref{eq:2-23a}) we get $\ypdatrtf{(s_j,\ysi_2,\yga\bot)}{\star}{\yA}{}{(s_1,\ysi_2,\yga\bot)}$.
	Thus, using $t'$ we have 
	$$\ypdatrtf{(q,\ysi_1x\ysi_2,\yal_1\bot)}{\star}{\yA}{}{(u,x\ysi_2,\ybe\bot)}\ypdatrtf{}{1}{\yA}{(s_j,\ysi_2,\yga\bot)}\ypdatrtf{}{\star}{\yA}{(s_1,\ysi_2,\yga\bot)}\ypdatrtf{}{k}{\yB}{(p,\yeps,\yal_2\bot)}.$$
	Inductively, we have $\ypdatrtf{(s_1,\ysi_2,\yga\bot)}{\star}{\yA}{(p,\yeps,\yal_2\bot)}$.
	Putting it together, we now have $\ypdatrtf{(q,\ysi_1x\ysi_2,\yal_1\bot)}{\star}{\yA}{(p,\yeps,\yal_2\bot)}$.
\end{description}
Next, we want to show that if we have a run of $\yA$ starting at $S_{in}$ and ending in $F$, then we can extend the run to start at $Q_{in}$ and end in $E$.
\begin{description}
\item[\sf Claim 3:] If $\ypdatrtf{(r,\ysi,\bot)}{\star}{\yA}{(f,\yeps,\yal\bot)}$ with $r\in S_{in}$, $f\in F$, then  $\ypdatrtf{(u,\ysi,\bot)}{\star}{\yA}{(v,\yeps,\yal\bot)}$ with $u\in Q_{in}$, $v\in E$.

{\sc Proof.}
Assume $\ysi \in L(\yA)$ so that  $\ypdatrtf{(r,\ysi,\bot)}{\star}{\yA}{(f,\yeps,\yal\bot)}$ with $r\in S_{in}$, $f\in F$.
	First, we want to show that we also have $\ypdatrtf{(u,\ysi,\bot)}{\star}{\yA}{(v,\yeps,\yal\bot)}$ with $u\in Q_{in}$ and $v\in E$.
	If $r\in Q_{in}$ let $u=r$.
	If $r\not\in Q_{in}$, because  $\vert S_{in}\vert\leq 1$ item (b) of the construction implies that $Q_{in}=\{s_1\}$ and $r=s_\ell\in S_{in}$ for some $s_\ell\in J$.
	Since $\yA$ is deterministic, condition (3) in Definition~\ref{def:vpa-determinism} says that $(s_i,\yeps,\yait,s_{i+1})$ is the unique transition out of $s_i$, for all $s_i\in J$.
Let $u=s_1$.
	Hence, we must have $\ypdatrtf{(r,\ysi,\bot)}{\star}{\yA}{(u,\ysi,\bot)}\ypdatrtf{}{\star}{\yA}{(f,\yeps,\yal\bot)}$ with $u\in Q_{in}$.
	If $f\in E$, let $v=f$.
	If $f\not\in E$, item (b) of the construction says that $f=s_\ell$, for some $s_\ell\in J$, and $s_1\in E$.
Let $v=s_1$.
	We now have $\ypdatrtf{(u,\ysi,\bot)}{\star}{\yA}{(f,\ysi,\yal\bot)}\ypdatrtf{}{\star}{\yA}{(v,\ysi,\yal\bot)}$ with $v\in E$.
We can now assume $\ypdatrtf{(r,\ysi,\bot)}{\star}{\yA}{(f,\yeps,\yal\bot)}$ with $r\in Q_{in}$, $f\in E$. 
\end{description}
Now we are ready for language equivalence.
\begin{description}
\item[\sf Claim 4:] $L(\yA)=L(\yB)$.

{\sc Proof.} 
Now let $\ysi\in L(\yB)$, so that we have $\ypdatrtf{(q,\ysi,\bot)}{n}{\yB}{(p,\yeps,\yal\bot)}$ with $q\in Q_{in}$, $p\in E$, $n\geq 0$.
Using Claim 2, we can write $\ypdatrtf{(q,\ysi,\bot)}{\star}{\yA}{(p,\yeps,\yal\bot)}$. 
	If $q\not\in S_{in}$ then by item (b) of the construction we must have $s_\ell\in S_{in}$ for some $s_\ell\in J$ and $q=s_1$.
	From the $\yeps$-cycle at Eq. (\ref{eq:2-23a}) we have $\ypdatrtf{(s_\ell,\ysi,\bot)}{\star}{\yA}{(s_1,\ysi,\bot)}=(q,\ysi,\bot)$, so that $\ypdatrtf{(s_\ell,\ysi,\bot)}{\star}{\yA}{(p,\yeps,\yal\bot)}$.
	Likewise, if $p\not\in F$ then by item (b) again we must have $p=s_1$ and $s_i\in F$ for some $s_i\in J$.
	Again,  the $\yeps$-cycle  gives $(p,\yeps,\yal\bot)=\ypdatrtf{(s_1,\yeps,\yal\bot)}{\star}{\yA}{(s_i,\yeps,\yal\bot)}$.
	Composing, we obtain $\ypdatrtf{(s_\ell,\ysi,\bot)}{\star}{\yA}{(s_i,\yeps,\yal\bot)}$ with $s_\ell\in S_{in}$ and $s_i\in F$.
Thus we always have $\ypdatrtf{(u,\ysi,\bot)}{\star}{\yA}{(v,\yeps,\yal\bot)}$, with $u\in S_{in}$ and $v\in F$.
Hence, $\ysi\in L(\yA)$.

Now let $\ysi\in L(\yA)$.
Using Claim 3, we can write $\ypdatrtf{(r,\ysi,\bot)}{\star}{\yA}{(f,\yeps,\yal\bot)}$ with $r\in Q_{in}$, $f\in E$. 
	If no state $s_j\in J$ occurs in this run, then all transitions are also in $\mu$, and we  get $\ypdatrtf{(r,\ysi,\bot)}{\star}{\yB}{(f,\yeps,\yal\bot)}$.
	Next, let  $\ysi=\ysi_1\ysi_2$ with $\ypdatrtf{(r,\ysi_1\ysi_2,\bot)}{\star}{\yA}{(s_j,\ysi_2,\ybe\bot)}\ypdatrtf{}{\star}{\yA}{(f,\yeps,\yal\bot)}$, where this is the first occurrence of a state of $J$ in the run  over $\yA$.
	Invoking condition (3) in Definition~\ref{def:vpa-determinism} again, we must have $f=s_i\in J$, $\ysi_2=\yeps$ and $\ybe=\yal$.
	Since $f\in E$, we must have $f=s_1\in E$.
	We now have $\ysi=\ysi_1$ and $\ypdatrtf{(r,\ysi,\bot)}{\star}{\yA}{(s_1,\yeps,\yal\bot)}$, $r\in Q_{in}$, $s_1\in E$.
	When $n=0$ we get  $\ypdatrtf{(r,\ysi,\bot)}{0}{\yB}{(s_1,\yeps,\yal\bot)}$ and we are done.
	Else, we must have $\ysi =\ysi'x$ with $x\in\yGa\cup\{\yeps\}$ and $\ypdatrtf{(r,\ysi'x,\bot)}{\star}{\yA}{(t,x,\yga\bot)}\ypdatrtf{}{1}{\yA}{(s_1,\yeps,\yal\bot)}$ with a transition $(t,x,Z,s_1)\in\rho$ used in the last move.
	Because $s_j$ is the first occurrence of a state in $J$ we get $t\not\in J$, and all transitions in $\ypdatrtf{(r,\ysi'x,\bot)}{\star}{\yA}{(t,x,\yga\bot)}$ are in $\mu$.
	Hence, $\ypdatrtf{(r,\ysi'x,\bot)}{\star}{\yB}{(t,x,\yga\bot)}$.
	Item (a) of the construction readily gives $(t,x,Z,s_1)\in\mu$.
	Thus, $\ypdatrtf{(r,\ysi,\bot)}{\star}{\yB}{(t,x,\yga\bot)}\ypdatrtf{}{1}{\yB}{(s_1,\yeps,\yal\bot)}$, $r\in Q_{in}$, $s_1\in E$. 
We conclude that $L(\yA)\ysse L(\yB)$.
\end{description}	
At this point we know how to remove an $\yeps$-cycle from $\yA$, while maintain language equivalence and determinism. 
Thus, we can repeat the construction for all $\yeps$-cycles in $\yA$, so that we can now assume that there is no $\yeps$-cycle in $\yA$.
As a final step, we show how to remove all remaining $\yeps$-transitions from $\yA$.
	
Let $\yA=\yvpaA$ be deterministic and with no $\yeps$-cycles.
Let $t=(p,\yeps,\yait,q)\in\rho$. 
Since $\yA$ has no $\yeps$-cycles, we can assume that for all transitions $(q,x,Z,r)\in \rho$ we have $x\neq \yeps$.
	Start with $\yB=\yA$ and transform $\yB$ as follows.
	\begin{itemize}
		\item[(c)] Transitions: (i) $\mu=\rho-\{t\}$; (ii) for all $(q,a,Z,r)\in\rho$, add $(p,a,Z,r)$ to $\mu$.
		\item[(d)] States: (i) if $p\in S_{in}$, $p\not\in F$, let $Q_{in}=\{q\}$; (ii) if $q\in F$, let $E=F\cup\{p\}$.
	\end{itemize}
We still have determinism.
\begin{description}
\item[\sf Claim 5:] $\yB$ is deterministic.\\
{\sc Proof.} 
	We want to show that $\yB$ is also deterministic.
	It is clear that $\vert Q_{in}\vert\leq 1$ because $\yA$ is deterministic.
	Assume that $\yB$ is not deterministic.
	Since $\yA$ is already deterministic and the only transitions added to $\mu$ are of the form $(p,a,Z,r)$, we see that the only possibility for $\yB$ not deterministic is that we have two transitions $t_i=(p,a_i,Z_i,r_i)\in\mu$, $i=1,2$ in violation  Definition~\ref{def:vpa-determinism}. 
	Since $t\in\rho$ and $\yA$ is deterministic, according to condition (3) of Definition~\ref{def:vpa-determinism}, this is the \emph{only} transition out of $p$ in $\rho$. 
	Moreover $t$ was removed from $\rho$ according to item (c) of the construction.
	Therefore, $t_1$, $t_2$ are new transitions added to $\mu$ by the construction.
	Hence, we must have transitions $t'_i=(q,a_i,Z_i,r_i)\in\rho$, $i=1,2$.
	But these two transitions in $\rho$ also violate Definition~\ref{def:vpa-determinism}, and so they contradict the determinism of $\yA$.
We conclude that $\yB$ must be deterministic too.
\end{description}	
We complete the proof arguing for language equivalence.
Again, we must first relate runs in $\yA$ to runs in $\yB$, and vice-versa.
\begin{description}
\item[\sf Claim 6:] If $\ypdatrtf{(s,\ysi,\ybe\bot)}{\star}{\yB}{(f,\yeps,\yal\bot)}$, then we also have $\ypdatrtf{(s,\ysi,\ybe\bot)}{\star}{\yA}{(f,\yeps,\yal\bot)}$.\\
{\sc Proof.} 
	First, we show that for any run in $\yB$ there exists a similar run in $\yA$.
	Let $\ypdatrtf{(s,\ysi,\ybe\bot)}{n}{\yB}{(f,\yeps,\yal\bot)}$.
	If $n=0$ we can immediately write $\ypdatrtf{(s,\ysi,\ybe\bot)}{0}{\yA}{(f,\yeps,\yal\bot)}$.
	Now let $x\in A_\yeps$, $\ysi=x\ysi'$ and write
	\begin{align}\label{eq:vpadet1}
		\ypdatrtf{(s,x\ysi',\ybe\bot)}{1}{\yB}{(r,\ysi',\yga\bot)}\ypdatrtf{}{n-1}{\yB}{(f,\yeps,\yal\bot)},
	\end{align}
	with $t'=(s,x,Z,r)$ the first transition used in this run.
	Inductively,  $\ypdatrtf{(r,\ysi',\yga\bot)}{\star}{\yA}{(f,\yeps,\yal\bot)}$.
	If $t'\in\rho$, we have  $\ypdatrtf{(s,x\ysi',\ybe\bot)}{1}{\yA}{(r,\ysi',\yga\bot)}$, so that 
	$\ypdatrtf{(s,\ysi,\ybe\bot)}{\star}{\yA}{(f,\yeps,\yal\bot)}$.
	If $t'\not\in\rho$, the construction gives $s=p$, $t'=(p,x,Z,r)$ and $(q,x,Z,r)\in\rho$.
	Then, using $t=(p,\yeps,\yait,q)$ we have
	$$(s,x\ysi',\ybe\bot)=\ypdatrtf{(p,x\ysi',\ybe\bot)}{1}{\yA}{(q,x\ysi',\ybe\bot)}\ypdatrtf{}{1}{\yA}{(r,\ysi',\yga\bot)}\ypdatrtf{}{\star}{\yA}{(f,\yeps,\yal\bot)}.$$
\end{description}	
Now we examine how $\yB$ can imitate runs in $\yA$.
\begin{description}
\item[\sf Claim 7:] 
Let $\ypdatrtf{(s_1,\ysi,\yal_1\bot)}{n}{\yA}{(s_2,\yeps,\yal_2\bot)}$ with $n\geq 0$.
	We show that $\ypdatrtf{(s_1,\ysi,\yal_1\bot)}{\star}{\yB}{(s'_2,\yeps,\yal_2\bot)}$ where $s'_2=p$ if the last transition in the run from $\yA$ was $t$, else $s_2'=s_2$.\\
{\sc Proof.}
	If $n=0$ the result is immediate, with $s'_2=s_2$. 
	Next, let $n\geq 1$, $x\in A_\yeps$, $\ysi=\ysi'x$ and
	\begin{align}\label{eq:vpadet3}
		\ypdatrtf{(s_1,\ysi' x,\yal_1\bot)}{n-1}{\yA}{(s_3,x,\yal_3\bot)}\ypdatrtf{}{1}{\yA}{(s_2,\yeps,\yal_2\bot)},
	\end{align}
	where $t'=(s_3,x,Z,s_2)$ is the transition used in the last move in Eq. (\ref{eq:vpadet3}).

	The first simple case is when $t'=t$.
	Then, $s_3=p$, $x=\yeps$, $s_2=q$ and $\yal_3=\yal_2$.
	Since $s_3=p$, the last transition in $\ypdatrtf{(s_1,\ysi x,\yal_1\bot)}{n-1}{\yA}{(s_3,x,\yal_3\bot)}$ cannot be $t$, otherwise we would have $s_3=p=q$ and then $t$ would be a simple $\yeps$-cycle in $\yA$, a contradiction.
	Inductively from Eq. (\ref{eq:vpadet3}) we 
	can now write $\ypdatrtf{(s_1,\ysi' x,\yal_1\bot)}{\star}{\yB}{(s_3,x,\yal_3\bot)}=(p,\yeps,\yal_2\bot)$.
	Picking $s'_2=p$, we are done with this case.

Let now that $t'\neq t$.
This gives $t'=(s_3,x,Z,s_2)\in \yB$.
If $n-1=0$ in Eq. (\ref{eq:vpadet3}) we get $s_1=s_3$, $\ysi'=\yeps$, $\yal_1=\yal_3$.
The last move in Eq. (\ref{eq:vpadet3}) gives $(s_3,x,\yal_3\bot)=\ypdatrtf{(s_1,x,\yal_1\bot)}{1}{\yA}{(s_2,\yeps,\yal_2\bot)}$.
Since $t'\in \yB$ we also have $\ypdatrtf{(s_1,x,\yal_1\bot)}{1}{\yB}{(s_2,\yeps,\yal_2\bot)}$, and the result follows because $t'\neq t$.
We now proceed under the hypothesis that $n-1>0$.
Inductively, from Eq. (\ref{eq:vpadet3}) we get
	$\ypdatrtf{(s_1,\ysi' x,\yal_1\bot)}{\star}{\yB}{(s_3',x,\yal_3\bot)}$, where $s'_3=p$ if the last transition, say $t''$, in $\ypdatrtf{(s_1,\ysi' x,\yal_1\bot)}{n-1}{\yA}{(s_3,x,\yal_3\bot)}$ was $t$, otherwise we must have $s'_3=s_3$.
If $t''\neq t$, then $t'\in \yB$ and $s'_3=s_3$ give  $(s'_3,x,\yal_3\bot)=\ypdatrtf{(s_3,x,\yal_3\bot)}{1}{\yB}{(s_2,\yeps,\yal_2\bot)}$.
Thus, $\ypdatrtf{(s_1,\ysi' x,\yal_1\bot)}{\star}{\yB}{(s_2,\yeps,\yal_3\bot)}$.
	Picking $s'_2=s_2$ and recalling that $t'$ is not $t$, we have the desired result.
Lastly, let $t''=t$.

We can rewrite the first $n-1$ moves in Eq. (\ref{eq:vpadet3}) thus
	\begin{align}\label{eq:vpadet4}
		\ypdatrtf{(s_1,\ysi',\yal_1\bot)}{n-2}{\yA}{(p,\yeps,\yal_3\bot)}\ypdatrtf{}{1}{\yA}{(q,\yeps,\yal_3\bot)}=(s_3,\yeps,\yal_3\bot),
	\end{align}
	so that $q=s_3$.
	Inductively, we get 
	$\ypdatrtf{(s_1,\ysi',\yal_1\bot)}{\star}{\yB}{(p,\yeps,\yal_3\bot)}$, so that we also have $\ypdatrtf{(s_1,\ysi' x,\yal_1\bot)}{\star}{\yB}{(p,x,\yal_3\bot)}$.
	We now have $t'=(s_3,x,Z,s_2)=(q,x,Z,s_2)$ in $\yA$.
	By item (c) of the construction, we must have $(p,x,Z,s_2)$ in $\yB$. 
Since $\ypdatrtf{(s_3,x,\yal_3\bot)}{1}{\yA}{(s_2,\yeps,\yal_2\bot)}$ using $t'$, we also get
$\ypdatrtf{(p,x,\yal_3\bot)}{1}{\yB}{(s_2,\yeps,\yal_2\bot)}$.
	Composing, we obtain $\ypdatrtf{(s_1,\ysi' x,\yal_1\bot)}{\star}{\yB}{(s_2,\yeps,\yal_2\bot)}$.
	Picking $s'_2=s_2$ and remembering that we assumed $t'\neq t$, we have the result again.
\end{description}
The first half of language equivalence now follows. 
\begin{description}
\item[\sf Claim 8:] $L(\yB)\ysse L(\yA)$.\\
{\sc Proof.}
	Let  $\ypdatrtf{(s,\ysi,\bot)}{\star}{\yB}{(f,\yeps,\yal\bot)}$ with
	$s\in Q_{in}$ and $f\in E$.
Claim 6 gives 
	\begin{align}\label{eq:vpadet2}
		\ypdatrtf{(s,\ysi,\bot)}{\star}{\yA}{(f,\yeps,\yal\bot)}.
	\end{align}
	If $f\in F$, let $f'=f$ and we get $\ypdatrtf{(s,\ysi,\bot)}{\star}{\yA}{(f',\yeps,\yal\bot)}$ with $f'\in F$.
	If $f\not\in F$, item (d) of the construction says that $f=p$ and $q\in F$. Picking $f'=q$ we have $f'\in F$. Using $t$ again, we have
	$(f,\yeps,\yal\bot)=\ypdatrtf{(p,\yeps,\yal\bot)}{1}{\yA}{(q,\yeps,\yal\bot)}=(f',\yeps,\yal\bot)$ with $f'\in F$.
	So, we can assume that $f\in F$ in Eq. (\ref{eq:vpadet2}).
	If $s\in S_{in}$, then Eq. (\ref{eq:vpadet2}) says that $\ysi\in L(\yA)$, as desired.
	Now assume $s\not\in S_{in}$.
	Since $s\in Q_{in}$, item (d) of the construction implies that $s=q$, $p\in S_{in}$.
	Using transition $t$ again, we now have 
	$\ypdatrtf{(p,\ysi,\bot)}{1}{\yA}{(q,\ysi,\bot)}=(s,\ysi,\bot)$.
	Composing with  Eq. (\ref{eq:vpadet2}) we get $\ypdatrtf{(p,\ysi,\bot)}{\star}{\yA}{(f,\yeps,\yal\bot)}$.
	Because $p\in S_{in}$ and $f\in F$, we see that $\ysi\in L(\yA)$.
\end{description}	
The next claim completes the proof.
\begin{description}
\item[\sf Claim 9:] $L(\yA)\ysse L(\yB)$.\\
{\sc Proof.}  
Let $s\in S_{in}$, $f\in F$, $n\geq 0$, and 		$\ypdatrtf{(s,\ysi,\bot)}{n}{\yA}{(f,\yeps,\yal\bot)}$.
Assume that $s\not\in Q_{in}$.
Since $s\in S_{in}$, then item (d) of the construction gives $Q_{in}=\{q\}$, $p\in S_{in}$ and $p\not\in F$.
Since $\yA$ is deterministic, we get $S_{in}=\{p\}=\{s\}$, so that $p=s$.
Hence, $s\not\in F$ and then $s\neq f$.
Thus, $n\geq 1$.
Since $p=s$ and $\yA$ is deterministic, we see that $t$ is the first transition in the run $\ypdatrtf{(s,\ysi,\bot)}{n}{\yA}{(f,\yeps,\yal\bot)}$.
Thus, we can write $(s,\ysi,\bot)=\ypdatrtf{(p,\ysi,\bot)}{1}{\yA}{(q,\ysi,\bot)}\ypdatrtf{}{n-1}{\yA}{(f,\yeps,\yal\bot)}$, with $q\in Q_{in}$ and $n-1\geq 0$.
We can, thus, take $s\in Q_{in}$ in the run  	$\ypdatrtf{(s,\ysi,\bot)}{n}{\yA}{(f,\yeps,\yal\bot)}$.

If $n=0$ we get $\ysi=\yeps$, $\yal=\yeps$ and $s=f$.
We can then  write 	$(s,\ysi,\bot)=\ypdatrtf{(s,\yeps,\bot)}{0}{\yB}{(f,\yeps,\bot)}$.
Because $F\ysse E$ we get $f\in E$. 
Since $s\in Q_{in}$ we get $\ysi=\yeps\in L(\yB)$.

We proceed now with $s\in Q_{in}$, $f\in E$, $n\geq 1$ and 		$\ypdatrtf{(s,\ysi,\bot)}{n}{\yA}{(f,\yeps,\yal\bot)}$.
Use Claim 7 to write 
		$\ypdatrtf{(s,\ysi,\bot)}{\star}{\yB}{(f',\yeps,\yal\bot)}$,
where $f'=p$ if $t$ was the last transition  in the run over $\yA$, otherwise we have $f'=f$.
If $f'=f$, we get $f'\in E$ and then $\ysi\in L(\yB)$.
Assume now that $f'=p$ and $t$ was the last transition  in the run over $\yA$.
This gives $q=f$, and then item (d) of the construction says that $p\in E$.
Thus, $f'\in E$ and we get again $\ysi\in L(\yB)$.
\end{description}	
\end{proof}

\subsection{Proof of Proposition~\ref{prop:product-behavior}}\label{app:prop:product-behavior}
\begin{proof}
Let $\ypdatrtf{((s,q),\ysi,\yga_0)}{n}{\yS\times \yQ}{((p,r),\yeps,\yga_n)}$, for some $n\geq 0$, and where
$\yga_0=(X_1,Y_1)\ldots (X_i,Y_i)\bot$ and $\yga_n=(Z_1,W_1)\ldots (Z_k,W_k)\bot$. 
When $n=0$ the result is immediate.
Proceeding inductively, let $n\geq 1$, $\ysi=\yde a$ with $a\in A\cup \{\yeps\}$, and 
\begin{align}\label{prop2.28a}
\ypdatrtf{((s,q),\yde a,\yga_0)}{n-1}{\yS\times \yQ}{((u,v),a,\yga_{n-1})}\ypdatrtf{}{1}{\yS\times \yQ}{((p,r),\yeps,\yga_n)},
\end{align}
where $\yga_{n-1}=(U_1,V_1)\ldots (U_j,V_j)\bot$, for some $j\geq 0$. 
In order to ease the notation, let $\yal_0=X_1 \ldots X_i\bot$, $\yal_{n-1}=U_1 \ldots U_j\bot$, $\yal_n=Z_1\ldots Z_k\bot$, $\ybe_0=Y_1\ldots Y_i\bot$, $\ybe_{n-1}=V_1 \ldots V_j\bot$,  and $\ybe_n=W_1\ldots W_k\bot$. 
Then, using the induction hypothesis we get 
\begin{align}\label{prop2.28b}
\ypdatrtf{(s,\yde,\yal_0)}{\star}{\yS}{(u,\yeps,\yal_{n-1})},\quad  \ypdatrtf{(q,\yde,\ybe_0)}{\star}{\yQ}{(v,\yeps,\ybe_{n-1})}. 
\end{align}
Let $((u,v),a,Z,(p,r))$ be the transition used in the last $\yS\times \yQ$ move.
From Definition~\ref{def:productVPA}, there are two possibilities: 
\begin{description}
\item[\sc Case 1:] 
$a\neq \yeps$, $(u,a,X,p)\in\rho$, $(v,a,Y,r)\in\mu$, and either $Z=(X,Y)$, or $Z=X=Y\in\{\bot,\yait\}$. 

If $a\in A_c$, then from Eq. (\ref{prop2.28a}) we get $Z\not\in\{\bot,\yait\}$, so that we must have $Z=(X,Y)$.
With $a\in A_c$, Eq (\ref{prop2.28a}) gives $\yga_n=(X,Y)\yga_{n-1}=(X,Y)(U_1,V_1)\cdots (U_j,V_j)\bot$.
But now, with $(u,a,X,p)\in\rho$ and Eq. (\ref{prop2.28b}) we also get 
$\ypdatrtf{(u,a,\yal_{n-1})}{1}{\yS}{(p,\yeps,X\yal_{n-1} )}$.
Using Eq. (\ref{prop2.28b}) and composing, we get  $\ypdatrtf{(s,\ysi,\yal_0)}{\star}{\yS}{(p,\yeps,X\yal_{n-1})}$, where $X\yal_{n-1}=XU_1\cdots U_j\bot$.
Likewise, we obtain $\ypdatrtf{(q,\ysi,\ybe_0)}{\star}{\yQ}{(r,\yeps,Y\ybe_{n-1})}$, with 
$Y\ybe_{n-1}=YV_1\cdots V_j\bot$, as needed.

When $a\in A_r$,  given that  $((u,v),a,Z,(p,r))$ was the last transition used Eq. (\ref{prop2.28a}) with
 $\yga_{n-1}=(U_1,V_1)\ldots (U_j,V_j)\bot$, we get either (i) $j\geq 1$, $Z=(U_1,V_1)$ and $\yga_n=(U_2,V_2)\cdots (U_j,V_j)\bot$, or (ii) $j=0$ and $Z=\bot$ and $\yga_n=\bot$.
Since $Z=(X,Y)$, the first case gives  $X=U_1$ and $Y=V_1$.
We can now repeat the argument above when $a\in A_c$ and reach the desired result.
In the second case, $Z=\bot$ forces $Z=X=Y=\bot$.
Also, $\yal_{n-1}=U_1\cdots U_j\bot$ reduces to $\yal_{n-1}=\bot$.
Since $a\in A_r$ and $(u,a,X,p)\in\rho$, we can write $\ypdatrtf{(u,a,\yal_{n-1})}{1}{\yS}{(p,\yeps,
\bot )}$.
Using Eq. (\ref{prop2.28b}) and composing, we get $\ypdatrtf{(s,\ysi,\yal_0)}{\star}{\yS}{(p,\yeps,\bot)}$.
Since $\yga_n=\bot$ we have the desired result for $\yS$.
A similar reasoning also gives  $\ypdatrtf{(q,\ysi,\ybe_0)}{\star}{\yQ}{(r,\yeps,\bot)}$, as needed.

If $a\in A_i$, since    $((u,v),a,Z,(p,r))$ was the last transition used Eq. (\ref{prop2.28a}), we get $Z=\yait$ and $\yga_n=\yga_{n-1}=(U_1,V_1)\cdots (U_j,V_j)\bot$.
Since $Z\in\{\bot,\yait\}$, we must also have $Z=X=Y=\yait$.
Thus, $(u,a,\yait,p)\in\rho$.
Using Eq. (\ref{prop2.28b}) and composing, we get $\ypdatrtf{(s,\ysi,\yal_0)}{\star}{\yS}{(p,\yeps,\yal_{n-1})}$.
Since $\yal_{n-1}=U_1 \ldots U_j\bot$, we get the desired result for $\yS$.
By a similar reasoning we also get $\ypdatrtf{(q,\ysi,\ybe_0)}{\star}{\yQ}{(r,\yeps,\ybe_{n-1})}$, completing this case.

\item[\sc Case 2:] 
$a=\yeps$, $Z=\yait$ with either $u=p$ and $(v,\yeps,\yait,r)\in\mu$, or $v=r$ and $(u,\yeps,\yait,p)\in\rho$.
We look at the first case, the other being entirely similar.
 Since $((u,v),\yeps,\yait,(u,r))$ was the transition used in Eq. (\ref{prop2.28a}) we get $\yga_{n-1}=\yga_n=(U_1,V_1)\cdots (U_j,V_j)\bot$.
From Eq. (\ref{prop2.28b}) and $u=p$, we know that $\ypdatrtf{(s,\ysi,\yal_0)}{\star}{\yS}{(p,\yeps,\yal_{n-1})}$,
and because $\yal_{n-1}=U_1 \ldots U_j\bot$ we get the desired result for $\yS$.
Since $(v,\yeps,\yait,r)\in\mu$, from Eq. (\ref{prop2.28b}) we obtain
$\ypdatrtf{(q,\ysi,\ybe_0)}{\star}{\yQ}{(r,\yeps,\ybe_{n-1})}$, and we have the result also for $\yQ$ because $\ybe_{n-1}=V_1 \ldots V_j\bot$.
This case is now complete.
\end{description}

Next, we look at the converse.
Let $\ypdatrtf{(s,\ysi,\yal_0)}{n}{\yS}{(p,\yeps,\yal_n)}$  and 
$\ypdatrtf{(q,\ysi,\ybe_0)}{m}{\yQ}{(r,\yeps,\ybe_m)}$, with $n,m\geq 0$, where $\yal_0=X_1\ldots X_i\bot$, $\ybe_0=Y_1\ldots Y_i\bot$, $\yal_n=Z_1\ldots Z_k\bot$ and  $\ybe_m=W_1\ldots W_k\bot$, for some $i,k\geq 0$.
Write $\yga_0=(X_1,Y_1)\ldots(X_i,Y_i)\bot$ and $\yga=(Z_1,W_1)\ldots(Z_k,W_k)\bot$.

When $n+m=0$ we get $m=0$ and $n=0$, so that  the result is immediate.

With no loss, let $n\geq 1$.
Then, $\ysi=a\yde$ and  a transition $(s,a,Z,t)\in\rho$ used in the first step in $\yS$. We now have 
\begin{equation}\label{eq:1.24a}
\ypdatrtf{(s,a\yde,\yal_0)}{1}{\yS}{(t,\yde,\yal_1)}\ypdatrtf{}{n-1}{\yS}{(p,\yeps,\yal_n)},\quad \ypdatrtf{(q,\ysi,\ybe_0)}{m}{\yQ}{(q,\yeps,\ybe_m)}.
\end{equation}
We first look at the case when $a=\yeps$.  
Then, $\ysi=\yde$ , $\yal_0=\yal_1$ and  Definition~\ref{def:fsa-move} implies $Z=\yait$.
Because $(s,\yeps,\yait,t)\in \rho$, Definition~\ref{def:productVPA} item (2) says $((s,q),\yeps,\yait,(t,q))$ is in $\yS\times \yQ$.
Thus,
\begin{align}\label{prop2.28c} \ypdatrtf{((s,q),\ysi,\yga_0)}{1}{\yS\times\yQ}{((t,q),\yde,\yga_0)}.
\end{align}
Now we have $\ypdatrtf{(t,\yde,\yal_0)}{n-1}{\yS}{(p,\yeps,\yal_n)}$ and $\ypdatrtf{(q,\yde,\ybe_0)}{m}{\yQ}{(r,\yeps,\ybe_m)}$.
Since $n-1+m<n+m$, inductively, we may write  
$\ypdatrtf{((t,q),\yde,\yga_0)}{\star}{\yS\times\yQ}{((p,r),\yeps,\yga)}$.
From Eq. (\ref{prop2.28c}) we get   $\ypdatrtf{((s,q),\ysi,\yga_0)}{\star}{\yS\times\yQ}{((p,r),\yeps,\yga)}$.
Recalling the definitions of $\yal_0$, $\ybe_0$, $\yal_n$, $\ybe_n$, and of $\yga_0$, $\yga$, we see that the result holds in this case.

We now turn to the case $a\neq \yeps$.
If $m=0$, from $\ypdatrtf{(q,\ysi,\ybe_0)}{m}{\yQ}{(r,\yeps,\ybe_m)}$ we obtain $\ysi=\yeps$. 
But then $\ysi=a\yde$ gives $a=\yeps$, a contradiction.
Thus, $m\geq 1$. 
This gives 
\begin{equation}\label{eq:1.24b}
\ypdatrtf{(q,a\yde,\ybe_0)}{1}{\yQ}{(u,\yde,\ybe_1)}\ypdatrtf{}{m-1}{\yQ}{(r,\yeps,\ybe_m)},
\end{equation}
and we must have a transition $(q,a,W,u)\in\mu$ which was used in the first step. 
There are three cases.
\begin{description}
\item[\sc Case 3:] $a\in A_i$. 
Then, Eqs.~(\ref{eq:1.24a}) and~(\ref{eq:1.24b}) imply  $(s,a,\yait,t)\in\rho$, $\yal_0=\yal_1$ and $(q,a,\yait,u)\in\mu$, $\ybe_0=\ybe_1$.
From the same equations, we now obtain 
$\ypdatrtf{(t,\yde,\yal_0)}{n-1}{\yS}{(p,\yeps,\yal_n)}$ and $\ypdatrtf{(u,\yde,\ybe_0)}{n-1}{\yQ}{(r,\yeps,\ybe_m)}$.
The induction hypothesis now implies $\ypdatrtf{((t,u),\yde,\yga_0)}{\star}{\yS\times\yQ}{((p,r),\yeps,\yga)}$.
Since $(s,a,\yait,t)\in\rho$ and $(q,a,\yait,u)\in\mu$ and $a\neq \yeps$, Definition~\ref{def:productVPA} item (1)  
gives $((s,q),a,\yait,(t,u))$ in $\yS\times \yQ$, and we may write $\ypdatrtf{((s,q),a\yde,\yga_0)}{1}{\yS\times\yQ}{((t,q),\yde,\yga_0)}$.
Composing, we get the result.

\item[\sc Case 4:] $a\in A_c$.
From  Eq.~(\ref{eq:1.24a}) we get $(s,a,X,t)\in\rho$, $\ypdatrtf{(s,a\yde,\yal_0)}{1}{\yS}{(t,\yde,X\yal_0)}$ for some $X\in\yGa$.
And from Eq. (\ref{eq:1.24b}) we get $(q,a,Y,u)\in\mu$, $\ypdatrtf{(q,a\yde,\ybe_0)}{1}{\yQ}{(u,\yde,Y\ybe_0)}$for some  $Y\in\yDe$.

Definition~\ref{def:productVPA} item (1) says that $((s,q),a,(X,Y),(t,u))$ is in $\yS\times \yQ$, and we may write 
$\ypdatrtf{((s,q),a\yde,\yga_0)}{1}{\yS\times\yQ}{((t,u),\yde,(X,Y)\yga_0)}$.
Inductively, Eqs.~(\ref{eq:1.24a}) and~(\ref{eq:1.24b}) also imply $\ypdatrtf{((t,u),\yde,(X,Y)\yga_0)}{\star}{\yS\times\yQ}{((p,r),\yeps,\yga)}$.
Composing again, the result follows.

\item[\sc Case 5:] $a\in A_r$, and recall that $\yal_0=X_1\ldots X_i\bot$, $\ybe_0=Y_1\ldots Y_i\bot$, for some $i\geq 0$.
When $i\geq 1$, the reasoning is very similar to Case 4.

Let $i=0$. 
From Eq. (\ref{eq:1.24a}) we see that $\yal_0=\yal_1=\bot$, $(s,a,\bot,t)\in\rho$ and $\ypdatrtf{(t,\yde,\bot)}{n-1}{\yS}{(p,\yeps,\yal_n)}$.
Likewise, From Eq. (\ref{eq:1.24b}) gives $\ybe_0=\ybe_1=\bot$, $(q,a,\bot,u)\in\mu$ and $\ypdatrtf{(u,\yde,\bot)}{m-1}{\yQ}{(r,\yeps,\ybe_m)}$.
Inductively, we have
$\ypdatrtf{((t,u),\yde,\bot)}{\star}{\yS\times\yQ}{((p,r),\yeps,\yga)}$.
Because $(s,a,\bot,t)\in\rho$, $(q,a,\bot,u)\in\mu$,
Definition~\ref{def:productVPA} item (1) says that $((s,q),a,\bot,(t,u))$
is in $\yS\times \yQ$.
Since $\ysi=a\yde$, composing we get $\ypdatrtf{((s,q),\ysi,\bot)}{\star}{\yS\times\yQ}{((p,r),\yeps,\yga)}$, which is the desired result. 
\end{description}
The proof is now complete.
\end{proof}

\subsection{Proof of Proposition~\ref{prop:stack-size}}\label{app:prop:stack-size}
\begin{proof}
	A simple induction on $n=\vert\ysi\vert$. 
	When $n=0$ we get $\ysi=\yeps$.
	According to Definition~\ref{def:fsa-move},  $\yeps$-moves do not change the stack, and so we get $\yal_1=\yal_2$ and $\ybe_1=\ybe_2$, and the result follows.
	
	Assume $n\geq 1$ and $\ysi=\yde a$, where $a\in A $ and $\vert\yde\vert=n-1$.
	Then computations can be separated thus
	\begin{align*}
		&\ypdatrtf{(s_1,\yde a,\yal_1\bot)}{\star}{\yA}{(s_3,a,\yal_3\bot)}\ypdatrtf{}{1}{\yA}{(s_4,\yeps,\yal_4\bot)}
		\ypdatrtf{}{\star}{\yA}{(s_2,\yeps,\yal_2\bot)}\\ 
		&\ypdatrtf{(q_1,\yde a,\ybe_1\bot)}{\star}{\yB}{(q_3,a,\ybe_3\bot)}\ypdatrtf{}{1}{\yB}{(q_4,\yeps,\ybe_4\bot)}
		\ypdatrtf{}{\star}{\yB}{(q_2,\yeps,\ybe_2\bot)}.
	\end{align*}
	Clearly, we get $\ypdatrtf{(s_1,\yde,\yal_1\bot)}{\star}{\yA}{(s_3,\yeps,\yal_3\bot)}$ and $\ypdatrtf{(q_1,\yde,\ybe_1\bot)}{\star}{\yB}{(q_3,\yeps,\ybe_3\bot)}$.
	The induction hypothesis implies $\vert\yal_3\vert=\vert\ybe_3\vert$.
	The proof will be complete if we can show that $\vert\yal_4\vert=\vert\ybe_4\vert$, because from these configurations onward we have only $\yeps$-moves, which would imply that $\vert \yal_2\vert = \vert\ybe_2\vert$.
	
	Let $(s_3,a,Z,s_4)$ and $(q_3,a,W,q_4)$ be the $\yA$ and $\yB$ transitions, respectively, used in the one-step 
	computations, as indicated above.
	From Definition~\ref{def:fsa-move} we have three simple cases: (i) when $a\in A_i\cup\{\yeps\}$ we get $\yal_3=\yal_4$ and $\ybe_3=\ybe_4$; 
	(ii) when $a\in A_c$ we get $\yal_4=Z\yal_3$ and $\ybe_4=W\ybe_3$; 
	and (iii) when $a\in A_r$, since $\vert \yal_3\vert=\vert\ybe_3\vert$, we get either $\yal_3\neq \yeps\neq\ybe_3$ and then $\yal_3=Z\yal_4$ and $\ybe_3=W\ybe_4$, or 
	$\yal_3= \yeps=\ybe_3$ and then $\yal_4= \yeps=\ybe_4$.
	In any case we see that $\vert\yal_4\vert=\vert\ybe_4\vert$, because we already have $\vert\yal_3\vert=\vert\ybe_3\vert$.
\end{proof}

\subsection{Proof of Proposition~\ref{prop:cup-vpa}}\label{app:prop:cup-vpa}
\begin{proof}
Let $\yS=\yvpaA$ and $\yQ=\yvpaB$.
Using Proposition~\ref{prop:non-blocking} we can assume that $\yS$ and $\yQ$ are non-blocking VPAs with $n+1$ and $m+1$ states, respectively.

Let $\yP$ be the product of $\yS$ and $\yQ$ as in Definition~\ref{def:productVPA}, except that we redefine the final states of $\yP$ as $(F\times Q)\cup (S\times G)$.
Clearly, $\yP$ has $(n+1)(m+1)$ states.

We now argue that $\yP$ is also a non-blocking VPA.
Let $\ysi\in A^\star$, $(s,q)\in S\times Q$, and let $\yga=(Z_1,W_1)\ldots(Z_k,W_k)\in (\yGa\times\yDe)^\star$.
Since $\yS$ is a non-blocking VPA, we get $n\geq 0$, $p\in S$, $\yal\in\yGa^\star$ such that  
$\ypdatrtf{(s,\ysi,Z_1\ldots Z_k\bot)}{n}{\yS}{(p,\yeps,\yal\bot)}$. 
Likewise, $\ypdatrtf{(q,\ysi,W_1\ldots W_k\bot)}{m}{\yQ}{(r,\yeps,\ybe\bot)}$, for some
$m\geq 0$, $r\in Q$, $\ybe\in\yDe^\star$.
Applying Proposition~\ref{prop:stack-size} we get $\vert \yal\vert=\vert\ybe\vert$, and then using  Proposition~\ref{prop:product-behavior} we have 
$\ypdatrtf{((s,q),\ysi,\yga\bot)}{\star}{\yP}{((p,r),\yeps,\yde\bot)}$ where $\yga=(Z_1,W_1)\ldots (Z_k,W_k)$ and  $\yde\in (\yGa\times\yDe)^\star$.
This shows that $\yP$ is a non-blocking VPA.

Now suppose that $\ysi\in L(\yP)$, that is $\ypdatrtf{((s_0,q_0),\ysi,\bot)}{\star}{\yP}{((p,r),\yeps,\yal\bot)}$, where $(s_0,q_0)\in S_{in}\times Q_{in}$,  $(p,r)\in (F\times Q)\cup(S\times G)$, and $\yal\in(\yGa\times\yDe)^\star$.
Take the case when $(p,r)\in (F\times Q)$.
We get $p\in F$ and $s_0\in S_{in}$. 
Using Proposition~\ref{prop:product-behavior} we can also write $\ypdatrtf{(s_0,\ysi,\bot)}{\star}{\yS}{(p,\yeps,\ybe\bot)}$, for some $\ybe\in \yGa^\star$.
This shows that $\ysi\in L(\yS)$.
By a similar reasoning, when $(p,r)\in (S\times G)$ we get $\ysi\in L(\yQ)$.
Thus, $L(\yP)\ysse L(\yS)\cup L(\yQ)$. 

Now let $\ysi\in  L(\yS)\cup L(\yQ)$.
Take the case $\ysi\in L(\yS)$, the case $\ysi\in L(\yQ)$ being similar.
Then we must  have $\ypdatrtf{(s_0,\ysi,\bot)}{n}{\yS}{(p,\yeps,\yal\bot)}$ for some $n\geq 0$, some $\yal\in \yGa^\star$, and some $p\in F$.
Pick any $q_0\in Q_{in}$.
Since $\yQ$ is a non-blocking VPA, Definition~\ref{def:forward} gives some $r\in Q$ and some $\ybe\in \yDe^\star$ such that $\ypdatrtf{(q_0,\ysi,\bot)}{m}{\yQ}{(r,\yeps,\ybe\bot)}$ for some $m\geq 0$, some $r\in Q$ and some $\ybe\in\yDe^\star$. 
Using Proposition~\ref{prop:stack-size} we conclude that $\vert \yal\vert=\vert\ybe\vert$.
The only-if part of Proposition~\ref{prop:product-behavior} now yields  
$\ypdatrtf{((s_0,q_0),\ysi,\bot)}{\star}{\yP}{((p,r),\yeps,\yga\bot)}$, where $\yga\in (\yGa\times\yDe)^\star$.
Clearly, $(s_0,q_0)\in S_{in}\times Q_{in}$ is an initial state of $\yP$ and $(p,r)\in  (F\times Q)\cup (S\times G)$ is a final state of $\yP$.
Hence $\ysi\in L(\yP)$. 
Thus $ L(\yS)\cup L(\yQ)\ysse L(\yP)$, and so $ L(\yS)\cup L(\yQ)=L(\yP)$.

Applying Proposition~\ref{prop:epsilon-deterministic} (2) we see that when $\yS$ and $\yQ$ are deterministic, then $\yP$ is also deterministic and has no $\yeps$-moves.

The proof is now complete.
\end{proof}

\subsection{Proof of Proposition~\ref{prop:compl-vpa}}\label{app:prop:compl-vpa}
\begin{proof}
Applying Propositions~\ref{prop:no-eps-determ} and~\ref{prop:non-blocking}, we can assume that $\yS$ is a non-blocking and deterministic VPA with $n+1$ states and no $\yeps$-moves.

Let $\yQ=\yvpa{S}{S_{in}}{A}{\yGa}{\rho}{S-F}$, that is, we switch the final states of $\yS$.
Clearly, since $\yS$ and $\yQ$ have the same set of initial states and the same transition relation, 
we see that $\yQ$ is also a non-blocking and deterministic VPA with $n+1$ states and no $\yeps$-moves.

Pick any $\ysi\in L(\yQ)$, and let $n=\vert \ysi\vert$.
Since $\yQ$ has no $\yeps$-moves, we must have $\ypdatrtf{(q_0,\ysi,\bot)}{n}{\yQ}{(r,\yeps,\yal\bot)}$ for some $q_0\in S_{in}$, $\yal\in \yGa^\star$, and some $r\in S-F$.
For the sake of contradiction, assume that $\ysi\in L(\yS)$.
Then, because $\yS$ has no $\yeps$-moves, we must have $\ypdatrtf{(s_0,\ysi,\bot)}{n}{\yS}{(p,\yeps,\ybe\bot)}$ for some $s_0\in S_{in}$, $\ybe\in \yGa^\star$, and some $p\in F$.
Since $\yS$ is deterministic we have $\vert S_{in}\vert=1$, so that $q_0=s_0$.
Since $\yS$ and $\yQ$ have exactly the same set of transitions, we can now write
$\ypdatrtf{(s_0,\ysi,\bot)}{n}{\yS}{(r,\yeps,\yal\bot)}$.
Together with $\ypdatrtf{(s_0,\ysi,\bot)}{n}{\yS}{(p,\yeps,\ybe\bot)}$ and Proposition~\ref{prop:vpa-determ}, we conclude that $p=r$, which is a contradiction since $r\in S-F$ and $p\in F$.
Thus, $L(\yQ)\ysse \ycomp{L(\yS)}$.

Now let $\ysi\not\in L(\yS)$.
Since $\yS$ is non-blocking, we have $\ypdatrtf{(s_0,\ysi,\bot)}{\star}{\yS}{(p,\yeps,\yal\bot)}$ for some $p\in S$ and $\yal\in \yGa^\star$.
Because $\ysi\not\in L(\yS)$ we need $p\in S-F$.
Since $\yS$ and $\yQ$ have the same set of transitions, we can write $\ypdatrtf{(s_0,\ysi,\bot)}{n}{\yQ}{(p,\yeps,\yal\bot)}$.
Thus, $\ysi\in L(\yQ)$ because $p\in S-F$. 
Hence, $ \ycomp{L(\yS)}\ysse L(\yQ)$, and the proof is complete. 
\end{proof}

\subsection{Proof of Proposition~\ref{prop:conct-vpa}}\label{app:prop:conct-vpa}
\begin{proof}
We have $\yS=\yvpaA$ and $B\ysse A$.
		From Proposition~\ref{prop:non-blocking} we can assume that $\yS$ has $n+1$ states, is non-blocking and has no $\yeps$-moves. 
	Define $Q=S \cup \hat{S}$, where $\hat{S}=\{\hat{s}\yst s\in S\}$  is a set of new, distinct states.
	These new states are the only final states of $\yQ$, that is, we let $\yQ=\yvpa{Q}{S_{in}}{A}{\yGa}{\mu}{\hat{S}\,}$.
We start with $\mu=\rho$ and apply the following steps:
	\begin{enumerate}
		\item {\sf Replace some original transitions:} For all $s\in F$ and all $b\in B$:
		if  $(s,b,Z,t)\in \rho$, then add $(s,b,Z,\hat{t}\,)$ to $\mu$, and remove $(s,b,Z,t)$ from $\mu$; 
		\item {\sf Create new transitions:} For all $(t,a,Z,r)\in\rho$, if
		\begin{enumerate}
			\item $a\in B$ and $t\in F$: we add $(\hat{t},a,Z,\hat{r})$ to $\mu$; 
			\item $a\not\in B$ or $t\not\in F$: we add $(\hat{t},a,Z,r)$ to $\mu$. 
		\end{enumerate}
	\end{enumerate}
	It is clear that $\yQ$ has $2n+2$ states.
	Since no $\yeps$-moves are added to $\mu$, we see that $\yQ$ has no $\yeps$-moves.
	
	In order to show that $L(\yQ)= L(\yS)B$ we need to relate the computations of $\yS$ and $\yQ$. 
	To ease the notation we let $\tilde{r}$ stand for either $r$ or $\hat{r}$, that is, $\tilde{r}\in\{r,\hat{r}\}$.
	\begin{description}
		\item[\sf Claim 1:] If $\ypdatrtf{(s,\ysi,\yal\bot)}{\star}{\yS}{(t,\yeps,\ybe\bot)}$, then for all $\tilde{s}\in\{s,\hat{s}\}$ we have some $\tilde{t}\in\{t,\hat{t}\}$ such that $\ypdatrtf{(\tilde{s},\ysi,\yal\bot)}{\star}{\yQ}{(\tilde{t},\yeps,\ybe\bot)}$.
		
		\emph{Proof.} We induct on $\vert \ysi\vert=n\geq 0$.
		When $n=0$, since $\yS$ has no $\yeps$-moves, we get $(s,\ysi,\yal\bot)=(t,\yeps,\ybe\bot)$ and the result follows easily.

		Now, let $\ysi=\yde a$ where $a\in A$, and
			$\ypdatrtf{(s,\yde a,\yal\bot)}{\star}{\yS}{(r,a,\yga\bot)}\ypdatrtf{}{1}{\yS}{(t,\yeps,\ybe\bot)}$,
		where $(r,a,Z,t)\in\rho$ is the transition used in the last step.
		The induction gives $\tilde{r}\in\{r,\hat{r}\}$ and $\ypdatrtf{(\tilde{s},\yde,\yal\bot)}{\star}{\yQ}{(\tilde{r},\yeps,\yga\bot)}$.
		
		Assume that $\tilde{r}=r$. Since  $(r,a,Z,t)\in\rho$, either $(r,a,Z,t)$ remains in $\mu$, or  step (1) of the construction gives $(r,a,Z,\hat{t}\,)\in\mu$.
		In any case, there is some $\tilde{t}\in\{t,\hat{t}\,\}$ such that $(r,a,Z,\tilde{t})$ is in $\mu$.
This gives  $\ypdatrtf{(r,a,\yga\bot)}{1}{\yQ}{(\tilde{t},\yeps,\ybe\bot)}$.
		Composing, we get 
			$\ypdatrtf{(\tilde{s},\yde a,\yal\bot)}{\star}{\yQ}{(r,a,\yga\bot)}\ypdatrtf{}{1}{\yQ}{(\tilde{t},\yeps,\ybe\bot)}$.
		Now assume $\tilde{r}=\hat{r}$. 
Since  $(r,a,Z,t)\in\rho$, step (2) of the construction implies $(\hat{r},a,Z,\tilde{t}\,)\in\mu$, for some $\tilde{t}\in\{t,\hat{t}\,\}$.
Repeating the argument we  get again  $\ypdatrtf{(\tilde{s},\ysi,\yal\bot)}{\star}{\yQ}{(\tilde{t},\yeps,\ybe\bot)}$, proving the claim. 
	\end{description}

\noindent Now let $\ysi\in L(\yS)$ and $b\in B$.
Then, we get $s\in F$ and $s_0\in S_{in}$ such that $\ypdatrtf{(s_0,\ysi ,\bot)}{\star}{\yS}{(s,\yeps,\yal\bot)}$.
Using Claim 1 we get $\tilde{s}\in\{s,\hat{s}\}$ such that $\ypdatrtf{(s_0,\ysi ,\bot)}{\star}{\yQ}{(\tilde{s},\yeps,\yal\bot)}$.
Since $\yS$ is non-blocking, we get $t\in S$ and $\ybe\in \yGa^\star$ such that
		$\ypdatrtf{(s,b,\yal\bot)}{1}{\yS}{(t,\yeps,\ybe\bot)}$.
Let $(s,b,Z,t)\in\rho$ be the transition used in this step.
First, assume that $\tilde{s}=s$.
Since $s\in F$ and $b\in B$, step (1) of the construction gives $(s,b,Z,\hat{t}\,)\in\mu$.
Thus, $(\tilde{s},b,\yal\bot)=\ypdatrtf{(s,b,\yal\bot)}{1}{\yQ}{(\hat{t},\yeps,\ybe\bot)}$.
Alternatively, if $\tilde{s}=\hat{s}$, from step (2) and with $s\in F$ and $b\in B$ we get $(\hat{s},b,Z,\hat{t}\,)\in\mu$.
Thus, $(\tilde{s},b,\yal\bot)=\ypdatrtf{(\hat{s},b,\yal\bot)}{1}{\yQ}{(\hat{t},\yeps,\ybe\bot)}$.
So, in any case, we can compose and obtain 
$\ypdatrtf{(s_0,\ysi b,\bot)}{\star}{\yQ}{(\tilde{s},b,\yal\bot)}\ypdatrtf{}{1}{\yQ}{(\hat{t},\yeps,\ybe\bot)}$.
Since $\hat{t}$ is final in $\yQ$ we see that $\ysi b\in L(\yQ)$.
Thus, $L(\yS)B\ysse L(\yQ)$.
	
	For the other direction, we need to relate computations in $\yQ$ to computations in $\yS$.
	\begin{description}
		\item[\sf Claim 2:] If $\ypdatrtf{(s,\ysi,\yal\bot)}{\star}{\yQ}{(\tilde{t},\yeps,\ybe\bot)}$ where $\tilde{t}\in\{t,\hat{t}\,\}$, then either:\\
		\begin{enumerate}
			\item $\tilde{t}=t$ and $\ypdatrtf{(s,\ysi,\yal\bot)}{\star}{\yS}{(t,\yeps,\yal\bot)}$
			\item $\tilde{t}=\hat{t}$, $\ysi=\yde b$ with $b\in B$, and $\ypdatrtf{(s,\yde b,\yal\bot)}{\star}{\yS}{(r,b,\yga\bot)}\ypdatrtf{}{1}{\yS}{(t,\yeps,\ybe\bot)}$ with $r\in F$, $\yga\in\yGa^\star$.
		\end{enumerate}
		\emph{Proof.}
		When $\ysi=\yeps$, since $\yQ$ has no $\yeps$-moves, we get $s=t$ and $\yal=\ybe$.
		The result follows easily.
		
		Now let $\ysi=\yde a$ with $a\in A$.
		Since $\yQ$ has no $\yeps$-moves, there are states $\tilde{r}$ and $\tilde{t}$ such that 
	$\ypdatrtf{(s,\yde a,\yal\bot)}{\star}{\yQ}{(\tilde{r},a,\yga\bot)}\ypdatrtf{}{1}{\yQ}{(\tilde{t},\yeps,\ybe\bot)}$,
		where $\tilde{r}\in\{r,\hat{r}\}$ and $\tilde{t}\in\{t,\hat{t}\,\}$.
		
		Assume first that $\tilde{r}=r$.
		Then, inductively, item (1) gives $\ypdatrtf{(s,\yde a,\yal\bot)}{\star}{\yS}{(r,a,\yga\bot)}$.
		Let $(r,a,Z,\tilde{t})$ be the transition in $\mu$ used in the last step $\ypdatrtf{(r,a,\yga\bot)}{1}{\yQ}{(\tilde{t},\yeps,\ybe\bot)}$.
		If $\tilde{t}=t$ it follows that $(r,a,Z,t)$ is also in $\rho$, and then we also have $\ypdatrtf{(r,a,\yga\bot)}{1}{\yS}{(t,\yeps,\ybe\bot)}$.
		Composing we get $\ypdatrtf{(s,\ysi,\yal\bot)}{\star}{\yS}{(t,\yeps,\ybe\bot)}$, and item (1) of the Claim holds.
		If $\tilde{t}=\hat{t}$ we get $(r,a,Z,\hat{t}\,)\in\mu$.
		According to item (1) of the construction, we get $r\in F$, $a\in B$ and $(r,a,Z,t)\in\rho$.
		Again, we may write  $\ypdatrtf{(r,a,\yga\bot)}{1}{\yS}{(t,\yeps,\ybe\bot)}$.
		Composing, we see that now item (2) of the Claim holds.
		
		Alternatively, let $\tilde{r}=\hat{r}$.
		Now, from $\ypdatrtf{(s,\yde,\yal\bot)}{\star}{\yQ}{(\tilde{r},\yeps,\yga\bot)}=(\hat{r},\yeps,\yga\bot)$, the induction and item (2) of the Claim give $\yde=\eta b$, $b\in B$ and 
		$\ypdatrtf{(s,\eta b, \yal\bot)}{\star}{\yS}{(q,b,\yte\bot)}\ypdatrtf{}{1}{\yS}{(r,\yeps,\yga\bot)}$ for some $q\in F$ and $\yte\in\yGa^\star$.
		Let $(\hat{r},a,Z,\tilde{t}\,)\in\mu$ be the transition  used in the last move of $\yQ$, namely $(\tilde{r},a,\yga\bot)=\ypdatrtf{(\hat{r},a,\yga\bot)}{1}{\yQ}{(\tilde{t},\yeps,\ybe\bot)}$.
		Item (2) of the construction gives $(r,a,Z,t)\in\rho$.
		Then we may also write $\ypdatrtf{(r,a,\yga\bot)}{1}{\yS}{(t,\yeps,\ybe\bot)}$.
		If $\tilde{t}=t$, recalling that $\ysi=\yde a=\eta ba$, we compose to get 
		$\ypdatrtf{(s,\eta ba, \yal\bot)}{\star}{\yS}{(q,ba,\yte\bot)}\ypdatrtf{}{1}{\yS}{(r,a,\yga\bot)}\ypdatrtf{}{1}{\yS}{(t,\yeps,\ybe\bot)}$, 
		thus showing that item (1) of the Claim holds. 
		Finally, when $\tilde{t}=\hat{t}$ we get $(\hat{r},a,Z,\hat{t}\,)\in\mu$ and item (2a) of the construction gives $a\in B$ and $r\in F$. 
		Putting it together, we have $\ysi=\eta b a$, $a\in B$, $r\in F$ and now
		$\ypdatrtf{(s,\eta b a, \yal\bot)}{\star}{\yS}{(r,a,\yga\bot)}\ypdatrtf{}{1}{\yS}{(t,\yeps,\ybe\bot)}$.
		We see that item (2) of the Claim holds, completing the proof.
	\end{description}
	Now assume that $\ysi\in L(Q)$.
	Since $S_{in}$ is the set of initial states and $\hat{S}$ is the set of final states of $\yQ$,
	we must have $\ypdatrtf{(s_0,\ysi, \bot)}{\star}{\yQ}{(\hat{t}\,,\yeps,\yal\bot)}$ for some $s_0\in S_{in}$, $\hat{t}\,\in\hat{S}$, and $\yal\in\yGa^\star$.
	From Claim 2, item 2, we get $\ysi=\yde b$, $b\in B$ and $\ypdatrtf{(s_0,\yde,\yal\bot)}{\star}{\yS}{(r,\yeps,\yga\bot)}$ with $r\in F$.
	Clearly, $\yde\in L(\yS)$, so that $\ysi\in L(\yS)B$.
	We now have $L(\yQ)\ysse L(\yS)B$.
	
	Since we already had $L(\yS)B\ysse L(\yQ)$, we conclude that $L(\yQ)= L(\yS)B$, as expected.
\end{proof}

%% file: sec3-apendice.tex

\subsection{Proof of Proposition~\ref{prop:contracted-vpts}}\label{app:prop:contracted-vpts}
\begin{proof}
We have $\yS=\yvpts{S}{S_{in}}{L}{\yGa}{T}$.
	We first construct a context-free grammar (CFG)~\cite{sipser,hopcu-introduction-1979} $G$ whose leftmost
derivations will simulate traces of $\yS$, and vice-versa.

The terminals of $G$ are the transitions in $T$. 
The non-terminals are of the form $[s,Z,p]$ where $s,p\in S$ are states of $\yS$ and $Z\in\yGa_\bot$ is a stack symbol. 
If state $p$ is not important, we may write $[s,Z,-]$.
For  the main idea, let  $t_i=[s_i,a_i,Z_i,p_i]$, $1\leq i\leq n$ be transitions of $\yS$ and let $\ysi=a_1a_2\cdots a_n$.
Then, if in $G$ we have a leftmost derivation
$$\ycfgtrtfl{[s_0,\bot,-]}{\star}{}{t_1\cdots t_n[r_1,W_1,r_2][r_2,W_2,r_3]\cdots[r_m,W_m,r_{m+1}][r_{m+1},\bot,-]}$$
it must be the case that $\yS$ starting at the initial configuration $(s_0,\bot)$ can move along the transitions $t_1, \ldots, t_n$, in that order, to reach the configuration $(r_1,W_1W_2\cdots W_m\bot)$,
that is, $\ytrtf{(s_0,\bot)}{\ysi}{\yS}{(r_1,W_1W_2\cdots W_m\bot)}$ and vice-versa.
We are also guessing that, when $\yS$ removes  some $W_i$ from the stack  --- if it eventually does, --- then it will enter state $r_i$, $i=1,\ldots, m$.

Formally, let $G=(V, T,P,I)$ where $V$ is the set of non-terminals, $T$ is the set of terminals $P$ is the set of productions and $I$ is the initial non-terminal of $G$. 
The set of terminals is the same set $T$ of transition of $\yS$.
The sets $V$ and $P$ are constructed as follows, and where $NV$ is an auxiliary set. 
Start with $V=\{I\}$, 
$NV=\{[s_0,\bot,-]:\,\,\text{for all } s_0\in S_{in}\}$, and 
$P=\{\ycfgtrtf{I}{}{}{[s_0,\bot,-]}:\,\,[s_0,\bot,-]\in NV\}$.
Next, we apply the following simple algorithm:

\noindent\makebox[\textwidth]{\rule[-4pt]{.4pt}{4pt}\hrulefill\rule[-4pt]{.4pt}{4pt}}\\
\texttt{\fontsize{8pt}{9pt}\selectfont\newline
\noindent While $NV\neq\yemp$:
\begin{enumerate} 
	\item Let $[s,Z,p]\in NV$. Remove $[s,Z,p]$ from $NV$ and add it to $V$.
	\item
	For all $t\in  T$: 
	\begin{enumerate} 
		\item If $t=(s,a,W,q)\in T_c$, add 
		$\ycfgtrtf{[s,Z,p]}{}{}{t[q,W,r][r,Z,p]}$ to $P$, for all $r\in S$. 
		If $[q,W,r] \notin V \cup NV$, add $[q,W,r]$ to $NV$, and if $[r,Z,p] \notin V \cup NV$ add $[r,Z,p]$ to $NV$; 
		\item If $t=(s,a,\yait,q)\in T_i$,  add $\ycfgtrtf{[s,Z,p]}{}{}{t[q,Z,p]}$ to $P$. 
		If $[q,Z,p] \notin V \cup NV$, add $[q,Z,p]$ to $NV$; 
		\item If $t=(s,a,Z,q)\in T_r$, then 
		\begin{enumerate}
			\item if $Z\neq \bot$ and $p=q$ add $\ycfgtrtf{[s,Z,p]}{}{}{t}$ to $P$; and 
			\item if $Z=\bot$ and $p=-$, add $\ycfgtrtf{[s,Z,p]}{}{}{t[q,\bot,-]}$ to $P$, and if
			$[q,\bot,-] \notin V \cup NV$ add $[q,\bot,-]$ to $NV$
		\end{enumerate}
	\end{enumerate}
\end{enumerate}
}
\vspace*{-3ex}\noindent\makebox[\textwidth]{\rule{.4pt}{4pt}\hrulefill\rule{.4pt}{4pt}}

We indicate by $\hookrightarrow$ the leftmost derivation relation induced by $G$
over $(V\cup  T)^\star$.

The next claim says that leftmost derivations of $G$ faithfully simulate traces of $\yS$.
\begin{description}
	\item[{\sf Claim 1:}] Let $t_i=(p_i,a_i,Z_i,q_i)$, $1\leq i\leq n$, $n\geq 0$, 
	and $\ysi=a_1a_2\cdots a_n$.
Assume that 
	\begin{equation}\label{prop2.37c1a}
		\ycfgtrtfl{[s,\bot,-]}{n}{G}{t_1t_2\cdots t_n[u_0,W_1,u_1][u_1,W_2,u_2]\cdots [u_{m-1},W_m,u_m]}.
	\end{equation}
	Then we must have:  
	\begin{align}
		&m\geq 1,  u_m=-, W_m=\bot,\,\text{and}\,\, W_i\neq \bot, 1\leq i<m \label{prop2.37c1i2}\\
		& \ytrtf{(s,\bot)}{\ysi}{\yS}{(u_0,W_1\cdots W_m)} \label{prop2.37c1i}\\
		& \text{either (i) $n=0$ with $s=u_0$; or (ii) $n\geq 1$ with $s=p_1$, $q_n=u_0$}\label{prop2.37c1i3}.
	\end{align}
	\item[{\sf Proof:}] 
	If $n=0$ then $\ysi=\yeps$, $[s,\bot,-]=[u_0,W_1,u_1]$ and $m=1$.
	Hence, $W_m=W_1=\bot$, $u_m=u_1=-$ and $s=u_0$.
	Since we can write $\ytrtf{(s,\bot)}{\ysi}{\yS}{(u_0,W_1)}$,
	the result follows.
	
	Proceeding inductively, fix some $n\geq 1$ and assume the assertive is true for $n-1$.
Now suppose that Eq.~(\ref{prop2.37c1a}) holds.
	Since derivations in $G$ are leftmost, and each production in $G$ has exactly one terminal symbol as the leftmost symbol in the right-hand side, using the induction hypothesis we can write
	\begin{align}
		\ycfgtrtfl{[s,\bot,-]&}{n-1}{G}{t_1\cdots t_{n-1}[u_0,W_1,u_1][u_1,W_2,u_2]\cdots [u_{m-1},W_m,u_m]}\notag\\
		&\ycfgtrtfl{}{1}{G}{t_1\cdots t_{n-1}t_n\ybe[u_1,W_2,u_2]\cdots [u_{m-1},W_m,u_m]}\label{prop2.37c1e},
	\end{align}
	where $\ycfgtrtf{[u_0,W_1,u_1]}{}{}{t_n\ybe}$ was the production used in the last step. 
	We also get 
	\begin{align}
		&m\geq 1,  u_m=-, W_m=\bot,\,\text{and}\,\, W_i\neq \bot, 1\leq i<m \label{prop2.37c1k2}\\
		& \ytrtf{(s,\bot)}{\ysi_1}{\yS}{(u_0,W_1\cdots W_m)},\,\, \ysi_1=a_1a_2\cdots a_{n-1}\label{prop2.37c1d}\\
		& \text{either (i) $n-1=0$, $s=u_0$; or (ii) $n-1\geq 1$ with $s=p_1$, $q_{n-1}=u_0$}\label{prop2.37c1g}.
	\end{align}
	
	By the construction of $G$ there are three cases.
	
	As a first alternative, assume $a_n\in L_c$, so that $t_n=[p_n,a_n,Z_n,q_n]\in T_c$.
	Then by item (2a) of the construction of $G$ we must have 
	$\ycfgtrtf{[u_0,W_1,u_1]}{}{}{t_n[q_n,Z_n,r][r,W_1,u_1]}$, where $r\in S$, and $p_n=u_0$.
	Thus, $\ybe=[q_n,Z_n,r][r,W_1,u_1]$.
	Together with Eq. (\ref{prop2.37c1e}) we get 
	\begin{align}\label{prop2.37c1f}
		\ycfgtrtfl{[s,\bot,-]&}{n}{G}{t_1\cdots t_{n-1}t_n[q_n,Z_n,r][r,W_1,u_1][u_1,W_2,u_2]\cdots [u_{m-1},W_m,u_m]}.
	\end{align}
	Now define $v_0=q_n$, $X_1=Z_n$, $v_1=r$.
	Also, let $v_{i+1}=u_i$ and $X_{i+1}=W_i$, $1\leq i\leq m$.
	We get
	$$\ycfgtrtfl{[s,\bot,-]}{n}{G}{t_1\cdots t_n[v_0,X_1,v_1][v_1,X_2,v_2]\cdots [v_{m},X_{m+1},v_{m+1}]}.$$
	Clearly, using condition (\ref{prop2.37c1k2}) we get $m+1\geq 1$, $X_{m+1}=W_m=\bot$ and $v_{m+1}=u_m=-$.
	Also $X_1=Z_n$ and since $t_n\in T_c$ we get $Z_n\neq \bot$, so that $X_1\neq \bot$.
	Using condition (\ref{prop2.37c1k2}) we get $X_{i+1}=W_{i}\neq \bot$, $1\leq i\leq m$, and we conclude that condition (\ref{prop2.37c1i2}) holds.
	Next, we examine condition (\ref{prop2.37c1i3}).
	We already have $n\geq 1$, $p_n=u_0$ and $v_0=q_n$. 
	If condition (\ref{prop2.37c1g}.i) holds, then $n=1$  and $s=u_0$, so that $s=p_n=p_1$, and condition (\ref{prop2.37c1i3}) holds.
	On the other hand, if condition (\ref{prop2.37c1g}.ii) holds, we immediately get $s=p_1$ and, because $v_0=q_n$, we conclude 
	condition (\ref{prop2.37c1i3}) holds again.
	Finally, since  $t_n=[p_n,a_n,Z_n,q_n]\in T_c$ and $p_n=u_0$, together with condition (\ref{prop2.37c1d}), and because $\ysi=\ysi_1a_n$, we can write 
	\begin{align*}
		\ytrtf{(s,\bot)&}{\ysi_1}{\yS}{(u_0,W_1\cdots W_m)=(p_n,W_1\cdots W_m)}\\
		\ytrtf{&}{a_n}{\yS}{(q_n,Z_nW_1\cdots W_m)=(v_0,X_1\cdots X_{m+1})}.
	\end{align*}
Hence, condition	(\ref{prop2.37c1i}) is verified, and we conclude that the claim holds in this case.
	
	For the second alternative, let $a_n\in L_i\cup\{\ytau\}$, so that $t_n=[p_n,a_n,Z_n,q_n]\in T_i$. 
	The reasoning is entirely similar to the preceding case.
	
	As a final alternative, let  $a_n\in L_r$, so that $t_n=[p_n,a_n,Z_n,q_n]\in T_r$.
	Looking at item (2c) of the construction of $G$, and recalling that $\ycfgtrtf{[u_0,W_1,u_1]}{}{}{t_n\ybe}$ is the production used to get Eq. (\ref{prop2.37c1e}), we need $u_0=p_n$ and $Z_n=W_1$.
	
	The first sub-case is when we followed step (2c.i).
We must then have $W_1\neq \bot$, $u_1=q_n$ and $\ybe=\yeps$.
	From Eq. (\ref{prop2.37c1e}) we get
	$\ycfgtrtfl{[s,\bot,-]}{n}{G}{t_1\cdots t_{n-1}t_n[u_1,W_2,u_2]\cdots [u_{m-1},W_m,u_m]}$.
	Since $W_m=\bot$ and $W_1\neq \bot$ we must have $m\geq 2$.
	Define
	$v_{i-1}=u_{i}$, $X_{i}=W_{i+1}$, $1\leq i\leq m-1$.
	Now we have 
	\begin{align}\label{prop2.37c1j}
		\ycfgtrtfl{[s,\bot,-]}{n}{G}{t_1\cdots t_n[v_0,X_1,v_1][v_1,X_2,v_2]\cdots [v_{m-2},X_{m-1},v_{m-1}]}.
	\end{align}
	Note that $m\geq 2$ implies $m-1\geq 1$, and $v_{m-1}=u_m=-$ and $X_{m-1}=W_m=\bot$. 
	Together with condition (\ref{prop2.37c1k2}) we see that
	$X_i=W_{i+1}\neq\bot$, $1\leq i\leq m-2$, so that condition (\ref{prop2.37c1i2}) holds.
	Since $m\geq 2$,  $u_0=p_n$, $W_1=Z_n$, $u_1=q_n$, and $[p_n,a_n,Z_n,q_n]\in T_r$ we get 
	$$\ytrtf{(u_0,W_1W_2\cdots W_m)}{a_n}{\yS}{(u_1,W_2\cdots W_m)=(v_0,X_1X_2\cdots X_{m-1})}.$$
	Because $\ysi=\ysi_1a_n$, we can compose with Eq. (\ref{prop2.37c1d}) and get
	$$\ytrtf{(s,\bot)}{\ysi}{\yS}{(u_1,W_2\cdots W_m)=(v_0,X_1X_2\cdots X_{m-1})},$$
	so that the condition (\ref{prop2.37c1i}) is satisfied.
	We already have $v_0=u_1$ and $u_1=q_n$, so that $v_0=q_n$.
	It is clear that $n\geq 1$, so in order to verify condition (\ref{prop2.37c1i3}) we need $s=p_1$.
	If $n=1$, condition (\ref{prop2.37c1g}) gives $s=u_0$, and since we already have $u_0=p_n=p_1$, we get $s=p_1$. 
	If $n\geq 2$ we have $n-1\geq 1$ and condition (\ref{prop2.37c1g}) immediately gives $s=p_1$, 
	concluding this sub-case.
	
	For the last sub-case, assume we followed step (2c.ii).
	In this case we get $u_0=p_n$, $W_1=Z_n=\bot$, $u_1=-$, and $\ybe=[q_n,\bot,-]$.
	Since $W_1=\bot$, condition (\ref{prop2.37c1k2}) says that $m=1$.
	We see that condition (\ref{prop2.37c1i2}) is immediately satisfied.
	From Eq. (\ref{prop2.37c1e}) we get
	$\ycfgtrtfl{[s,\bot,-]}{n}{G}{t_1\cdots t_{n-1}t_n[q_n,\bot,-]}$.
	Let $v_0=q_n$, $v_1=-$, and $X_1=\bot$.
	Thus, 
$\ycfgtrtfl{[s,\bot,-]}{n}{G}{t_1\cdots t_n[v_0,X_1,v_1]}$.
	Since $m=1$, condition (\ref{prop2.37c1d}) reduces to 
	$\ytrtf{(s,\bot)}{\ysi_1}{\yS}{(u_0,W_1)}$.
Since 	$u_0=p_n$, $W_1=Z_n=\bot$, and $[p_n,a_n,Z_n,q_n]\in  T_r$ we get
	$$\ytrtf{(s,\bot)}{\ysi_1}{\yS}{(u_0,W_1)=(p_n,\bot)}\ytrtf{}{a_n}{\yS}{(q_n,\bot)=(v_0,\bot)},
	$$
	and we see that condition (\ref{prop2.37c1i}) holds.
	As a final step, we verify that condition (\ref{prop2.37c1i3}) also holds.
	Clearly, $n\geq 1$, and we already have $v_0=q_n$.
	Again, if $n=1$, condition (\ref{prop2.37c1g}) gives $s=u_0$, and since we already have $u_0=p_n=p_1$, we get $s=p_1$. 
	If $n\geq 2$ we have $n-1\geq 1$ and condition (\ref{prop2.37c1g}) immediately gives $s=p_1$, 
	concluding the argument for this sub-case.
	This completes the argument for the last alternative.
	
	Since we examined all three alternatives in case (2c), we conclude that the claim holds. \yfim
\end{description}

For the converse, we show that any trace of $\yS$ can be simulated by a leftmost derivation of $G$.
\begin{description}
	\item[{\sf Claim 2:}] Let $\ytrtf{(s_0,\bot)}{\ysi}{\yS}{(u_0,W_1\cdots W_m\bot)}$ with $s_0\in S_{in}$, $m\geq 0$, 
	$\ysi=a_1a_2\cdots a_n\in L_\ytau^\star$, and $n\geq 0$. 
	Assume that the transitions used in this trace were, in order, $t_i=(p_i,a_i,Z_i,q_i)\in  T$, $1\leq i\leq n$.
	Then for all $u_i\in S$, $1\leq i\leq m$, we have
	\begin{align}
		&\ycfgtrtfl{[s_0,\bot,-]}{n}{G}{t_1t_2\cdots t_n[u_0,W_1,u_1][u_1,W_2,u_2]\cdots [u_{m-1},W_m,u_m][u_m,\bot,-]}\label{prop2.37-claim2a}\\
		&\text{If $n\geq 1$ then  $s_0=p_1$, $u_0=q_n$.} \label{prop2.37-claim2c}
	\end{align}
	\item[{\sf Proof:}]
	Assume first that $n=0$, so that $\ysi=\yeps$.
Then $\ytrtf{(s_0,\bot)}{\ysi}{\yS}{(u_0,W_1\cdots W_m\bot)}$  implies $s_0=u_0$, $m=0$.
Since we can write $s_0\in S_{in}$, we get that $[s_0,\bot,-]$ is a non-terminal of $G$.
	Also, 
	Since $\ycfgtrtfl{[s_0,\bot,-]}{0}{G}{[u_0,\bot,-]}$, condition (\ref{prop2.37-claim2a}) holds with $m=0$.
	Condition (\ref{prop2.37-claim2c}) holds vacuously.
	
	Now assume $n\geq 1$. 
	Since $t_n$ was the transition used in the last step, we can write 
	\begin{align}
		&\ytrtf{(s_0,\bot)}{\ysi_1}{\yS}{(p_n,X_1\cdots X_k\bot)}\ytrtf{}{a_n}{\yS}{(q_n,\ybe \bot)},\label{prop2.37-claim2d}
	\end{align}
	where $\ysi_1=a_1\cdots a_{n-1}$, $k\geq 0$, and $u_0=q_n$.
	From the induction hypothesis, condition (\ref{prop2.37-claim2a}), for all $u_i\in S$, $1\leq i\leq k$, we get 
	\begin{align}
		\ycfgtrtfl{[s_0,\bot,-]}{n-1}{G}{t_1\cdots t_{n-1}[p_n,X_1,u_1]\cdots [u_{k-1},X_k,u_k][u_k,\bot,-]}. \label{prop2.37-claim2e}
	\end{align}
	If $n=1$, so that $n-1=0$, Eq. (\ref{prop2.37-claim2e}) says that 
	$s_0=p_n=p_1$.
	If $n-1>0$, the induction hypothesis in Eq. (\ref{prop2.37-claim2e})
	says that $s_0=p_1$.
	In any case, $s_0=p_1$. 
	Since we already have $u_0=q_n$, we conclude that condition (\ref{prop2.37-claim2c}) always holds.
	
	Next, we argue that either $k>0$ and $[p_n,X_1,u_1]$ was added to the set $NV$, or $k=0$ and $[p_n,\bot,-]$ was added to $NV$  during the construction of $G$.
	If $n-1=0$ then $[p_n,X_1,u_1]=[s_0,\bot,-]$. 
	Since $s_0\in S_{in}$ we know that $[p_n,\bot,-]$ was added to the initial $NV$ set. 
	If $n-1>1$ then  $[p_n,X_1,u_1]$ was on the right-hand side of a production of $G$ and so it was also added to $NV$ during the construction of $G$.
	In any case, we can assume that $[p_n,X_1,u_1]$ was added to the set $NV$. 
	
	Following step (2) in the construction of $G$, we break the argument in there cases:
	\begin{description}
		\item[$a_n\in L_c$:] 
		in this case we have $t_n=(p_n,a_n,Z_n,q_n)\in T_c$.
		Pick any $v\in S$.
		Following item (2a) in the construction of $G$,  we must have  
		$\ycfgtrtf{[p_n,X_1,u_1]}{}{}{t_n[q_n,Z_n,v][v,X_1,u_1]}$ as a production of $G$.
		Together with Eq. (\ref{prop2.37-claim2e}) we can now write 	
		\begin{align}
			\ycfgtrtfl{[s_0,\bot,-]&}{n}{G}{t_1\cdots t_{n-1}t_n[q_n,Z_n,v][v,X_1,u_1]\cdots [u_{k-1},X_k,u_k][u_k,\bot,-]
			},\notag
		\end{align}
for all $v\in S$, and all $u_i\in S$ $i=1,\cdots, k$.
		Because $t_n$ was the last transition used in Eq. (\ref{prop2.37-claim2d}), we must have $\ybe=Z_nX_1\cdots X_k$.
Hence, we can use Eq. (\ref{prop2.37-claim2d}) and write 
		$\ytrtf{(s_0,\bot)}{\ysi}{\yS}{(q_n,Z_nX_1\cdots X_k\bot)}$, where $\vert \ysi \vert=n$.
		We conclude that condition (\ref{prop2.37-claim2a}) holds.
		\item[$a_n\in L_i\cup\{\ytau\}$:]  we have $t_n=(p_n,a_n,Z_n,q_n)\in T_i$ and we can reason as in the preceding case.
		
		\item[$a_n\in L_r$:]	
		now we have $t_n=(p_n,a_n,Z_n,q_n)\in T_r$.
		
		First assume that $k=0$ in Eq. (\ref{prop2.37-claim2d}).
		Since $t_n$ was the last transition used in that equation, we need
		$Z_n=\bot$ and $\ybe=\yeps$.
		We now have $\ytrtf{(s_0,\bot)}{\ysi}{\yS}{(q_n,\bot)}$.
		We already have $u_0=q_n$ and, with $k=0$, Eq. (\ref{prop2.37-claim2e}) reduces to 
		$\ycfgtrtfl{[s_0,\bot,-]}{n-1}{G}{t_1\cdots t_{n-1}[p_n,\bot,-]}$.
		We now have $[p_n,\bot,-]$ as a non-terminal of $G$, and $t_n=(p_n,a_n,\bot,q_n)\in T_r$.
		By item (2c.ii) of the construction of $G$, it follows that
		$\ycfgtrtf{[p_n,\bot,-]}{}{}{t_n[q_n,\bot,-]}$ is a production of $G$.
		Composing, we get   $\ycfgtrtfl{[s_0,\bot,-]}{n}{G}{t_1\cdots t_{n}[q_n,\bot,-]}$, and condition (\ref{prop2.37-claim2a}) holds with $m=0$.
		
		Lastly, assume $k\geq 1$ in Eq. (\ref{prop2.37-claim2d}).
		From Definition~\ref{def:simplemove}, it is clear that $X_1\neq \bot$.
		Again, because $t_n$ was the last transition used in Eq. (\ref{prop2.37-claim2d}), we must have $Z_n=X_1\neq \bot$ and $\ybe=X_2\cdots X_k$, so that we now have $\ytrtf{(s_0,\bot)}{\ysi}{\yS}{(q_n,X_2\cdots X_k\bot)}$.
		Recall that $[p_n,X_1,u_1]$ is a non-terminal of $G$, and Eq. (\ref{prop2.37-claim2e}) holds for all $u_i\in S$, $1\leq i\leq k$.
		Choose $u_1=q_n$, so that we have $[p_n,X_1,u_1]$ as a non-terminal of $G$, and $t_n=(p_n,a_n,X_1,u_1)$.
		Now,  using item (2c.i) of the construction of $G$ we see that
		$\ycfgtrtf{[p_n,X_1,u_1]}{}{}{t_n}$ is a production of $G$.
		Composing with Eq. (\ref{prop2.37-claim2e}) we obtain 
		$$\ycfgtrtfl{[s_0,\bot,-]}{n}{G}{t_1\cdots t_{n}[q_n,X_2,u_2]\cdots [u_{k-1},X_k,u_k][u_k,\bot,-]},$$ 
		and condition (\ref{prop2.37-claim2a}) holds with 
		$m=k-1\geq 0$, concluding this case.
	\end{description} 
\end{description}

Now we can extract a contracted VPTS $\yQ=\yvpts{Q}{Q_{in}}{A}{\yGa}{ R}$ from the original VPTS $\yS$.
First, we determine the set $LN$ of all non-terminals of $G$ that
appear in the leftmost position in a derivation of $G$, that is,
\begin{align}
	LN=\big\{[s,Z,p]\in NV\yst \ycfgtrtfl{I}{\star}{G}{t_1t_2\cdots t_n[s,Z,p]\yal, \yal \in V^\star}\big\}.\label{prop2.37n}
\end{align}
This can be accomplished by a backward search on the productions of $G$ to find all non-terminals $[s,Z,p]$ of $G$ that generate at least one string
of terminals, that is, $\ycfgtrtfl{[s,Z,p]}{\star}{G}{t_1t_2\cdots t_n}$ for some $t_i\in  R$, $1\leq i\leq n$.
Next, a forward search on the productions of $G$ collects all non-terminals in $LN$.

In a second step, we collect the transitions in $R$ as follows.

\noindent\makebox[\textwidth]{\rule[-4pt]{.4pt}{4pt}\hrulefill\rule[-4pt]{.4pt}{4pt}}\\
\texttt{\fontsize{8pt}{9pt}\selectfont\newline
Start with $ R=\yemp$. \\
Next, for all $[s,Z,p]\in LN$ and all 
$t=(s,a,W,q)\in  T$, add $t$ to $ R$ if:
\begin{enumerate}
	\item[3.] $a\in L_c\cup L_i\cup\{\ytau\}$, or
	\item[4.] $a\in L_r$ and either (i) $Z=W\neq \bot$ and $p=q$; or (ii) $Z=W=\bot$ and $p=-$.   
\end{enumerate}
}
\vspace*{-3ex}\noindent\makebox[\textwidth]{\rule{.4pt}{4pt}\hrulefill\rule{.4pt}{4pt}}

In the resulting directed graph formed by all productions in $ R$, remove any state that is not reachable from an initial state in $S_{in}$, and name $Q$ the set of  remaining states.
Finally, let 
$Q_{in}=S_{in}$.  

We can now show that $\yQ$ is contracted and equivalent to $\yS$, as needed.
\begin{description}
	\item[{\sf Claim 3:}] $\yQ$ is a contracted VPTS.
	\item[{\sf Proof:}]
	Let $t=(s,a,W,q)\in R_r$, so that $a\in\ L_r$.
	By step (4) above, we need $[s,W,p]$ leftmost in $G$, and either 
	(i) $W\neq \bot$ and $p=q$, or (ii) $W=\bot$ and $p=-$.
	
	Since $[s,W,p]$ is leftmost in $G$, by the form of the productions in $G$, we need $t_i=(p_i,a_i,Z_i,q_i)$, $1\leq i\leq n$, $n\geq 0$, $[u_{j-1},W_j,u_j]$, $1\leq j\leq m$, $m\geq 1$, such that
	$$\ycfgtrtfl{[s_0,\bot,-]}{n}{G}{t_1t_2\cdots t_n[u_0,W_1,u_1][u_1,W_2,u_2]\cdots [u_{m-1},W_m,u_m]},$$
	with $s_0\in Q_{in}$ and $[s,W,p]=[u_0,W_1,u_1]$.
Hence, $s=u_0$ and $W_1=W$.
	
	Using Claim 1 we can write  $\ytrtf{(s_0,\bot)}{\ysi}{\yS}{(u_0,W_1W_2\cdots W_m)}
	=(s,WW_2\cdots W_m)$, and where $\ysi=a_1a_2\cdots a_n$.
	Let $\mu=h_{\ytau}(\ysi)$.
	We now have $\ytrut{(s_0,\bot)}{\mu}{\yS}{(s,WW_2\cdots W_m)}$.
	If $m>1$, and remembering that $t=(s,a,W,q)$, we see that condition (i) of Definition~\ref{def:vpts-reduced} is immediately satisfied.
	When $m=1$, condition \ref{prop2.37c1i2} in Claim 1 says that $W_m=\bot$, and we now have 
	$W_m=W_1=W=\bot$ and $\ytrut{(s_0,\bot)}{\mu}{\yS}{(s,\bot)}$.
Since $t=(s,a,\bot,q)$, we see that condition (ii) of Definition~\ref{def:vpts-reduced} can also be satisfied.
\end{description}

\begin{description}
	\item[{\sf Claim 4:}] $tr(\yQ)=tr(\yS)$ and $otr(\yQ)=otr(\yS)$. 
	\item[{\sf Proof:}]
	We trivially obtain that  $tr(\yQ)\ysse tr(\yS)$ since $ R \subseteq  T$, $Q_{in}=S_{in}$, and $Q\ysse S$. 
	
	Now assume $\ysi\in tr(\yS)$ with  $\ysi=a_1\cdots a_n$, $n\geq 0$.
	Then we have 
	\begin{equation}\label{prop2.37c}
		\ytrtf{(s_0,\bot)}{\ysi}{\yS}{(f,W_1\cdots W_m\bot)}
	\end{equation}
	with $m\geq 0$, and $s_0\in S_{in}$. 
	Let $t_i=(p_i,a_i,Z_i,q_i)\in  T$, $1\leq i\leq n$, be the transitions used in this trace of $\yS$, and in this order.
	Using Claim 2, we get $r_i\in S$, $1\leq i\leq m$,  
	such that 
	\begin{equation}\label{prop2.37b}
		\ycfgtrtfl{[s_0,\bot,-]}{n}{G}{t_1t_2\cdots t_n[f,W_1,r_1][r_1,W_2,r_2]\cdots [r_{m-1},W_m,r_m][r_m,\bot,-]}.
	\end{equation}
	We want to show that $t_k$ is also a transition of $\yQ$, $1\leq k\leq n$.
That is, we want to show that $t_k$ was added to $R$ according to rules (3) and (4) above.
	Fix some $k$, $1\leq k\leq n$, with $t_k=(p_k,a_k,Z_k,q_k)$.
	Since the derivation in Eq.~(\ref{prop2.37b}) is leftmost, and $G$ has exactly one terminal in the right-hand side of any production, we must have
	$$\ycfgtrtfl{[s_0,\bot,-]}{k-1}{G}{t_1\cdots t_{k-1}[u_1,X_1,u_2][u_2,X_2,u_3]\cdots [u_\ell,X_\ell,-]},$$
and the next production used in Eq.~(\ref{prop2.37b}) was $\ycfgtrtf{[u_1,X_1,u_2]}{}{}{t_k\ybe}$, for some $\ybe\in V^\star$.
From the construction of $G$ we always have $u_1=p_k$, so that  $[u_1,X_1,u_2]=[p_k,X_1,u_2]$ is leftmost in $G$.
From Eq. (\ref{prop2.37n}) we get $[p_k,X_1,u_2]\in LN$.
We now look at rules (3) and (4).
If $a_k\in L_c\cup L_i\cup\{\ytau\}$, rule (3) gives $t_k\in  R$.
Now $a_k\in L_r$ and recall that $\ycfgtrtf{[u_1,X_1,u_2]}{}{}{t_k\ybe}$ is a production of $G$.
By rule (2c) we  have either: 
	\begin{itemize}
		\item[(i)]   $u_2=q_k$, $X_1=Z_k\neq \bot$ using rule (2c.i).
Then, $[u_1,X_1,u_2]=[p_k,Z_k,q_k]\in LN$.
Since $t_k=(p_k,a_k,Z_k,q_k)$ and  $a_k\in L_r$, rule (4.i) says that $t_k\in  R$. 
		\item[(ii)] $u_2=-$, $X_1=Z_k=\bot$ using rule (2c.ii).
	Now $[u_1,X_1,u_2]=[p_k,Z_k,-]\in LN$.
Again, $t_k=(p_k,a_k,Z_k,q_k)$ and		$a_k\in L_r$ give $t_k\in R$ using rule (4.ii).
	\end{itemize}
	Now, we have  $t_k\in R$ for all $1\leq k\leq n$, $s_0\in S_{in}$ and $Q_{in}=S_{in}$.
	Then, as in Eq.~(\ref{prop2.37c}) we can now write
	$\ytrtf{(s_0,\bot)}{\ysi}{\yQ}{(f,W_1\cdots W_m\bot)}$.
	Thus $\ysi\in tr(\yQ)$, and we now have $tr(\yS)\ysse tr(\yQ)$. 
\end{description}

Now, if $tr(\yS)=tr(\yQ)$ then $otr(\yS)=otr(\yQ)$ follows easily from 
Definitions~\ref{def:path} and \ref{def:trace}.

Finally, assume that $\yS$ is deterministic. 
Since $\yQ$ is constructed by removing transitions from $\yS$, it is clear from Definition~\ref{def:vpts-determinism} that $\yQ$ is also deterministic.
This completes the proof.
\end{proof}

%% file: sec5-apendice.tex

\subsection{Proof of Lemma~\ref{lemm:ioco-complete}}\label{app:lemm:ioco-complete}
\begin{proof}
	We will construct a fault model $\yT$.
According to Definition~\ref{def:passes}, in order for $\yT$ to be $\yiocolike$ complete for $\yB$ we need that, for all implementations $\yI$, it holds that $\yI$ passes $\yT$ if and only if $\yI \yiocolike \yB$.
From Corollary~\ref{coro:ioco-charac} we know that $\yI \yiocolike \yB$ if and only if 
$otr(\yI)\cap T=\yemp$, where $T=\ycomp{otr}(\yB)\cap [otr(\yB) L_U]$.
That is, we need $\yT$ such that, for all implementations $\yI$ we have $\yI$ passes $\yT$ if and only if $otr(\yI)\cap T=\yemp$.

Let $\yB=\yiovpts{S_\yB}{S_{in}}{L_I}{L_U}{\yDe_\yB}{T_\yB}$ be the given deterministic specification, $L=L_I\cup L_U$ and $n=\vert S_\yB\vert$.
Using Proposition~\ref{prop:vpts-deterministic} we know that $\yB$ is a deterministic VPTS with no $\ytau$-moves
\footnote{To keep the notation unclutered, we will refer to an IOVPTS and to its associated VPTS by the same symbol.}.  
The desired fault model $\yT=\yiovpts{S_\yT}{T_{in}}{L_U}{L_I}{\yDe_\yT}{T_\yT}$ is constructed as follows. 
Let $T_{in}=S_{in}$, $\yDe_\yT=\yDe_\yB$, and  extend the state set $S_\yB$ and the transition set $T_\yB$ as follows.
Define $S_\yT=S_\yB\cup \{\yfail\}$ where $\yfail\not\in S_\yB$.
Fix some symbol $Z\in\yDe_\yB$, and let
\begin{align}
T_\yT = & T_\yB\notag\\
 &\cup \big\{(s,\ell,Z,\yfail)\yst \ell\in L_U\cap L_c \text{ and } (s,\ell,W,p)\not\in T_\yB \text{ for any  $p\in S_\yB$, any $W\in\yDe_\yB$}\big\}\label{eq:testpurpose1:Lc}\\
&\cup \big\{(s,\ell,W,\yfail)\yst \ell\in L_U\cap L_r \,\, \text{and  $(s,\ell,W,p)\not\in T_\yB$	for any  $p\in S_\yB$} \big\}\label{eq:testpurpose1:Lr}\\
&\cup \big\{(s,\ell,\yait,\yfail)\yst \ell\in L_U\cap L_i \text{ and } (s,\ell,\yait,p)\not\in T_\yB \text{ for any } p\in S_\yB\big\}\label{eq:testpurpose1:Li}
\end{align}

It is clear that $\yT$ has a single \yfail\ state, and has $n+1$ states.
\begin{description}
\item[\sf Claim 1.]  $\yT$ is deterministic.

\textsf{Proof.} 
By construction $\yT_{in}=S_{in}$.
Since $\yB$ is deterministic we get $\vert S_{in}\vert\leq 1$.
Thus, $\yT$ has at most one initial state. 

For the sake of contradiction, assume that $\yT$ does not satisfy Definition~\ref{def:vpts-determinism}.
Then, we must have $\ysi\in L^\star$, $p, q\in S_\yT$ and $\yal,\ybe\in \yDe_\yT$ such that 
\begin{align}
\ytrtf{(t_0,\bot)}{\ysi}{\yT}{(p,\yal\bot)}\quad\text{and}\quad\ytrtf{(t_0,\bot)}{\ysi}{\yT}{(q,\ybe\bot)}\label{lem5.4a}
\end{align}
with $p\neq q$ or $\yal\neq \ybe$, and $t_0$ initial in $\yT$.
If all transitions used in the runs in Eq. \ref{lem5.4a} are in $T_\yB$ we get an immediate contradiction to the determinism of $\yB$.

Assume now that not all transitions used in Eq. \ref{lem5.4a} are in $T_\yB$.
Since $\yfail$ is a sink state, we must have either $p=\yfail$ or $q=\yfail$.
Assume first $p=\yfail$ and $q\neq \yfail$.
Then, $\ysi=\ysi_1x$, with $x\in L$ and write the runs in Eq. \ref{lem5.4a} as
\begin{align}
\ytrtf{(t_0,\bot)}{\ysi_1}{\yT}{(p_1,\yal_1\bot)}\ytrtf{}{x}{\yT}{(\yfail,\yal\bot)}\quad\text{and}\quad\ytrtf{(t_0,\bot)}{\ysi_1}{\yT}{(q_1,\ybe_1\bot)}\ytrtf{}{x}{\yT}{(q,\ybe\bot)}\label{lem5.4b}
\end{align}
where $p_1\neq \yfail$, $q_1\neq \yfail$.
Thus, all transitions over $\ysi_1$ are in $\yB$, and since $\yB$ is deterministic, we get $p_1=q_1$ and $\yal_1=\ybe_1$.
Hence, we have transitions $(p_1,x,X,\yfail)$, $(p_1,x,Y,q)$ in $T_\yT$.
Now, we follow the construction of $\yT$ and distinguish three cases.
When $x\in L_U\cap L_c$, condition (\ref{eq:testpurpose1:Lc}) gives $X=Z$.
Since $q\neq \yfail$ we have $(p_1,x,Y,q)$ in $T_\yB$, contradicting 
condition (\ref{eq:testpurpose1:Lc}) of the construction.
When $x\in L_U\cap L_i$ the reasoning is the same, now using 
condition (\ref{eq:testpurpose1:Li}).
When $x\in L_U\cap L_r$, the first run in (\ref{lem5.4b}) gives $\yal_1=X\yal$, and the second run gives $\ybe_1=Y\ybe$.
Since $\yal_1=\ybe_1$ we get $X=Y$.
Now we have $(p_1,x,X,\yfail)$, $(p_1,x,X,q)$ in $T_\yT$.
Again, $q\neq \yfail$ implies $(p_1,x,X,q)$ in $T_\yB$, contradicting
condition  (\ref{eq:testpurpose1:Lr}).
The case $p\neq\yfail$, $q=\yfail$ is similar and also leads to a contradiction.
This shows that we must have $p=\yfail=q$ in (\ref{lem5.4a}).
So, we can write
\begin{align}
\ytrtf{(t_0,\bot)&}{\ysi_1}{\yT}{(p_1,\yal_1\bot)}\ytrtf{}{x}{\yT}{(\yfail,\yal\bot)}\label{lem5.4c}\\
\ytrtf{(t_0,\bot)&}{\ysi_1}{\yT}{(q_1,\ybe_1\bot)}\ytrtf{}{x}{\yT}{(\yfail,\ybe\bot)}\label{lem5.4d}
\end{align}
where $p_1\neq \yfail$, $q_1\neq \yfail$.
We recall that we need either $p\neq q$ or $\yal\neq\ybe$ in (\ref{lem5.4a}).
We are then left with  $\yal\neq \ybe$ in (\ref{lem5.4c}) and (\ref{lem5.4d}).

The determinism of $\yB$ forces $p_1=q_1$ and $\yal_1=\ybe_1$ because both runs on $\ysi_1$ in (\ref{lem5.4c}) and (\ref{lem5.4d}) use only transitions in $T_\yB$.
We have transitions $(p_1,x,X,\yfail)$ in (\ref{lem5.4c}) and 
$(p_1,x,Y,\yfail)$ in (\ref{lem5.4d}).
We now argue that we always reach $\yal=\ybe$, a contradiction.
If $x\in L_U\cap L_c$, condition (\ref{eq:testpurpose1:Lc}) says that $X=Z=Y$, and then $\yal=Z \yal_1$ and $\ybe=Z\ybe_1$, and we get $\yal=\ybe$.
If $x\in L_U\cap L_i$, condition (\ref{eq:testpurpose1:Li}) immediately gives  $\yal=\yal_1=\ybe_1=\ybe$.
When $x\in L_U\cap L_r$ we need $\yal_1=X\yal$ and $\ybe_1=Y\ybe$, so that $\yal=\ybe$ again.
We conclude that $\yT$ is, indeed,
deterministic.
\end{description}

The next claim shows that any $\ysi\in\ycomp{otr}(\yB)\cap\big[otr(\yB)L_U\big]$ leads $\yT$ to the $\yfail$ state.
\begin{description}
\item[\sf Claim 2.] Let  $\ysi\in\ycomp{otr}(\yB)\cap\big[otr(\yB)L_U\big]$.
Then, $\ytrut{(s_0,\bot)}{\ysi}{\yT}{(\yfail,\yal\bot)}$ for some $s_0\in T_{in}$, $\yal\in (\yDe_\yT)^\star$.

\textsf{Proof.} 
Let  $\ysi=\mu\ell$, $\ysi\not\in otr(\yB)$, $\ell\in L_U$ and $\mu\in otr(\yB)$.
Because $otr(\yB)=tr(\yB)$, we get $\ytrtf{(q_0,\bot)}{\mu}{\yB}{(p,\yal\bot)}$, where $q_0\in S_{in}=T_{in}$ and $\yal\in(\yDe_\yB)^\star$.
By construction, all transitions in $\yB$ are also transitions of $\yT$, so that  $\ytrtf{(q_0,\bot)}{\mu}{\yT}{(p,\yal\bot)}$.
Now we argue that $\ytrtf{(p,\yal\bot)}{\ell}{\yT}{(\yfail,\ybe\bot)}$ for some $\ybe\in(\yDe_\yT)^\star$, so that composing we get  $\ytrtf{(q_0,\bot)}{\mu\ell}{\yT}{(\yfail,\ybe\bot)}$.
We note that we cannot have $\ytrtf{(p,\yal\bot)}{\ell}{\yB}{(z,\yga\bot)}$ for any $z\in S_\yB$, $\yga\in (\yDe_\yB)^\star$, because then we would get 
$\ytrtf{(q_0,\bot)}{\mu\ell}{\yB}{(z,\yga\bot)}$, and then $\ysi=\mu\ell\in tr(\yB)=otr(\yB)$, a contradiction.
There are three simple cases.
If $\ell\in L_U\cap L_c$, then $(p,\ell,W,z)\not\in T_\yB$ for any $z\in S_B$ and any $W\in\yDe_\yB$.
Then, Eq. (\ref{eq:testpurpose1:Lc}) gives $(p,\ell,Z,\yfail)\in T_\yT$, as needed.
If $\ell\in L_U\cap L_i$, the reasoning is the same,  using  Eq. (\ref{eq:testpurpose1:Li}).
Now let $\ell\in L_U\cap L_r$.
Since $\ytrtf{(p,\yal\bot)}{\ell}{\yB}{(z,\yga\bot)}$ is not allowed, 
we cannot have  $(p,\ell,W,z)$ in $T_\yB$ for any $z\in S_\yB$, where $W\in \yDe_\yB$ is the first symbol in $\yal\bot$.
Now,  Eq. (\ref{eq:testpurpose1:Lr}) gives $(p,\ell,W,\yfail)\in T_\yT$.
Since $\ytrtf{(q_0,\bot)}{\mu\ell}{\yT}{(\yfail,\ybe\bot)}$ then, from Proposition~\ref{prop:vpts-deterministic} we also have $\ytrut{(q_0,\bot)}{\mu\ell}{\yT}{(\yfail,\ybe\bot)}$.
\end{description}
The next claim deals with the converse. 
\begin{description}
\item[\sf Claim 3.]  Let $\mu\in L^\star$, $\ell\in L$ and  $\ytrut{(t_0,\bot)}{\mu}{\yT}{(p,\yal\bot)}\ytrut{}{\ell}{\yT}{(\yfail,\ybe\bot)}$ where $t_0\in T_{in}$, $p\in S_\yT$, $\yal$, $\ybe\in (\yDe_\yT)^\star$.
Then we must have $\mu\ell\in\ycomp{otr}(\yB)\cap\big[otr(\yB)L_U\big]$.

\textsf{Proof.} 
By construction $\yT$ has no $\ytau$-moves, since $\yB$ has no $\ytau$-moves.
Thus, we must  have $\ytrtf{(t_0,\bot)}{\mu}{\yT}{(p,\yal\bot)}\ytrtf{}{\ell}{\yT}{(\yfail,\ybe\bot)}$.

Since  $\yfail$ is a sink state we get $p\neq \yfail$, and we know that all transitions in  $\ytrtf{(t_0,\bot)}{\mu}{\yT}{(p,\yal\bot)}$ are in $\yB$, so that 
$\ytrtf{(t_0,\bot)}{\mu}{\yB}{(p,\yal\bot)}$.
We must also have a transition $(p,\ell,X,\yfail)$ in $\yT$.
By the construction, it can only be inserted in $T_\yT$ by force of Eqs. (\ref{eq:testpurpose1:Lc}),(\ref{eq:testpurpose1:Lr}), or (\ref{eq:testpurpose1:Lr}).
In any case, we get $\ell\in L_U$.
Hence $\mu\ell\in otr(\yB)L_U$.
We now argue that $\mu\ell\not \in otr(\yB)$. 
This will give $\mu\ell\in \ycomp{otr}(\yB)\cap otr(\yB)L_U$, completing the proof. 

For the sake of contradiction, assume $\mu\ell\in otr(\yB)$, so that 
$\ytrtf{(t_0,\bot)}{\mu}{\yB}({p',\yal'\bot)}\ytrtf{}{\ell}{\yB}({r,\yga\bot)}$, with $t_0\in S_{in}$, $p',r\in S_\yB$ and $\yga\in (\yDe_\yB)^\star$.
Recall that $\yB$ is deterministic and has no $\ytau$-move. 
Hence we can write $\ytrut{(t_0,\bot)}{\mu}{\yB}{(p,\yal\bot)}$ and $\ytrut{(t_0,\bot)}{\mu}{\yB}({p',\yal'\bot)}$. 
Now from Definition~\ref{def:vpts-determinism} we get $p=p'$ and  $\yal=\yal'$.
Also recall that  $(p,\ell,X,\yfail)$ is a transition of $\yT$.
Together with $\ytrtf{(p,\yal\bot)}{\ell}{\yB}{(r,\yga\bot)}$,
there are three cases.

If $\ell \in L_c$ then $(p,\ell,W,r)$ is a transition of $\yB$ for some $W\in \yDe_\yB$.
In this case we must have used Eq. (\ref{eq:testpurpose1:Lc}) to insert $(p,\ell,X,\yfail)$ in $T_\yT$.
But then we need $X=Z$ and $(p,\ell,Y,q)\not\in T_\yB$ for any  $q\in S_\yB$,  $Y\in\yDe_\yB$, and we get a contradiction.
If $\ell \in L_i$ then $W=\yait$, $(p,\ell,\yait,r)$ is a transition of $\yB$ and we must have used Eq. (\ref{eq:testpurpose1:Li}) to insert $(p,\ell,\yait,\yfail)$ in $T_\yT$.
But this requires $(p,\ell,\yait,q)\not\in T_\yB$ for any $q\in S_\yB$, and we reach a contradiction again.
If $\ell\in L_r$ then $(p,\ell,W,r)$ is a transition of $\yB$ where $W$ is the first symbol in $\yal\bot$.
But, according to used Eq. (\ref{eq:testpurpose1:Lr}) we now need 
$(p,\ell,W,q)\not\in T_\yB$ for any $q\in S_\yB$.
This final contradiction completes the proof.
\end{description}
Let $\yI=\yiovpts{S_\yI}{I_{in}}{L_I}{L_U}{\yDe_\yI}{T_\yI}$ be an arbitrary IUT.
As argued above, we need $\yI$ does not pass $\yT$ if and only if $otr(\yI)\cap T \neq\yemp$, where $T=\ycomp{otr}(\yB)\cap\big[otr(\yB)L_U\big]$.
According to Definition~\ref{def:passes}, $\yI$ does not pass $\yT$ if and only if for some $\ysi\in L^\star$ we get  $\ytrut{((t_0,q_0),\bot)}{\ysi}{\yT\times \yI}{((\yfail,q),\yal\bot)}$ 
where $((t_0,q_0),\bot)$ is an initial configuration of $\yT\times \yI$,
$q\in S_\yI$ and $\yal=(X_1,Y_1)\cdots(X_n,Y_n)\in (\yDe_\yT\times \yDe_\yI)^\star$, $n\geq 0$.
From the construction, it is clear that $\yfail$ is a sink state of $\yT$, and transitions into $\yfail$ in $\yT$ 
are over symbols in $L_U$. 
Hence, we can say that $\yI$ does not pass $\yT$ if and only if 
$$\ytrut{((t_0,q_0),\bot)}{\mu}{\yT\times \yI}{((p,r),\ybe\bot)}\ytrut{}{\ell}{\yT\times \yI}{((\yfail,q),\yal\bot)},$$
for some $\mu\in L^\star$, $\ell\in L_U$, $p\in S_\yT$, $p\neq\yfail$, $r\in S_\yI$, $\ybe =(W_1,Z_1)\cdots (W_m,Z_m)\in(\yDe_\yT\times \yDe_\yI)^\star$, $m\geq 0$. 

From Definition~\ref{def:cross}, we know that $\yT\times \yI$ is the IOVPTS associated to  $\yA_\yT\times \yA_\yI$, where  $\yA_\yT$ and $\yA_\yI$ are the VPAs associated with $\yT$ and $\yI$, respectively.
Using Proposition~\ref{prop:lang-vpa-vlpts} we can say that 
$\yI$ does not pass $\yT$ if and only if  $$\ypdatrtf{((t_0,q_0),\mu\ell,\bot)}{i}{\yA_\yT\times\yA_\yI}{((p,r),\ell,\ybe\bot)}\ypdatrtf{}{k}{\yA_\yT\times\yA_\yI}{((\yfail,q),\yeps,\yal\bot)},$$
where $i\geq\vert\mu\vert$, $k\geq 1$.
Using  Proposition~\ref{prop:product-behavior} we get $\yI$ does not pass $\yT$ if and only if 
\begin{align}
&\ypdatrtf{(t_0,\mu\ell\bot)}{\star}{\yA_\yT}{(p,\ell,W_1\cdots W_m\bot)}\ypdatrtf{}{\star}{\yA_\yT}{(\yfail,\yeps,X_1\cdots X_n\bot)}\label{lem5.4e}\\
&\ypdatrtf{(q_0,\mu\ell\bot)}{\star}{\yA_\yI}{(r,\ell,Z_1\cdots Z_m\bot)}\ypdatrtf{}{\star}{\yA_\yI}{(q,\yeps,Y_1\cdots Y_n\bot)}\label{lem5.4f}.
\end{align}

Assume that $\yI$ does not pass $\yT$. Then, from Eq. (\ref{lem5.4f})
and using Proposition~\ref{prop:plts-pda} we obtain $\mu\ell\in L(\yA_\yI)=otr(\yI)$. 
From Eq. (\ref{lem5.4e}) and Proposition~\ref{prop:lang-vpa-vlpts}
we get 
$$\ytrut{(t_0,\bot)}{\mu}{\yT}{(p,W_1\cdots W_m\bot)}\ytrut{}{\ell}{\yT}{(\yfail,X_1\cdots X_n\bot)}.$$
Using Claim 3 and recalling that $p\neq \yfail$, we get $\mu\ell\in T= \ycomp{otr}(\yB)\cap\big[otr(\yB)L_U\big]$.
Thus $otr(\yI)\cap T\neq \yemp$.

Now assume $otr(\yI)\cap T\neq \yemp$, so that we have $\mu\ell\in otr(\yI)$, 
$\mu\ell\in otr(\yB)L_U$ and $\mu\ell\not\in otr(\yB)$, where $\mu\in L^\star$ and $\ell\in L_U$.
From  Proposition~\ref{prop:lang-vpa-vlpts} we get $\ypdatrtf{(q_0,\mu,\bot)}{\star}{\yA_\yI}{(r,\ell,Z_1\cdots Z_k\bot)}\ypdatrtf{}{\star}{\yA_\yI}{(q,\yeps,Y_1\cdots Y_n\bot)}$, where $q_0\in I_{in}$, $r,q\in S_\yI$, $Z_1\cdots Z_k, Y_1\cdots Y_n\in (\yDe_\yI)^\star$, and $k,n\geq 0$.
Using Claim 2 and Proposition~\ref{prop:lang-vpa-vlpts} we also get a run over $\mu\ell$ in $\yA_\yT$ as  $\ypdatrtf{(t_0,\mu,\bot)}{\star}{\yA_\yT}{(p,\ell,W_1\cdots W_j\bot)}\ypdatrtf{}{\star}{\yA_\yT}{(\yfail,\yeps,X_1\cdots X_m\bot)}$, with $t_0\in T_{in}$, $p\in S_\yT$, $W_1\cdots W_j,X_1\cdots X_m\in (\yDe_\yT)^\star$, $j,m\geq 0$.
Use Proposition~\ref{prop:stack-size} to conclude that $k=j$ and $n=m$.
We now have 
$$\ypdatrtf{(q_0,\mu\ell,\bot)}{\star}{\yA_\yI}{(q,\yeps,Y_1\cdots Y_n\bot)},\quad\ypdatrtf{(t_0,\mu\ell,\bot)}{\star}{\yA_\yT}{(\yfail,\yeps,X_1\cdots X_n\bot)}.$$
From Proposition~\ref{prop:product-behavior} it follows that
$\ypdatrtf{((t_0,q_0),\mu\ell,\bot)}{\star}{\yA_\yT\times\yA_\yI}{((\yfail,q),\yeps,\yal\bot)}$, and using Proposition~\ref{prop:lang-vpa-vlpts} we now have
$\ytrut{((t_0,q_0),\bot)}{\mu\ell}{\yT\times \yI}{((\yfail,q),\yal\bot)}$.
This shows that $\yI$ does not pass $\yT$

Now we have that $\yI$ does not pass $\yT$ if and only if $otr(\yI)\cap T\neq \yemp$, as needed.
\end{proof}

\subsection{Proof of Theorem~\ref{teo:alg1}}\label{app:teo:alg1}
\begin{proof}
A VTPS $\yP=\yvpts{Q}{Q_{in}}{L}{\yGa}{\rho}$ with no transitions of the form $(p,a,\bot,q)$ in $\rho$, and two states $s_i$, $s_e$ with $s_i\neq s_e$, are input to Algorithm~\ref{alg1}.

At lines 1 and 2 we start with $V=null$ and $R[p,q]=0$ for all pairs $(p,q)$. 
Inspecting lines 6, 10, 14, 16 and 19, we see that a pair $(p,q)$ is added to $V$ only when we currently have $R[p,q]=0$ and, upon entering $V$, we immediately set $R[p,q]\neq 0$.
Further, at no other point in the code we reset $R[p,q]$ to zero again.
Hence, a pair $(p,q)$ can enter $V$ at most once and, therefore, the main loop at line 11 must terminate. 
So,  Algorithm~\ref{alg1} always stops.

Next we claim that, at any point during the execution of the algorithm, if $R[p,q]\neq 0$ then it codes for a string that induces a balanced run from $p$ to $q$.
This is immediate from the initialization lines 6 and 10. 
Proceeding inductively, assume that this property holds after a number of executions of the main loop. 
At line 14, we have removed $(p,q)$ from $V$, and so we now have  $R[p,q]\neq 0$ because $(p,q)$ entered $V$ in a previous iteration.
At that moment we made $R[p,q]\neq 0$ and then, inductively, it codes for a string $\ysi$ such that $\ytrtf{(p,\bot)}{\ysi}{}{(q,\bot)}$.
Now, at line 14 we require $R[s,p]\neq 0$ so that, inductively, it also codes for a string $\mu$ such that     $\ytrtf{(s,\bot)}{\mu}{}{(p,\bot)}$.
Composing, we get $\ytrtf{(s,\bot)}{\mu\ysi}{}{(q,\bot)}$, and so, we extended the induction in this case when we now make $R[s,q]\neq 0$ coding for the string $\mu\ysi$.
The reasoning at line 16, is very similar.
We now look at line 19.
At that point we have a push transition $(s,a,Z,p)$, a pop transition
$(q,b,Z,t)$, and $\ytrtf{(p,\bot)}{\ysi}{}{(q,\bot)}$ for some $\ysi\in L_\ytau^\star$.
Recall that the algorithm assumes that the given VPTS $\yP$ has no pop transitions on the empty stack.
With this hypothesis, we claim the following general property of $\yP$:
\begin{quote}
If $\ytrtf{(p,\yal_1\bot)}{\ysi}{}{(q,\yal_2\bot)}$ for some $p$ and $q\in Q$ and $\yal_1,\yal_2\in\yGa^\star$, then for all $\ybe_1,\ybe_2\in\yGa^\star$ we also have $\ytrtf{(p,\yal_1\ybe_1\bot)}{\ysi}{}{(q,\yal_2\ybe_2\bot)}$. 

A simple proof can be obtained by induction on $\vert\ysi\vert\geq 0$.
\end{quote}
With $\yal_1=\yal_2=\yeps$, $\ybe_1=\ybe_2=Z$, from  $\ytrtf{(p,\bot)}{\ysi}{}{(q,\bot)}$ we get $\ytrtf{(p,Z\bot)}{\ysi}{}{(q,Z\bot)}$.
Now we have 
$\ytrtf{(s,\bot)}{a}{}{(p,Z\bot)}\ytrtf{}{\ysi}{}{(q,Z\bot)}\ytrtf{}{b}{}{(t,\bot)}$, so that making $R[s,t]$ code for the string $a\ysi b$ also extends the induction after line 19 is passed.
Since we have completed one more iteration of the main loop, we see that upon termination of the main loop, if we have $R[s_i,s_e]\neq 0$, then we do have a string  that induces a balanced run from $s_i$ to $s_e$.
Moreover, it easy to see that the simple recursive call \texttt{getstring($s_i,s_e$)} at line 21 does correctly extract one such string.

Next, we argue in the other direction.
Suppose that the main loop terminates with $V=null$. Then we claim that for all pairs $(p,q)$, with $p\neq q$, if $R[p,q]=0$ then there is no string capable of inducing a balanced run from $p$ to $q$.
For the sake of contradiction, assume  that the main loop terminates with $V=null$, and we have $p\neq q$, $R[p,q]=0$, and a string $\ysi$ such that $\ytrtf{(p,\bot,)}{\ysi}{}{(q,\bot)}$.
Among all such pairs, choose one for which $\vert\ysi\vert$ is minimum.
Since $p\neq q$, we need $\vert\ysi\vert\geq 1$.
If $\vert\ysi\vert=1$, then we need a transition $(p,\ysi,\yait,q)$ in $\rho$.
But then, at line 6, we make $R[p,q]=[p,\ysi,q]$ and it is never rest to $0$ again. 
This is a contradiction, and we can assume $\vert\ysi\vert\geq 2$.

Now there are two cases, depending on the internal part of the run $\ytrtf{(p,\bot)}{\ysi}{}{(q,\bot)}$.
\begin{description}
\item[\sf A configuration $(r,\bot)$ occurs in the run $\ytrtf{(p,\bot)}{\ysi}{}{(q,\bot)}$.] Write
$\ytrtf{(p,\bot)}{\ysi_1}{}{(r,\bot)}\ytrtf{}{\ysi_2}{}{(q,\bot)}$, with $\ysi=\ysi_1\ysi_2$.
We know that $\ysi_1\neq\yeps\neq \ysi_2$ because $\yP$ has no transitions on the empty stack.
If $p=r$ we get $\ytrtf{(p,\bot)}{\ysi_2}{}{(q,\bot)}$.
Since $\vert\ysi_2\vert<\vert\ysi\vert$, the minimality of  $\vert\ysi\vert$ forces $R[p,q]\neq 0$, a contradiction.
Similarly, $q=r$ also leads to a contradiction.

Now, assume $p\neq r\neq q$.
Since $\vert\ysi_1\vert<\vert\ysi\vert$ and $\vert\ysi_2\vert<\vert\ysi\vert$, when the main loop terminates with $V=null$ we must have $R[p,r]\neq 0$ and $R[r,q]\neq 0$.
Moreover, for this to happen, both $(p,r)$ and $(r,q)$ were added to $V$.
Suppose that $(p,r)$ is removed from $V$ before $(r,q)$.
Then, at the iteration of the main loop when $(r,q)$ is removed from $V$ we have $R[p,r]\neq 0$ and $p\neq q$.
Hence, at line 14, since $R[p,q]=0$ we make $R[p,q]=[p,r,q]$ and, since it is never reset to $0$ again, we have a contradiction.
If $(r,q)$ is removed first from $V$, the reasoning is the same using line 16.
So, this case can not happen.

\item[\sf A configuration $(r,\bot)$ does not occur in the run $\ytrtf{(p,\bot)}{\ysi}{}{(q,\bot)}$.]  
Then, we must have a push transition $(p,x,Z,s)$ with $\ysi=x\ysi_1$, and
we are left with $\ytrtf{(p,\bot)}{x}{}{(s,Z\bot)}\ytrtf{}{\ysi_1}{}{(q,\bot)}$.
If $\vert\ysi_1\vert=1$, we need a pop transition $(s,y,Z,q)$ and now line 10 makes $R[p,q]=[x,q,q,y]$, a contradiction.
Hence $\vert\ysi_1\vert\geq 2$.
Since no configuration of the form $(u,\bot)$ occurs in the run over $\ysi$, we must have a pop transition $(t,y,Z,q)$ and
$\ytrtf{(p,\bot)}{x}{}{(s,Z\bot)}\ytrtf{}{\mu}{}{(t,Z\bot)}\ytrtf{}{y}{}{(q,\bot)}$,
with $\ysi=x\mu y$.

Next we claim that in any VPTS $\yS=\yvptsS$ if a run does not shorten initial stack, then that string can be substituted for any other.
More precisely,
\begin{quote}
Let $p$, $q\in S$, $\ysi\in L_\ytau^\star$ and $\yal\in\yGa^\star$ with $\ytrtf{(p,\yal\bot)}{\ysi}{}{(q,\yal\bot)}$.
Assume that a configuration $(u,\yga)$, with $\vert\yga\vert< \vert\yal\vert$, does not occur in that run over $\ysi$.
Then, for any $\ybe\in \yGa^\star$ we also have  $\ytrtf{(p,\ybe\bot)}{\ysi}{}{(q,\ybe\bot)}$.

An easy induction over $\vert\ysi\vert\geq 0$ gives the result.
\end{quote}
Recall that we already have $\ytrtf{(s,Z\bot)}{\mu}{}{(t,Z\bot)}$ and a configuration $(u,\bot)$ dos not occur on the run over $\mu$ since it can not occur on a run over $\ysi$.
Using the claim we get $\ytrtf{(s,\bot)}{\mu}{}{(t,\bot)}$. 
Now, since $\vert\mu\vert<\vert\ysi\vert$, the minimality of $\vert\ysi\vert$ says that when the main loop terminates with $V=null$, we must have $R[s,t]\neq 0$.
But then, at some moment $(s,t)$ was added to $V$.
Since the main loop terminates with $V=null$, at some iteration we have removed $(s,t)$ from $V$.
Note that we have a push transition $(p,x,Z,s)$ and a pop transition
$(t,y,Z,q)$. 
Hence, line 19 says that we will set $R[p,q]=[x,s,t,y]$ and, since $R[p,q]$ is never reset, we see that the main loop terminates with $R[p,q]\neq 0$, which is a contradiction.
\end{description}
Since both cases lead to contradictions, we conclude that when the main loop terminates with $V=null$ and $R[s_i,s_e]=0$, then there is no string capable of inducing a balanced run from $s_i$ to $s_e$.
Hence, line 20 correctly reports the inexistence of any such strings.

Thus, in any case, lines 20--21 always report as expected, and Algorithm~\ref{alg1} is correct.
\end{proof}